\documentclass[sigconf,balance=false]{acmart}
\usepackage{popets}

\setcopyright{popets}
\copyrightyear{YYYY}

\acmYear{YYYY}
\acmVolume{YYYY}
\acmNumber{X}
\acmDOI{XXXXXXX.XXXXXXX}
\acmISBN{}
\acmConference{Proceedings on Privacy Enhancing Technologies}
\usepackage{tikz}
\usepackage{amsmath}
\usepackage{wrapfig}
\usepackage{caption}
\usepackage{subcaption}

\usepackage{mdwmath}
\usepackage{multirow}
\usepackage{mdwtab}
\usepackage{bbm}

\usepackage{dblfloatfix}

\usepackage{url}

\usepackage{xcolor}

\usepackage{hyperref}

\usepackage{algorithm,algpseudocode}

\algnewcommand{\Inputs}[1]{%
  \State \textbf{Inputs:}
  \Statex \hspace*{\algorithmicindent}\parbox[t]{\linewidth}{\raggedright #1}
}
\algnewcommand{\Output}[1]{%
  \State \textbf{Output:}
  \Statex \hspace*{\algorithmicindent}\parbox[t]{\linewidth}{\raggedright #1}
}
\algnewcommand{\Initialize}[1]{%
  \State \textbf{Initialize:}
  \Statex \hspace*{\algorithmicindent}\parbox[t]{\linewidth}{\raggedright #1}
}
\algnewcommand{\comm}[1]{ {\ttfamily\textcolor{blue}{// #1}} }

\usepackage{relsize}

\usepackage{amsthm}
\newtheorem{theorem}{Theorem}[section]

\newtheorem{lemma}[theorem]{Lemma}

\newtheorem{definition}{Definition}

\newcommand{\R}{\mathbb{R}}

\newcommand{\Ni}{\mathbb{N}}

\newcommand{\Si}{\mathcal{S}}
\newcommand{\T}{\mathcal{T}}
\newcommand{\Users}{\mathcal{U}}
\newcommand{\Pop}{\Omega}

\newcommand{\E} {\mathbb{E}}

\newcommand{\underscore}{\rule{0.5em}{0.3pt}}

\begin{document}

\title[A Zero Auxiliary Knowledge Membership Inference Attack on Aggregate Location Data]{A Zero Auxiliary Knowledge Membership Inference Attack on Aggregate Location Data}

\author{Vincent Guan*}
\orcid{0009-0005-5218-9233}
\affiliation{%
  \institution{Imperial College London}
  \city{}
  \state{}
  \country{}}
\email{vincentguan23@gmail.com}

\author{Florent Gu\'epin*}
\orcid{0009-0008-5098-0963}
\affiliation{%
  \institution{Imperial College London}
  \city{}
  \state{}
  \country{}}
\email{florent.guepin@imperial.ac.uk}

\author{Ana-Maria Cretu}
\orcid{0000-0002-9009-7381}
\affiliation{%
\institution{EPFL}
  \city{}
  \state{}
  \country{}
}
\email{ana-maria.cretu@epfl.ch}

\author{Yves-Alexandre de Montjoye}
\orcid{0000-0002-2559-5616}
\affiliation{%
 \institution{Imperial College London}
 \city{}
 \state{}
 \country{}
 }
\email{demontjoye@imperial.ac.uk}

\renewcommand{\shortauthors}{V. Guan, F. Guépin et al.}

\begin{abstract}
Location data is frequently collected from populations and shared in aggregate form to guide policy and decision making. However, the prevalence of aggregated data also raises the privacy concern of membership inference attacks (MIAs). MIAs infer whether an individual's data contributed to the aggregate release. Although effective MIAs have been developed for aggregate location data, these require access to an extensive auxiliary dataset of individual traces over the same locations, which are collected from a similar population. This assumption is often impractical given common privacy practices surrounding location data. To measure the risk of an MIA performed by a realistic adversary, we develop the first Zero Auxiliary Knowledge (ZK) MIA on aggregate location data, which eliminates the need for an auxiliary dataset of real individual traces. Instead, we develop a novel synthetic approach, such that suitable synthetic traces are generated from the released aggregate. We also develop methods to correct for bias and noise, to show that our synthetic-based attack is still applicable when privacy mechanisms are applied prior to release. Using two large-scale location datasets, we demonstrate that our ZK MIA matches the state-of-the-art Knock-Knock (KK) MIA across a wide range of settings, including popular implementations of differential privacy (DP) and suppression of small counts. Furthermore, we show that ZK MIA remains highly effective even when the adversary only knows a small fraction ($10\%$) of their target's location history. This demonstrates that effective MIAs can be performed by realistic adversaries, highlighting the need for strong DP protection.

\end{abstract}


\maketitle
\def\thefootnote{*}\footnotetext[1]{These authors contributed equally to this work.}

\section{Introduction}

\label{sec:introduction}
Human mobility and location data are widely used across many important domains, such as epidemiology~\cite{hara2021japanese, grantz2020use}, humanitarian response~\cite{yabe2022mobile}, and finance~\cite{holmes2013mobile}, as they offer insights into movement and density patterns. However, many people are concerned about the extensive collection of personal location data~\cite{van2016privacy, zhou2017understanding, hope2021millions}, particularly since this data may provide information regarding a person's social, economic, and political life~\cite{georgiadou2019location}.

Individual-level location datasets have been shown to be highly vulnerable to re-identification attacks, due to the unicity and temporal consistency of people's mobility patterns~\cite{zang2011anonymization,de2013unique,tournier2022expanding}. To address these privacy concerns, data practitioners commonly use aggregate statistics, instead of individual-level records~\cite{aktay2020google, popa2011privacy,xu2015understanding}. For example, the Public Health Agency of Canada studied citizens' movement during the COVID-19 pandemic, using aggregate location data from millions of mobile devices, provided by TELUS~\cite{OPC2023Investigation,oli2021canada}. British researchers conducted similar COVID-19 mobility analysis~\cite{jeffrey2020anonymised,trasberg2023spatial}, using aggregate location data obtained from O2 and Facebook. Because aggregate location data is often considered to be sufficiently de-identified~\cite{OPC2023Investigation}, it is commonly sold by data brokers to interested parties~\cite{NYT2021DIASurveillance, boorstein2023colorado}. Notably, the U.S. government has been criticized for using commercial aggregate location data for law enforcement purposes and military intelligence~\cite{NYT2021DIASurveillance}. Aggregate location data is also used in other sectors, such as urban design, to optimize public transit networks~\cite{morgan2021travel, Kakakhel2022Optimising, o2_transport_smart_steps_2019}, and finance, to understand consumer behaviour~\cite{safegraph_spend, precisely_placeiq_movement}.

\begin{figure}
    \centering
    \includegraphics[width = \linewidth, page=1]{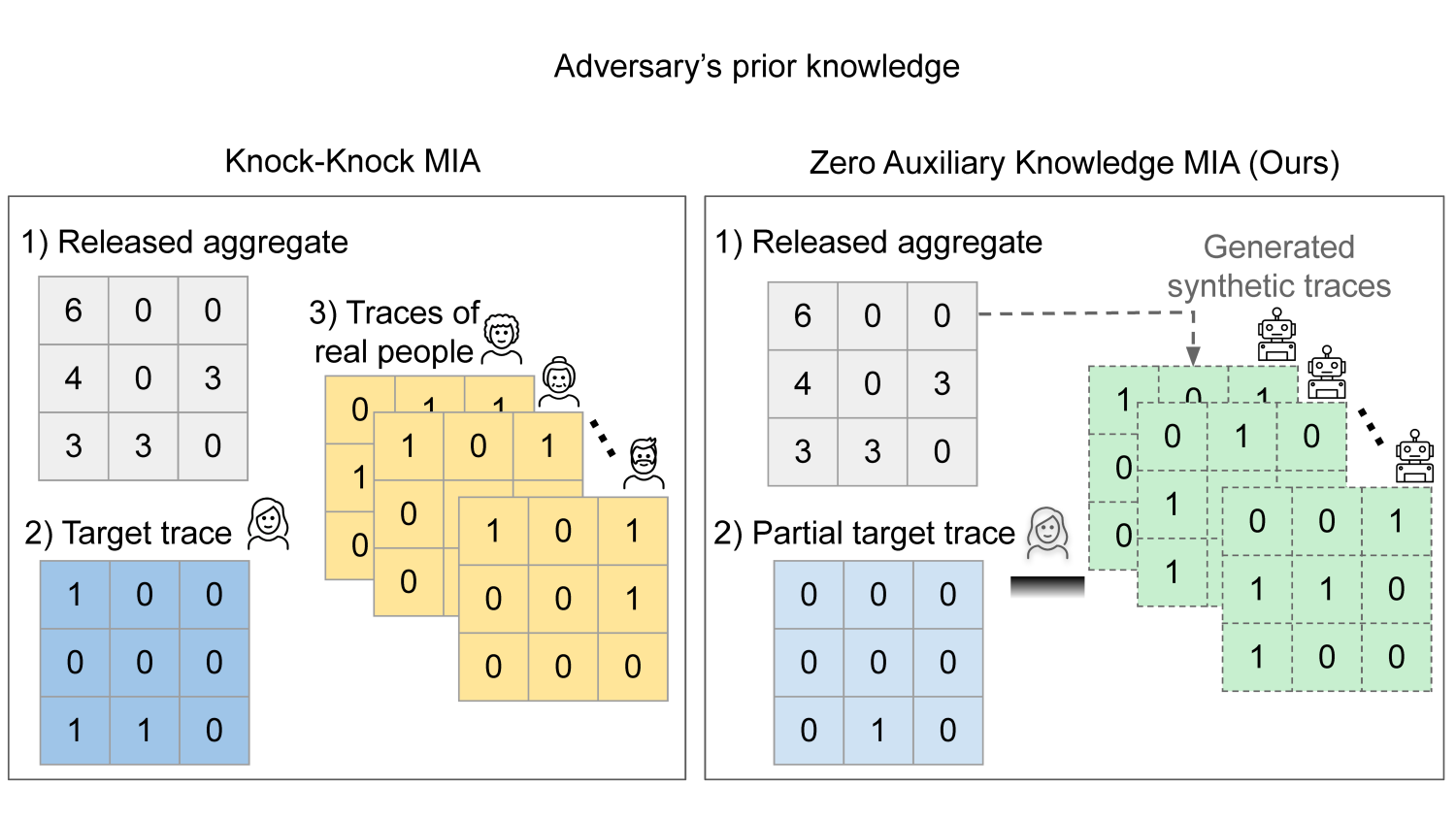}
    \caption{Adversary's prior knowledge in the previous work, Knock-Knock MIA~\cite{pyrgelis2017knock} (left), and our work, Zero Auxiliary Knowledge MIA (right). The ZK adversary does not require knowledge of location traces of real people to run the MIA.}
    \label{fig:zk_kk_priors}
\end{figure}

\textbf{Motivation. } As outlined in the E.U. Article 29 Working Party's guidance on anonymization techniques, aggregation reduces the risk of re-identification but does not eliminate all privacy risks~\cite{wp2014}. In particular, aggregates may still be vulnerable to membership inference attacks (MIAs), whose goal is to infer if an individual's data was included in the data release, e.g. aggregate data. MIAs have become the de facto standard in privacy auditing due to their practical threat model and theoretical properties. From a practical perspective, a successful MIA is a direct privacy violation whenever participation in the data release is sensitive~\cite{li2013membership}. Furthermore,  MIAs can be used as building blocks for other attacks, by first inferring a user's participation and then inferring their sensitive attributes. From a theoretical perspective, the success rate of an MIA is upper bounded following the application of differential privacy (DP)~\cite{dwork2006calibrating, yeom2018privacy,humphries2023investigating}. Hence, MIAs can be used as an auditing tool for DP implementations~\cite{jagielski2020auditing,nasr2021adversary}. Today, MIAs are widely used to assess the privacy risk of a broad range of data releases, including aggregate genetic data~\cite{homer2008resolving,sankararaman2009genomic}, aggregate survey data~\cite{bauer2020towards}, aggregate location data~\cite{pyrgelis2017knock,oehmichen2019opal}, machine learning models ~\cite{shokri2017membership,jayaraman2019evaluating,nasr2021adversary} and synthetic data releases~\cite{stadler2022synthetic,houssiau2022tapas,meeus2023achilles,guepin2023synthetic}.

MIAs pose an especially strong privacy threat on aggregate location data, since location data is often processed alongside sensitive attributes, such as socioeconomic status~\cite{trasberg2023spatial} and vaccination status~\cite{hope2021millions}. In a notable example, a high-ranking priest resigned after being outed as homosexual by a radical group that matched his smartphone data with location data from Grindr, a popular dating app among the LGBTQ+ community~\cite{boorstein2023colorado}. It is therefore important to understand the practical risk that MIAs pose on aggregate location data, particularly by a realistic adversary, who only possesses information about their target.

The first and most prominent MIA on aggregate location data was proposed by~\citet{pyrgelis2017knock}. Their ``Knock-Knock'' (KK) MIA works by training a binary classifier on a set of aggregates, wherein the adversary includes the target trace half of the time, and labels the aggregates accordingly. However, in addition to knowing the target trace, KK MIA requires the adversary to have access to a large auxiliary dataset of \textbf{individual-level traces} over the \textbf{same locations} and from a \textbf{similar population} as in the aggregate release. This is, when it comes to location data, a very strong assumption. This reliance on a strong adversary has led companies and practitioners to dismiss the risk posed by MIAs on location data. To the best of our knowledge, all previous works studying MIAs on aggregate location data require a similar auxiliary dataset~\cite{pyrgelis2020measuring,zhang2020locmia,oehmichen2019opal}. 

\textbf{Contributions. } 
To assess the realistic privacy risk of releasing aggregate location data, we introduce the Zero Auxiliary Knowledge (ZK) MIA. ZK MIA is the first MIA on aggregate location data that does not require the adversary to have access to an auxiliary dataset. To remove this strong assumption, we develop a novel synthetic data-based approach, in which the adversary generates a reference dataset of synthetic traces, using only statistical parameters estimated from the aggregate. Training aggregates are then created using the synthetic reference. To account for privacy mechanisms applied to the release, we develop techniques to correct the parameter estimation for bias and noise, which enables ZK MIA to effectively attack privacy-aware aggregates as well. We also demonstrate that a paired sampling technique further improves MIA performance by isolating the contribution of the target trace within the high-dimensional aggregate. In the setting of $\varepsilon$-DP aggregate location data, we show that paired sampling enables MIAs to approach the worst-case $\varepsilon$-DP bound, offering a significant increase in performance to previous implementations. 

We evaluate our Zero Auxiliary Knowledge (ZK) MIA against the state-of-the-art Knock-Knock (KK) MIA from~\citet{pyrgelis2017knock} using two location datasets: i) a large-scale call detail record (CDR) dataset, and ii) the Milan Twitter dataset~\cite{DVN/9IZALB_2015} from the Telecom Italia Big Data Challenge ~\cite{barlacchi2015multi}. We apply the MIAs on raw and privacy-aware aggregates computed over $1000$ users. Our results show that our ZK MIA closely matches the performance of KK MIA, without depending on extensive prior knowledge. On raw aggregates, both MIAs achieve $0.99$ AUC on both datasets, suggesting that aggregation in itself is an ineffective safeguard. Both MIAs also surpass $0.9$ AUC on both datasets under common privacy settings, including $\epsilon=1$ event-level DP noise addition.

We further relax assumptions and show that the adversary does not need the full target trace for ZK MIA to succeed. Indeed, ZK MIA still achieved $0.84$ AUC on the CDR dataset, with $\epsilon=1$ event-level DP in place, when the adversary only knew  a random 10\% of the target trace.

After extensive evaluations across different privacy mechanisms, namely the suppression of small counts ~\cite{chen2009privacy} and $\varepsilon$-DP noise addition ~\cite{dwork2006calibrating}, we argue that the commonly used $\epsilon$-DP implementations on aggregate location data~\cite{desfontaines2021list} do not protect against realistic privacy threats, such as our ZK MIA. We conclude that the only effective mitigation is the application of strong user level DP or user-day level DP guarantees, which is not yet a common practice ~\cite{telus_insights_2024, martinez2023netmob23, o2_transport_smart_steps_2019, safegraph_spend, precisely_placeiq_movement}.

\section{Definitions and Threat Model}
\label{sec:background}
We formally define location traces and aggregates in Section \ref{subsec:loc_trace_agg} and overview aggregate-level privacy measures in Section \ref{subsec:agg_privacy_measures}. In Section \ref{subsec:problem_formulation} and \ref{subsec:membership_classifier}, we outline the membership inference problem on aggregate location data and introduce the concept of a membership classifier. We present the threat model for our Zero Auxiliary Knowledge MIA and compare it against previous threat models for MIAs on location aggregates in Section \ref{subsec: threat_model}. Table \ref{tab:glossary_notations} of the Appendix contains a glossary of common terms.

\subsection{Location Traces and Aggregates}
\label{subsec:loc_trace_agg}

Let $\Si = \{s_1, ..., s_{|\Si|}\}$ represent the set of all regions of interest (ROIs) where location data is collected. Similarly, $\T = \{t_1, ..., t_{|\T|}\}$ denotes the set of time intervals, also known as epochs, during which data collection occurs. In this paper, we assume that the geographic positions (i.e. approximate longitude and latitude) of the ROIs are known. For example, $\Si$ may represent a set of square regions that partition a city into a grid, and $\T$ may represent contiguous hours over one month. 

We focus on the scenario where location data of a set of users $\Pop$ is collected over the ROIs $\Si$ and the epochs $\T$. We define the \textit{location trace} $L^u$ of a user $u \in \Pop$ as the set of geo-tagged and time-stamped visits $(s,t)$ that $u$ made within $\Si$ during $\T$. We formally represent a user's location trace as the binary matrix

\begin{equation}
    L^u_{s,t} =
\begin{cases}
    1 & \text{if user $u$ visited ROI $s$ during epoch $t$} \\
    0 & \text{otherwise}. 
\end{cases}
\label{def:location_trace}
\end{equation}

\begin{figure}[!h]
    \centering
    \includegraphics[width = \linewidth, page=3]{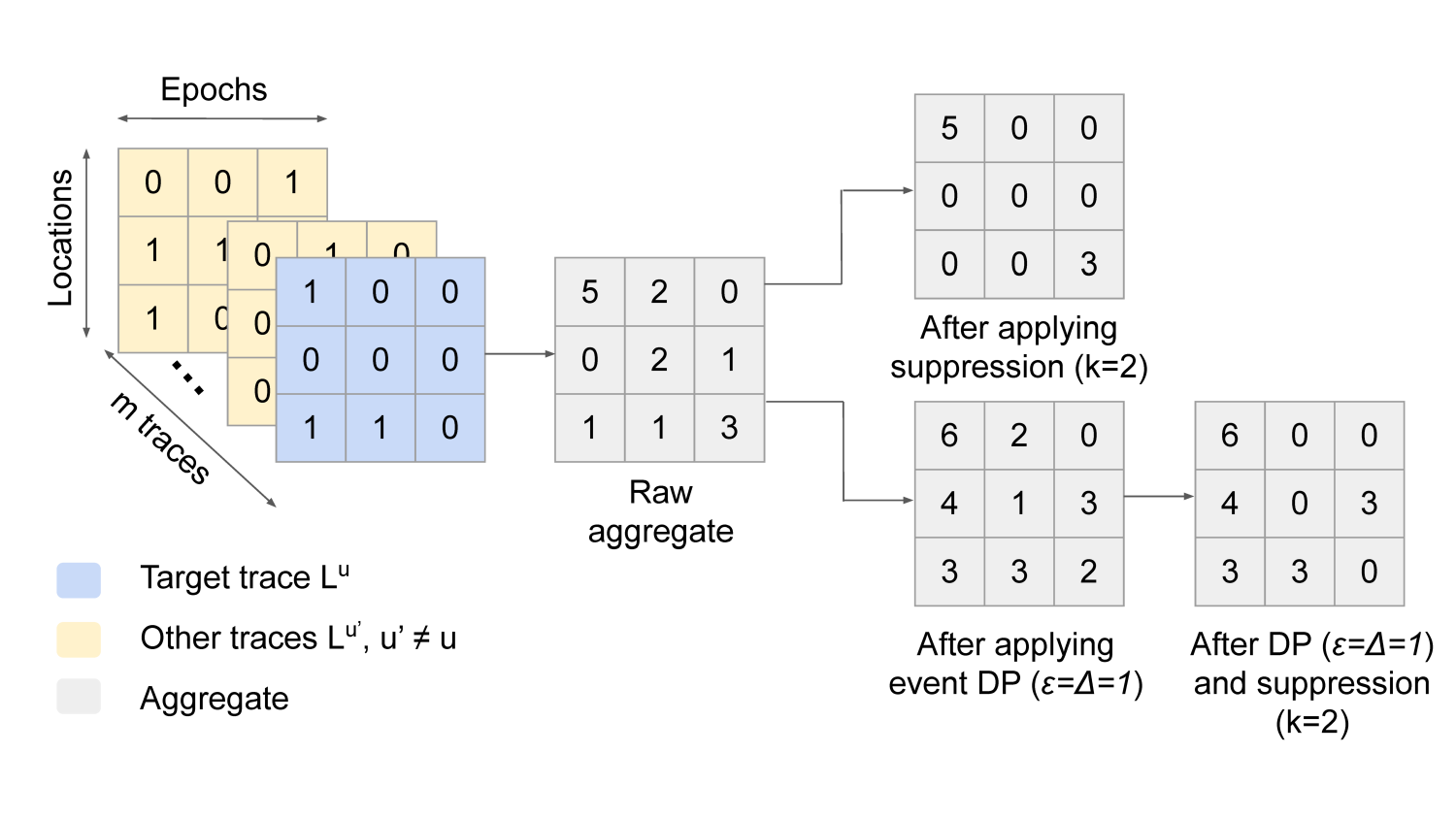}
    \caption{Example of how suppression of small counts and differential privacy may be applied to an aggregate with 3 ROIs (rows) and 3 epochs (columns).}
    \label{fig:privacy_measures_overview}
\end{figure}

Let $\Users \subset \Pop$ be a group of $m$ users whose location data is aggregated. We define an \textit{aggregate} $A^\Users$ to be the aggregate count statistics for $\Users$ over $\Si \times \T$. Formally, this is defined by the sum
\begin{equation}
    A^\Users = \sum\limits_{u \in \mathcal{U}}L^u.
\end{equation}
The entry $A^\Users_{s,t}$ therefore corresponds to the number of users in $\Users$ who visited ROI $s$ during epoch $t$.

\subsection{Privacy Measures on Location Aggregates}
\label{subsec:agg_privacy_measures}
The data collector may be wary of the privacy risks of releasing the raw aggregate $A^\Users$, and therefore apply privacy measures before releasing it.

\subsubsection{Differential Privacy}
\label{subsec:DP}
Differential privacy (DP)~\cite{dwork2006calibrating} is considered the gold standard for releasing information while protecting the privacy of individuals with formal guarantees. In essence, DP requires that the output of a computation over a dataset should not depend too much on the inclusion of any one record. 

\begin{definition}[$\varepsilon$-DP~\cite{dwork2006calibrating}]
    A randomized algorithm $M$ satisfies $\varepsilon$-DP if for all neighbouring datasets $D_1 \sim D_2$ (i.e.,
    differing in exactly one record), and all possible outputs $S \subset Range(M)$:
    \begin{equation}
        \Pr(M(D_1) \in S) \le e^{\varepsilon}\Pr(M(D_2) \in S) 
        \label{def:DP_equation}
    \end{equation}
\end{definition}

Thus, $\varepsilon$-DP limits the amount of information that can be inferred about individual records in the dataset, according to the privacy budget  $\varepsilon$~\cite{dwork2006calibrating}.  However, 
the privacy protection depends on what one considers as a ``record'', or \textit{privacy unit}, when defining the  neighbouring datasets. The most common definitions, in increasing level of privacy protection, are event-level, user-day level, and user-level DP (see ~\citet{desfontaines2021list} for an overview). The privacy unit for event-level DP is an individual data entry by any given user. For aggregate location data, this would be a single visit by a user to a ROI $s$ during an epoch $t$. The privacy unit for user-day level would be all visits registered by any given user over a day. Finally, for user-level, the unit would be all visits in any given user's trace. 

Randomised $\varepsilon$-DP mechanisms can be designed by adding noise sampled from the Laplace distribution~\cite{dwork2006calibrating}.
$A_{DP}^\Users(\varepsilon) = A^\Users + Lap(\frac{\Delta}{\varepsilon})$ would satisfy \eqref{def:DP_equation} and be an $\varepsilon$-DP aggregation mechanism, where $\Delta$ is the global sensitivity, determined by the privacy unit.

In this paper, we assume the common practice of post-processing to ensure legitimate aggregate counts \cite{ge2016updating, pyrgelis2017knock, pyrgelis2020measuring, zhu2022post}. Negative counts are set to $0$, counts exceeding the group size $m$ are set to $m$, and counts are rounded down to the nearest integer. These transformations will preserve $\varepsilon$-DP due to the post-processing theorem \cite{dwork2019differential}. We note that the adversary can always apply these transformations themselves if the data collector does not do so already. 

\subsubsection{Suppression of Small Counts (SSC)}
\label{subsec:bucket_suppression}

SSC is a privacy mechanism that aims to protect user privacy by hiding rare values. It has been frequently used across different types of datasets \cite{cretu2022querysnout, gadotti2019signal, pyrgelis2020measuring}, including mobility datasets \cite{kohli2023privacy,aktay2020google}. Instead of releasing the raw aggregate $A^\Users$, the data collector may choose a threshold $k \in \Ni \cup \{0\}$, and release the suppressed aggregate

\begin{equation}
  A_{SSC}^{\Users}(k)_{s,t} =
  \begin{cases}
                                    A^\Users_{s,t} & \text{ if } A^\Users_{s,t}>k \\
                                   0 & \text{ if }  A^\Users_{s,t} \le k.
  \end{cases}
  \label{def:suppressed_aggregates}
\end{equation}

$A_{SSC}^{\Users}(k)$ therefore contains the true count of users who visited a ROI $s$ during epoch $t$ as long as the count exceeds $k$. Lesser visited pairs $(s,t)$ that record $k$ or less visits are reported as $0$ instead. Suppression can also be applied following $\varepsilon$-DP noise addition, such that \eqref{def:suppressed_aggregates} is applied on a noisy aggregate ${A}_{DP}^{\Users}(\varepsilon)$. This produces ${A}_{DP, SSC}^{\Users}(\varepsilon,k)$, an $\varepsilon$-DP aggregate whose final counts have been suppressed with threshold $k$. This transformation would preserve $\varepsilon$-DP due to post-processing
\cite{dwork2019differential}, and may add a layer of complexity that mitigates attacks in practice. 

\subsection{Problem Formulation}
\label{subsec:problem_formulation}
 
We assume that the data collector releases aggregate count statistics $\overline{A}^{\Users}$ over the ROIs $\Si$ and the epochs $\T$, for the $m$  users in the group $\Users$ . There are various cases depending on the privacy measures applied prior to release: 
\[
\overline{A}^{\mathcal{U}} = 
\begin{cases} 
A^{\mathcal{U}}, & \text{if the raw aggregate counts are released,} \\
{A}_{DP}^{\mathcal{U}}(\varepsilon), & \text{if only $\varepsilon$-DP is applied,} \\
{A}_{SSC}^{\mathcal{U}}(k), & \text{if only threshold $k$ SSC is applied,} \\
{A}_{DP, SSC}^{\mathcal{U}}(\varepsilon, k), & \text{if threshold $k$ SSC is applied after $\varepsilon$-DP.}
\end{cases}
\]
The goal of an adversary $Adv$ performing an MIA  on $\overline{A}^{\Users}$ is to determine whether their target  $u^*$ contributed to $\overline{A}^{\Users}$, inferring \textit{IN} for $u^* \in \Users$ and \textit{OUT} for $u^* \not \in \Users$.

\subsection{Membership Classifier}
\label{subsec:membership_classifier}
Given an aggregate release $\overline{A}^{\Users}$ over $m$ users, an adversary infers membership of the target $u^*$ within the aggregation group $\Users$ by using a binary membership classifier. Classifiers are commonly instantiated as machine learning models~\cite{pyrgelis2017knock, pyrgelis2020measuring, zhang2020locmia}, but statistical models, like the log-likelihood function, have been applied as well~\cite{homer2008resolving, bauer2020towards}. To train the classifier, $Adv$ typically creates a balanced set of labeled size $m$ training aggregates~\cite{pyrgelis2017knock, zhang2020locmia, pyrgelis2020measuring, oehmichen2019opal}. Half of the aggregates include the target trace $L^{u^*}$ and are labeled $IN$, and the other half are labeled $OUT$. Training the classifier will create a decision boundary in the underlying space of aggregate releases~\cite{bishop2006pattern}. In the case of aggregate location data over ROIs $\Si$ and epochs $\T$, the decision boundary is characterized by a hypersurface that partitions the matrix space $\R^{|\Si| \times |\T|}$ into two sets.

\subsection{Threat Model}
\label{subsec: threat_model}
In this section, we present our Zero Auxiliary Knowledge (ZK) MIA threat model. It is commonly assumed in MIAs across various domains that the adversary has access to an auxiliary dataset and complete knowledge of the target record~\cite{pyrgelis2017knock,nasr2019comprehensive, salem2018ml, truex2019demystifying, yeom2018privacy, shokri2017membership}. Our ZK threat model relaxes both assumptions by eliminating the need for an auxiliary dataset and allowing for only partial knowledge of the target trace.

For context, we also describe threat models of previous MIAs on aggregate location data. All threat models consider an adversary $Adv$, whose goal is to determine whether a specific target user $u^*$ is included in the released aggregate $\overline{A}^{\Users}$. The aggregate is computed across $m$ users over ROIs $\Si$ and epochs $\T$. We assume that the locations of the ROIs $\Si$ are known.

\textbf{Knock-Knock~\cite{pyrgelis2017knock}: } The adversary  has an auxiliary dataset 
$Ref = \{L(u_1), ..., L(u_{|Ref|})\}$ of user traces, over the same locations and a similar population as the released aggregate. $Ref$ has at least $m$ traces, including the full target trace $L^{u^*}$.

\textbf{LocMIA~\cite{zhang2020locmia}: } The adversary knows $u^*$'s social network and has an auxiliary dataset $Ref = \{L(u_1), ..., L(u_{|Ref|})\}$ of user traces, over the same locations and a similar population as the released aggregate $\overline{A}^{\Users}$. $Ref$ has at least $m$ traces, including the traces of $u^*$'s friends, but not $L^{u^*}$.

\textbf{Zero Auxiliary Knowledge (ours): } The adversary knows a subset of the target $u^*$'s visits. Equivalently, the adversary knows a partial target trace $\tilde{L}^{u^*}$, such that $supp(\tilde{L}^{u^*}) \subset supp(L^{u^*})$.

KK MIA and LocMIA are reliant on the adversary's access to an extensive auxiliary dataset $Ref$. In particular, $Adv$ samples individual traces from $Ref$ to create training aggregates. These traces are also assumed to range over the same locations, and belong to a similar population as the traces aggregated in the release $\overline{A}^{\Users}$, in order to properly train the membership classifier (Section \ref{subsec:membership_classifier}). However, individual traces are known to be sensitive \cite{de2013unique}, and are unlikely to be made available, particularly when the data is aggregated as a privacy measure. Furthermore, because $Ref$ must contain at least $m$ traces, this assumption is impractical for even moderately sized aggregates. Although LocMIA removes prior knowledge about the target trace $L^{u^*}$, it must assume knowledge of $u^*$'s friends' traces to create a suitable proxy. More importantly, LocMIA still requires a large auxiliary dataset of individual traces.

In contrast, our Zero Auxiliary Knowledge adversary only requires that the adversary has knowledge about the target's location history. We emphasize that the adversary does not need to know the full trace $L^{u^*}$. Our threat model encompasses the case where only a few of the target's visits are known to the adversary. For example, the adversary may infer some of $u^*$'s visits from social media activity or direct observation.

\section{Related Work}
\label{sec:relatedwork}
\textbf{MIAs on Aggregate Location Data.} 
MIAs on aggregate location data have been shown to be successful on multiple location datasets \cite{pyrgelis2017knock,pyrgelis2020measuring,zhang2020locmia}, using a binary classifier to perform the inference task. The performance of the MIAs on small aggregates ($<500$) is especially well-studied, as the influence of the target is easier to distinguish~\cite{pyrgelis2017knock}. For example, the KK MIA by \citet{pyrgelis2017knock}, achieved  $AUC> 0.83$ when attacking size $100$ aggregates across two different mobility datasets. LocMIA \cite{zhang2020locmia} is another MIA on aggregate location data, which removes prior knowledge about the target trace. Instead, LocMIA assumes access to social network information, and the traces of the target's friends, in order to construct a proxy for the target's real trace. However, both KK MIA and LocMIA crucially require the adversary to have access to a large auxiliary dataset to train the binary classifier. In contrast, our ZK MIA does not require any auxiliary dataset, and only requires partial information about the target trace (e.g. A random $10\%$ proportion of their visits). ZK MIA therefore addresses the research gap of the MIA risk posed by a less knowledgeable attacker.  The distinctions in prior knowledge are discussed in depth in Section \ref{subsec: threat_model}. Our ZK MIA also features a novel approach, being the first MIA on aggregate location data to use synthetic trace generation.

\textbf{Generation of Synthetic Location Traces.} 
There are many techniques for generating synthetic location traces that capture realistic human mobility patterns ~\cite{kulkarni2017generating,kulkarni2018generative,ouyang2018non, karagiannis2007power, jahromi2016simulating, lee2009slaw}. However, since our ZK MIA requires generating suitable synthetic traces without using additional information, this heavily limits the scope of applicable techniques. RNNs, GANs, and copulas have been used to generate synthetic traces that collectively approximate a real mobility dataset~\cite{kulkarni2017generating,kulkarni2018generative,ouyang2018non}. However, these techniques require real traces to train the model, which the ZK adversary does not have. Many of the state-of-the-art mobility models are also unsuitable because they simulate small-scale continuous trajectories (e.g. walks on campus) \cite{karagiannis2007power, lee2009slaw}. In contrast, location aggregates typically comprise discrete traces over a metropolitan region. We therefore identified a probabilistic unicity model by \citet{farzanehfar2021risk}, which requires only four statistical parameters to guide the synthetic generation. We demonstrate that we can non-trivially adapt this model for the ZK MIA. In particular, we develop methods to precisely estimate these parameters from the aggregates in order to produce realistic synthetic location aggregates.

\textbf{MIAs with Reduced Auxiliary Data.} 
Previous attempts have been made~\cite{shokri2017membership,salem2018ml,truex2019demystifying,yeom2018privacy,crectu2021correlation,guepin2023synthetic} to relax the standard assumption of an adversary's access to an auxiliary dataset that has high statistical similarity with the attacked dataset, e.g. sampled from the same distribution~\cite{shokri2017membership,nasr2019comprehensive}. In the setting of machine learning (ML) models, where an MIA infers whether a record was a part of the ML model's training set, \citet{shokri2017membership} proposed an MIA without auxiliary data. Instead, they use synthetic data, which they generate using the ML model's confidence scores. Similarly, \citet{salem2018ml} trained an MIA using unrelated data, e.g., training on text data to attack an image model. These approaches require access to the ML model, and they train on features that are specific to ML models, such as the top-$k$ confidence scores, which may be shared across ML models pertaining to different types of data. 
\begin{figure}[h!]
    \centering
    \includegraphics[width = \linewidth, page=2]{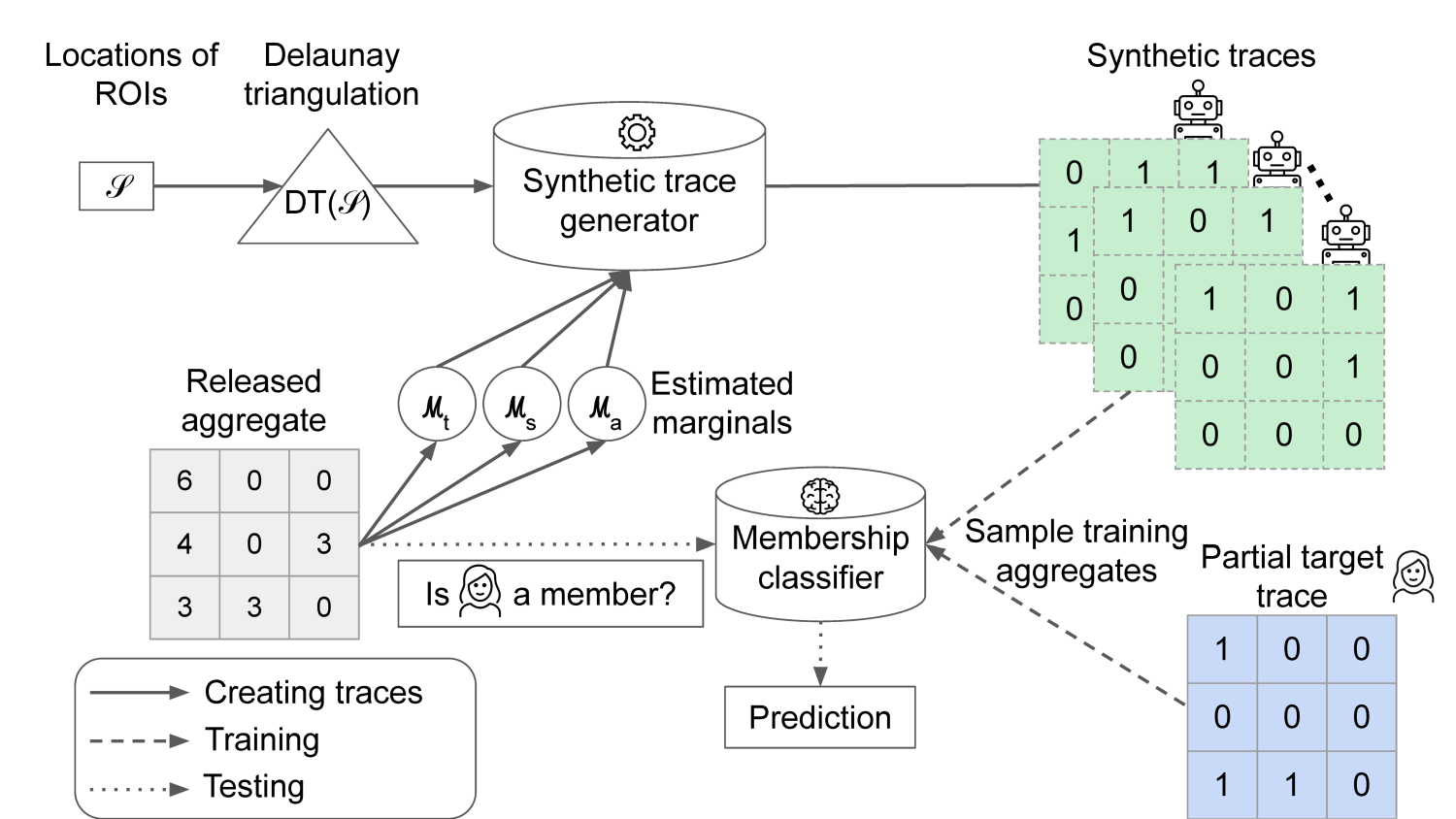}
    \caption{ZK MIA architecture: $Adv$ first creates synthetic traces, then uses them with the partial target trace to train the membership classifier, before predicting membership.}
\label{fig:ZK_architecture}
\end{figure}
In contrast, the features used to train an MIA on aggregate location data explicitly depend on the specific regions, times, and population over which the aggregates are computed. In the context of synthetic generators, \citet{guepin2023synthetic} performed an MIA against synthetic data using the released synthetic dataset as the auxiliary dataset. However, in the case of aggregate location data, the release cannot be directly used as a reference dataset, since it does not contain individual records.

\section{Methodology}
\label{sec:methodology}
We dedicate Sections \ref{subsec: ZK_MIA}-\ref{subsec:marginal_approximations} to explaining the synthetic-based methodology of our Zero Auxiliary Knowledge MIA. Section \ref{subsec:paired_sampling} explains the paired sampling mechanism on training aggregates, which boosts the performance of ZK MIA and KK MIA, as shown in Section \ref{exp: paired_sampling}.

\subsection{Zero Auxiliary Knowledge MIA Framework}
\label{subsec: ZK_MIA}
We implement the Zero Auxiliary Knowledge MIA as a binary classifier. However, whereas $Adv$ uses their auxiliary dataset as a reference for creating training aggregates for KK MIA and LocMIA, $Adv$ instead uses the reference of synthetic traces that they created from the released aggregate. Furthermore, if the full target trace $L^{u^*}$ is not known, $Adv$ may instead use a partial trace when creating the $IN$ training aggregates. Figure \ref{fig:ZK_architecture} illustrates ZK MIA's overall attack architecture.

\subsection{Generating Synthetic Traces from Aggregate Location Data}
\label{subsec:generating_synthetic_traces_from_agg}

In order to 
generate synthetic traces for our Zero Auxiliary Knowledge MIA, we adapt a probabilistic mobility model \cite{farzanehfar2021risk}. \citet{farzanehfar2021risk} developed this model to reproduce unicity patterns in large populations. The model requires four statistical parameters, described below and illustrated in Figure \ref{fig:four_inputs} of the Appendix. 
Recall that for a discrete random variable $X$ taking values in $x_1, ..., x_N$,  its probability mass function (p.m.f.) $\mathcal{P}: \{x_1, ..., x_N\} \to [0,1]$ maps each possible value to its corresponding probability.

\begin{enumerate}
    \item The marginal space distribution $\mathcal{P}_S: \Si \to [0,1]$  is a p.m.f. such that $\mathcal{P}_S(s)$ is proportional to the number of visits to ROI $s$ by users in $\Pop$ across all epochs in $\T$.
    \item The marginal time distribution $\mathcal{P}_T: \T \to [0,1]$ is a p.m.f. such that $\mathcal{P}_T(t)$ is proportional to the number of visits during epoch $t$ by users in $\Pop$ across all ROIs in $\Si$.
    \item The marginal activity distribution $\mathcal{P}_A$ models the total number of visits recorded within $\Si$ during $\T$ by a user drawn from $\Pop$. 
    \item The Delaunay triangulation, denoted $DT(\Si)$, is a triangulation with vertices corresponding to the set of positions (longitude and latitude) of ROIs in $\Si$. $DT(\Si)$ has the property that no vertex lies inside the circumcircle of any triangle in $DT(\Si)$~\cite{delaunay1934sphere}.
\end{enumerate}

We note that the Delaunay triangulation $DT(\Si)$ is determined by the locations of the ROIs. Since the locations of the ROIs are assumed to be known (Section \ref{subsec: threat_model}), $DT(\Si)$ can be immediately obtained from the release. We explain how the other statistical inputs, the three marginal distributions, can be approximated from the released aggregate in Section \ref{subsec:marginal_approximations}.

We now describe our procedure, adapted from ~\citet{farzanehfar2021risk}, for generating synthetic traces using the four inputs. For each synthetic trace $L^{syn}$, we first sample the number of visits $n_{visits}$ from the activity marginal $\mathcal{P}_A$. This determines the number of nonzero entries in the $\Si \times \T$ matrix $L^{syn}_{s,t}$. Second, we sample an origin ROI $s_0$ from the space marginal $\mathcal{P}_S$, and a connected sub-graph $C(s_0)$ from the Delaunay triangulation $DT(\Si)$, such that $s_0 \in C(s_0)$. $C(s_0)$ will correspond to the set of ROIs that may be visited in $L^{syn}$. This is done to emulate the natural tendency to move to and from the same proximate locations (e.g. home and work). Finally, we sample $n_{visits}$ spatiotemporal visits $(s,t)$ for which we set $L^{syn}_{s,t}=1$. For each visit, $s$ is sampled from the space marginal $\mathcal{P}_S$ restricted to $C(s_0)$, and $t$ is sampled from the time marginal $\mathcal{P}_T$. All the sampling steps are independent. 

We make two modifications of the original algorithm~\cite{farzanehfar2021risk}.
First, to avoid over-saturating unpopular regions, we sample the origin ROI $s_0$ according to the space marginal $\mathcal{P}_S$ rather than uniformly. Second, to allow users to visit multiple ROIs within the same epoch, we sample the epochs with replacement rather than without replacement. Our procedure is summarized in Algorithm~\ref{alg:synthetic_generation} of the appendix.

\subsection{Obtaining Accurate Marginals}
\label{subsec:marginal_approximations}
In order to generate suitable synthetic traces for ZK MIA, $Adv$ must estimate the marginal distributions $\mathcal{P}_S, \mathcal{P}_T, \mathcal{P}_A$, computed over the full population $\Omega$, using the aggregate release $\overline{A}^\Users$. In this section, we motivate and justify the techniques that we developed to obtain strong estimates $\widehat{\mathcal{P}}_S, \widehat{\mathcal{P}}_T, \widehat{\mathcal{P}}_A$. This task is especially challenging when privacy measures distort the aggregate data. We develop separate techniques to correct for bias in the case of SSC, and to correct for noise in the case of DP. The effects are shown in Figures \ref{fig:log_compress} and \ref{fig:power_transform} respectively. Algorithm \ref{alg:marginals} in the Appendix summarizes how we approximate all three marginals from $\overline{A}^\Users$.

\subsubsection{\textbf{Estimating Space and Time Marginals}}
\label{subsec: est_space_time_marginals}

 \begin{figure}[!h]
    \centering
    \includegraphics[width = 0.7\linewidth, page=1]{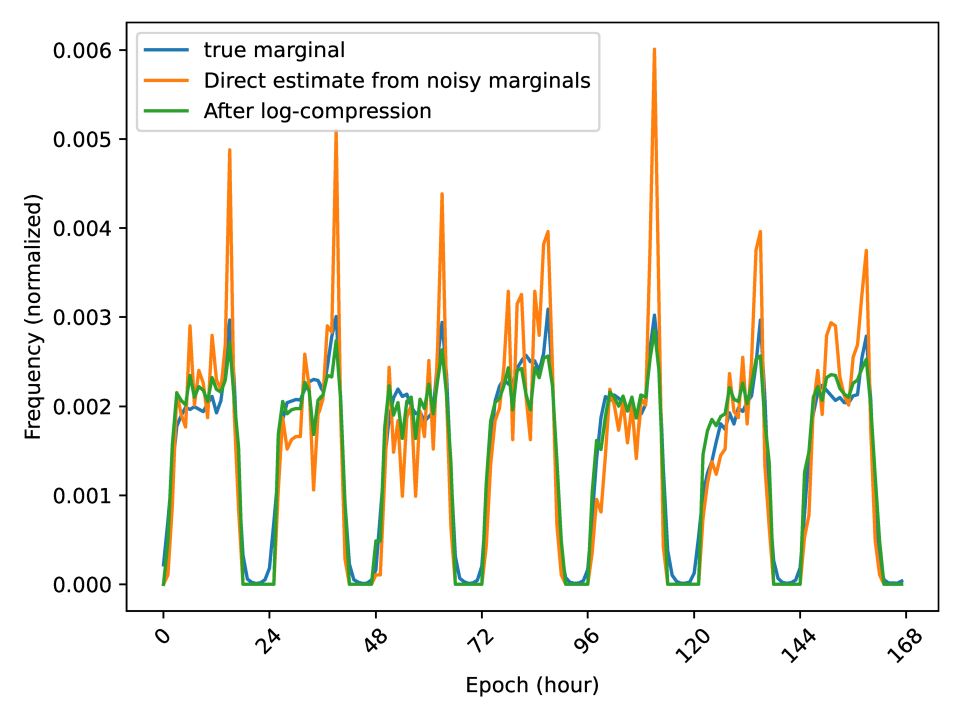}
    \caption{Log compression for SSC aggregates:
    SSC biases the estimate obtained from the aggregate by creating more extreme values. The true time marginal from the CDR dataset (plotted for the first week) is better approximated after the empirical estimate from the aggregate ($m=1000$, $k=1$) undergoes log compression $\log(1+\gamma x)$ with $\gamma$ chosen as in (8).}
    \label{fig:log_compress}
\end{figure}

Suppose that the data collector releases the aggregate $\overline{A}^\Users$, which may or may not have privacy measures. $Adv$ can directly compute the empirical space and time marginals, which we denote  $\widehat{\mathcal{P}}_S^0$ and $\widehat{\mathcal{P}}_T^0$, from the released aggregate matrix $\overline{A}^\Users$.
\begin{align}
    &\widehat{\mathcal{P}}^0_S(s) = \frac{1}{\lVert \overline{A}^\Users \rVert_1} \sum_{t=1}^{|\T|} \overline{A}^\Users_{s,t}\\
    &\widehat{\mathcal{P}}_T^0(t) = \frac{1}{\lVert \overline{A}^\Users \rVert_1} \sum_{s=1}^{|\Si|} \overline{A}^\Users_{s,t}
    \label{eq:empirical_marginals}
\end{align}

\textbf{Raw Aggregate: }
In the case where the released aggregate provide the raw counts, i.e. $\overline{A}^\Users = A^\Users$, the empirical marginals tend to be highly accurate. An example is shown in Figure \ref{fig:raw_time_marg} in the Appendix. The accuracy of these estimates is intuitive because we expect the mobility patterns of an aggregation group $\Users$ to resemble those of the population $\Pop$. Thus, we set $\widehat{\mathcal{P}}_{\underscore}= \widehat{\mathcal{P}}^0_{\underscore}$ if the released aggregate is unmodified. We use the subscript $\underscore$ to indicate generality for both the space $S$ and time $T$ marginals.

\textbf{Suppressed Aggregate: } However, if the data collector applies SSC with threshold $k$, i.e. $\overline{A}^\Users = A_{SSC}^\Users(k)$, then this will systematically bias the empirical marginal $\widehat{\mathcal{P}}_{\underscore}^0$, because popular ROIs and epochs are more likely to evade suppression. It is therefore easy to see that suppression will reduce the observed probabilities of less popular entries and boost the probabilities of more popular entries. 

To correct the bias, we flatten the empirical estimate $\widehat{\mathcal{P}}_{\underscore}^0$ by boosting low frequency counts and reducing high frequency counts. Upon the insight that $\widehat{\mathcal{P}}_{\underscore}^0$ can be likened to an audio signal, we adapt the logarithmic compression technique used to reduce dynamic range \cite{FMP_C3S1_LogCompression, miguel2004tempo}
\begin{align}
    x \to \log{(1+\gamma x)}, \text{  } x \ge 0, 
\end{align}
where the scaling factor $\gamma \ge 0$ regulates the compression level \cite{FMP_C3S1_LogCompression}. In music signal processing, $x 
\ge 0$ corresponds to the intensity of a given frequency. In our case, $x \ge 0$ corresponds to probabilities within the empirical marginal $\widehat{\mathcal{P}}_{\underscore}^0$. We choose 
\begin{align}
    \gamma(\widehat{\mathcal{P}}_{\underscore}^0) = \max\limits_{x \in \widehat{\mathcal{P}}_{\underscore}^0: x > 0 }\frac{1}{x}
\end{align}

 \begin{figure}[!h]
    \centering
    \includegraphics[width = 0.7\linewidth, page=1]{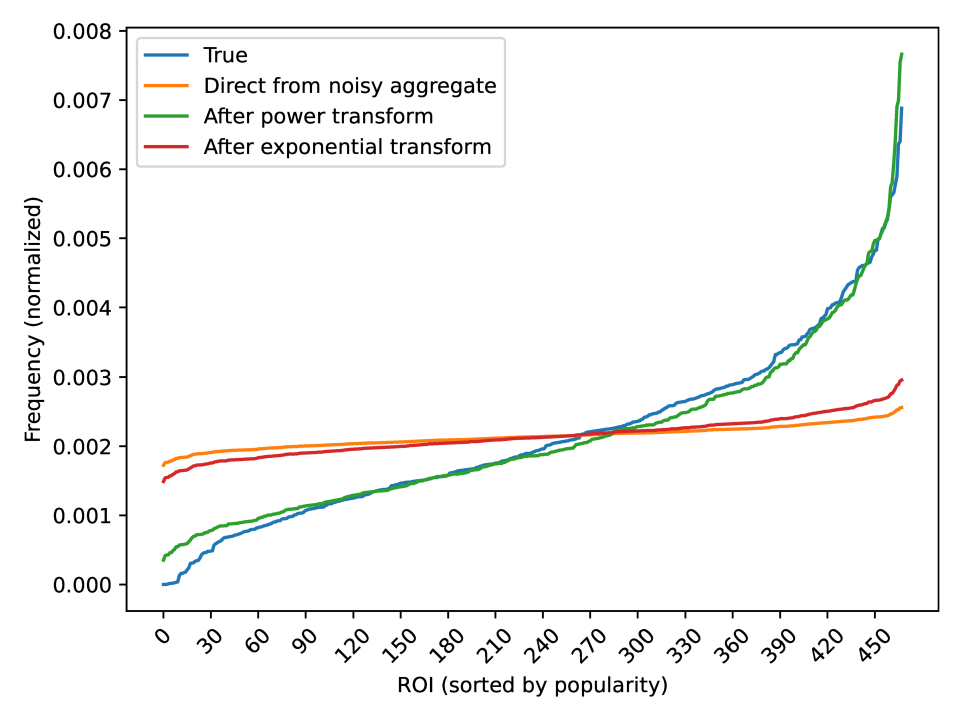}
    \caption{Power transformation for DP aggregates: DP noise compresses the estimate obtained from the aggregate. The true space marginal from the CDR dataset (organized by popularity) is better approximated after the empirical estimate from the aggregate ($m=1000$, $\frac{\Delta}{\varepsilon} = 1$) undergoes power transformation $x^p$ with $p$ selected according to Algorithm \ref{alg:p_selection}.}
    \label{fig:power_transform}
\end{figure}

 to automatically parameterize
 $\gamma$ based on the smallest observed non-zero probability. We therefore estimate $\widehat{\mathcal{P}}_{\underscore} = \log(1 + \gamma(\mathcal{P}_{\underscore}^0)$, where we omit the normalization constant. 
 
 We do not argue that our choice of method and parameter $\gamma$ is optimal. However, the debiasing substantially improves the estimate, as shown in Figure \ref{fig:log_compress}, and it is done without additional information.
 
\textbf{DP Aggregate.} If $\varepsilon$-DP noise is added to each entry in the aggregate release, i.e. $\overline{A}^\Users_{s,t} = A_{DP}^\Users(\varepsilon)_{s,t} = {A}^\Users_{s,t} + Lap(\frac{\Delta}{\varepsilon})$, then the noise will overpower the signal in the computation of the empirical marginals (Eq. \ref{eq:empirical_marginals}). This follows from the fact that location aggregates are high-dimensional sparse matrices \cite{pyrgelis2020measuring}. Therefore, conversely to the SSC case, the probabilities within the empirical marginals are compressed, since each probability is characterized mostly by thousands of independent noise samples. This effect is visualized for several different noise scales in Figure \ref{fig:Milan_space_laplace} of the Appendix. We also prove that under strong sparsity assumptions, $\widehat{\mathcal{P}}_{S}^0$ converges to the discrete uniform distribution on $\Si$ as the number of epochs in the observation period $|\T| \to \infty$, in Theorem \ref{theorem:marg_conv_unif} of the Appendix. 

To correct the low variance of the observed probabilities, we propose the power transformation $x^p$ with $p>1$, followed by renormalization. It is easy to see that this will increase the variance since the probabilities are in $[0,1]$. Automatically calibrating the power $p>1$ is a delicate matter. To do so, we start with $p=1$ and augment $p$ gradually until the transformed distribution achieves the target variance $\sigma^2$. Without any prior knowledge, we consider the case where each probability is randomly drawn. Equivalently, each probablity $x_i$ is sampled from $Unif(0,1)$, and then renormalized so that the total probability is $1$. Let $p_i$ denote the probabilities after normalization, and $\bar{p}$ be the mean of the normalized probabilities. For the space marginal, the variance is
\begin{align}
    \sigma^2= \frac{\sum_{i=1}^{|S|} (p_i - \bar{p})^2  }{|S|}.
\end{align}

 \begin{figure}[!h]
    \centering
    \includegraphics[width = 0.7\linewidth, page=1]{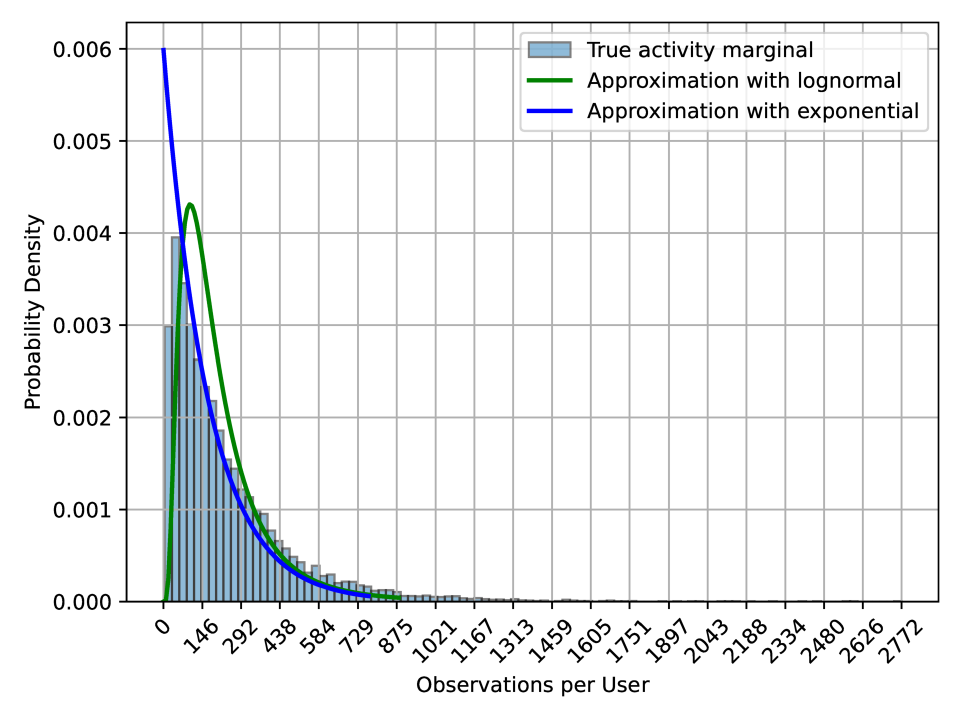}
    \caption{Exponential estimate for activity marginal: Once the mean number of visits is approximated, the activity marginal from the CDR dataset is best approximated by a lognormal distribution with an optimal skew parameter. However, it can also be approximated without additional parameters by an exponential distribution. }
    \label{fig:activity_lognormal_exp}
\end{figure}

For the variance from the time marginal, we replace $|S|$ with $|T|$ in the above equation. The algorithm for selecting $p$ is given in Algorithm \ref{alg:p_selection} of the Appendix. Figure \ref{fig:power_transform} shows that our automatically parameterized power transformation significantly improves the estimate. In contrast, the exponential transformation $(e^x-1)/\gamma$, which is the inverse of the log compression function $\log(1+\gamma x)$, fails to denoise because the inverse is inapplicable after considering normalization. 
 
\subsubsection{\textbf{Estimating the activity marginal}}

The released aggregate $\overline{A}^\Users$ does not leak granular information about the activity marginal $\mathcal{P}_A$. However, $Adv$ may obtain the empirical mean number of visits per user according to the released aggregate,
\begin{align}
    \widehat{\mu}_0 = \frac{1}{m} \sum_{s,t \in \Si \times \T} \overline{A}^\Users_{s,t}.
\end{align}
If the aggregate is raw, then we expect $\widehat{\mu}_0$ to be a strong estimate due to well-known regularity results about population-wide mobility activity \cite{schneider2013unravelling,seshadri2008mobile, farzanehfar2021risk}. In this case, $Adv$ sets $\hat{\mu}=\widehat{\mu}_0$.

However, if either SSC or $\varepsilon$-DP is applied, then the estimate $\widehat{\mu}_0$ would fail. Algorithm \ref{alg:estimate_mean} in the Appendix describes how $Adv$ can obtain a better estimate $\hat{\mu}$. Given an aggregate release $\overline{A}^\Users$ of size $m$, $Adv$ can use $\widehat{\mathcal{P}}_{S}$ and $\widehat{\mathcal{P}}_{T}$ to iteratively improve their estimate, starting with ${\mu}_0$. Given guess $\widehat{\mu}_n$, $Adv$ creates a synthetic aggregate by generating $m$ synthetic traces, parameterized by $\widehat{\mathcal{P}}_{S},\widehat{\mathcal{P}}_{T}$ and $\widehat{\mathcal{P}_A} \sim \widehat{\mu}_n$. $Adv$ may then apply the same privacy measures that were applied on $\overline{A}^\Users$. $\widehat{\mu}_{n+1}$ is obtained by increasing or decreasing $\widehat{\mu}_{n}$ relative to the difference in counts with $\overline{A}^\Users$.

Once $\hat{\mu}$ is obtained, $Adv$ can simply pick $\widehat{\mathcal{P}_A} \sim \widehat{\mu}$, such that each synthetic trace has $\hat{\mu}$ visits. However, it is well known that human mobility activity follows a heavy-tailed distribution, i.e., a heavier tail than the exponential distribution. The best approximations are lognormal, beta, or power-law distributions \cite{farzanehfar2021risk, schneider2013unravelling, seshadri2008mobile}.

It would be reasonable for $Adv$ to use a heavy-tailed distribution with mean $\hat{\mu}$, but these distributions require a second parameter, e.g. skewness, to determine the distribution shape. $Adv$ can use well-known parameters from other cities' datasets to complete the estimate \cite{schneider2013unravelling, seshadri2008mobile}, but to ensure that $Adv$ does not use additional knowledge, we assume that they use the sub-optimal estimate $\hat{\mathcal{P}}_A \sim Exp(\hat{\mu})$, as shown in Figure \ref{fig:activity_lognormal_exp}.

\subsection{Paired Sampling for Training}
\label{subsec:paired_sampling}

When MIAs target high-dimensional aggregate data, such as location data, the membership classifier must handle noise arising from thousands of entries, which are unrelated to the target record.  For example, many of the \textit{IN} training aggregates may coincidentally have high counts in entries that are absent in the target record. This would skew the decision boundary of the membership classifier, which may lead to false positives when testing. We may similarly obtain false negatives due to spurious patterns within the \textit{OUT} training aggregates. These challenges are compounded by the implementation of privacy measures. For example, $\varepsilon-$DP would add noise to each entry. Given the dimensionality and nature of the aggregate data, spurious patterns will likely skew the decision boundary, even when hundreds or thousands of training aggregates are sampled. Given a fixed number of training samples, we demonstrate that the way in which the training set is sampled strongly influences the performance of the MIA. In particular, the sampling technique can guide the convergence of the decision boundary in order to prevent misclassification due to noise.

To the best of our knowledge, all previous MIAs on aggregate location data sampled their training set via independent random sampling \cite{pyrgelis2017knock, oehmichen2019opal, zhang2020locmia, pyrgelis2020measuring}. Training aggregates are created by independently sampling groups of $m$ users from the population $\Omega$ and labeling them according to the target $u^*$'s presence.

On the one hand, independent sampling discourages overfitting to the training data by exposing the classifier to a wide variation of samples. On the other hand, independent sampling does nothing to prevent spurious patterns from distorting the decision boundary. 

We propose a paired sampling technique to guide the convergence of the decision boundary. The idea is to use sampling to help the classifier identify the differential impact of the target record at the aggregate level. Paired sampling independently samples groups of $m-1$ users from $\Omega \setminus \{u^*\}$. Then, an \textit{IN} sample is created by adding $u^*$ as the group's last member, and an \textit{OUT} sample is created by adding another randomly selected user. The training set is therefore characterized by a set of \textit{IN}/\textit{OUT} pairs, which differ in exactly one record (the target's). If noise is added to aggregates prior to release, then $Adv$ must inject the same noise sample $\varepsilon$ to each paired sample, $A^{\Users^{IN}}$ and $A^{\Users^{OUT}}$, to ensure that the target's differential impact is preserved between each \textit{IN}/\textit{OUT} pair.

\begin{figure*}[h!]
    \centering
    \begin{subfigure}[b]{0.49\linewidth}
        \centering
        \includegraphics[width=0.7\textwidth]{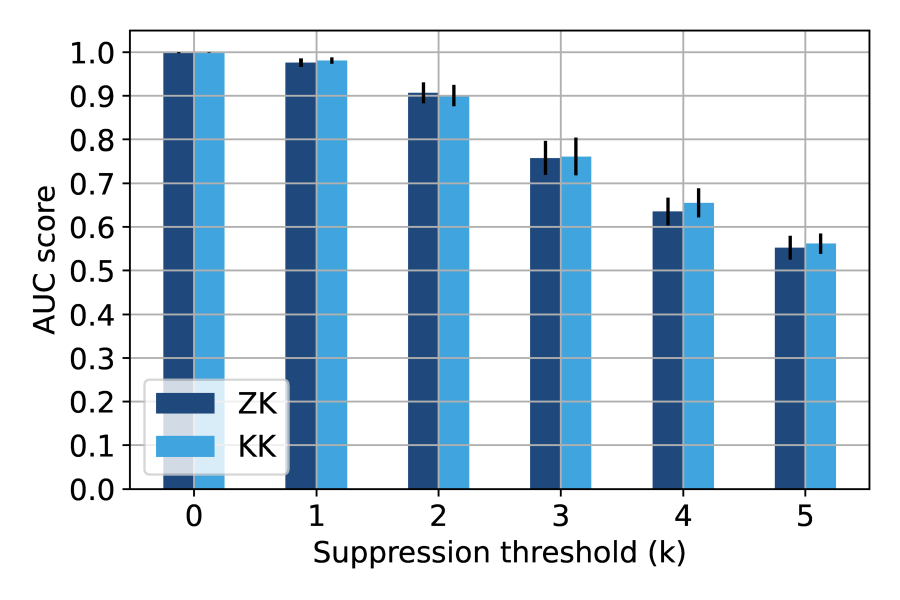}
        \caption{CDR dataset}
        \label{fig:bucket_KK_ZK_cdr}
    \end{subfigure}%
    \hfill 
    \begin{subfigure}[b]{0.49\linewidth}
        \centering
        \includegraphics[width=0.7\textwidth, page=1]{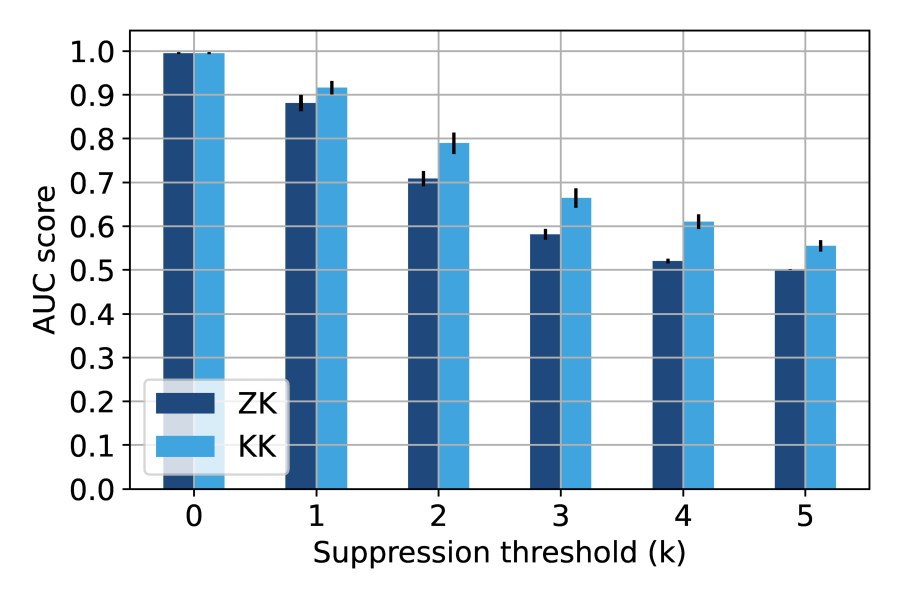}
        \caption{Milan}
        \label{fig:bucket_KK_ZK_milan}
    \end{subfigure}
    \caption{Mean AUC scores with standard error for ZK and KK on size $1000$ aggregates with various suppression thresholds}
    \label{fig:bucket_ZK_KK}
\end{figure*}

Paired sampling therefore actively encourages the membership decision boundary to be formed based on relevant criteria related to the target. It also discourages spurious decision boundaries because of the high degree of similarity between \textit{IN}/\textit{OUT} pairs. An extreme value in an aggregate entry from an \textit{IN} sample will be matched with a similar value from its paired \textit{OUT} sample, with high probability. However, we note that using paired sampling effectively halves the training variation compared to independent sampling. Our experiments in Section \ref{exp: paired_sampling} demonstrate that paired sampling outperforms independent sampling across all tested settings of $\varepsilon$-DP noise addition. Hence, guiding the decision boundary towards relevant membership criteria often takes precedence over maximizing training variation. For completeness, we note that while we developed, studied, and named paired sampling independently, we later found that a similar idea was used in \citet{bauer2020towards} but had not been compared to independent sampling nor used elsewhere in the literature so far to the best of our knowledge. 

\section{Experimental Setup}
\label{sec:experimentalsetup}

To evaluate the efficacy of our ZK MIA, we compare it against the state-of-the-art Knock Knock (KK) MIA~\cite{pyrgelis2017knock, pyrgelis2020measuring} using aggregated location data from two different datasets.

\subsection{Datasets.}
In this section, we describe the two location datasets used for evaluating the MIAs, and discuss ethical considerations of the data collection and usage.

\noindent
\textbf{CDR: }The first dataset, which we refer to as "CDR", is a private dataset, shared with us by Flowminder~\cite{flowminder} for the purpose of this research. The raw dataset comprises timestamped and geo-tagged call records of approximately $11,000$ mobile phone users within a Latin American metropolitan area. The observation period is June 2021, with epochs defined by the $720$ hourly timeslots. The ROIs are defined by the service regions of approximately $500$ cellular antenna towers within the metropolitan area, which spans $\sim 150\text{km}^2$. The users were selected such that they registered at least one visit per week, to omit users who changed SIM cards, and such that the majority of their visits are within the region, to ensure that they are residents. $50$ target users for the MIAs were randomly selected by Flowminder. A histogram of the number of visits over the target traces is plotted in Figure \ref{fig:trace_size_CDR} in the Appendix.

\vspace{1em}
\noindent
\textbf{Milan: }The second dataset is the Milan Social Pulse dataset~\cite{DVN/9IZALB_2015}, made publicly available as part of the Telecom Italia Big Data Challenge~\cite{barlacchi2015multi}. This dataset comprises timestamped and geo-tagged tweets from $4840$ mobile phone users within the Milano region. The ROIs are defined by a grid of $100$ points, each with an approximate area of 256 $\text{km}^2$. We consider the location data from the first week of data, yielding $168$ hourly epochs. We do not delete any users from the dataset prior to aggregation. We randomly select $50$ targets among users who tweeted at least $10$ times during the observation period. A histogram of the number of visits over the target traces is plotted in Figure \ref{fig:trace_size_Milan} in the Appendix.

\vspace{1em}
\noindent
\textbf{Ethical Considerations:} 
Because of the sensitivity of location data, we did not access raw individual-level data, and instead collaborated with Flowminder (FM) to develop a privacy-preserving data-sharing pipeline for the purpose of this research~\cite{de2018privacy}. More specifically, data sharing was restricted to pre-computed aggregate matrices (labeled according to target membership, computed by FM on the data provider server) and $50$ target traces randomly chosen by FM. To further mitigate the privacy risk, the $\sim 500$ ROI and $720$ epoch indices were randomly permuted in the shared aggregate and target trace matrices, according to a mapping known only by FM. This random permutation relabeled the space and time indices, enabling us to test the MIAs without knowing the true times or locations. The graphs of the marginal statistics (see Figures \ref{fig:log_compress}, \ref{fig:power_transform}, \ref{fig:activity_lognormal_exp}) were plotted by FM and shared with us. All the data shared by FM with us is subject to a research contract between FM and our institution and was kept on our segregated server. The Milan dataset, derived from geo-tagged tweets, remains publicly available, and was only used for the purpose of testing the MIAs.

\subsection{MIA Implementation}

We perform a fair comparison between KK MIA and ZK MIA by training the binary membership classifier using the same parameters and architecture. This also helps us isolate the effect of removing auxiliary data on performance. We use a Logistic Regression binary classifier with default hyperparameters and $\mathbb{L}_1$ regularization, implemented with $sklearn$. The number of training groups $n_{train}=400$ matches previous implementations~\cite{pyrgelis2017knock, pyrgelis2020measuring}, and the groups are selected using paired sampling, unless specified otherwise. We additionally fine tune the decision boundary using $n_{val} = 100$ balanced independently sampled validation groups. We flatten the aggregates and feed them directly into the classifier as a vector, without any processing, such as PCA or feature extraction. Finally, as done in~\citet{pyrgelis2017knock}, $Adv$ applies the same privacy measures to training and validation aggregates, if the released aggregate is privacy-aware.

\textbf{Knock-Knock.} \indent
To implement the Knock-Knock MIA, we provide $Adv$ with a reference group $Ref$ of $5000$ real user traces (including the target trace $L^{u^*}$) when attacking the larger CDR dataset. We set $|Ref| = 2500$ for the Milan dataset. We note that this significantly surpasses previous reference sizes ($|Ref| = 1100$) implemented by~\citet{pyrgelis2017knock}.

\begin{figure*}[h!]
    \centering
    \begin{subfigure}[b]{0.49\linewidth}
        \centering
        \includegraphics[width=0.7\linewidth]{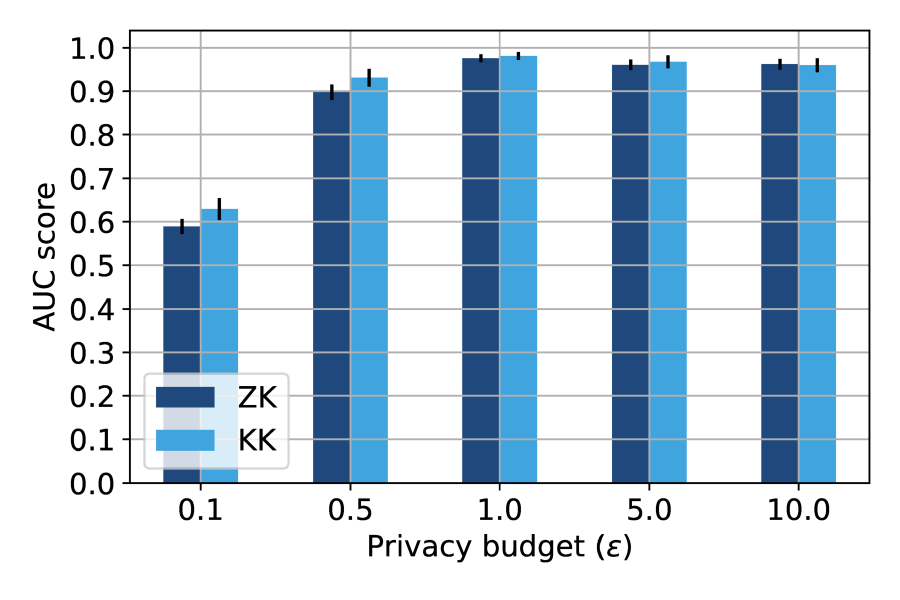}
        \caption{Event level DP: CDR dataset}
        \label{fig:event_level_DP_CDR}
    \end{subfigure}%
    \begin{subfigure}[b]{0.49\linewidth}
        \centering
        \includegraphics[width=0.7\linewidth]{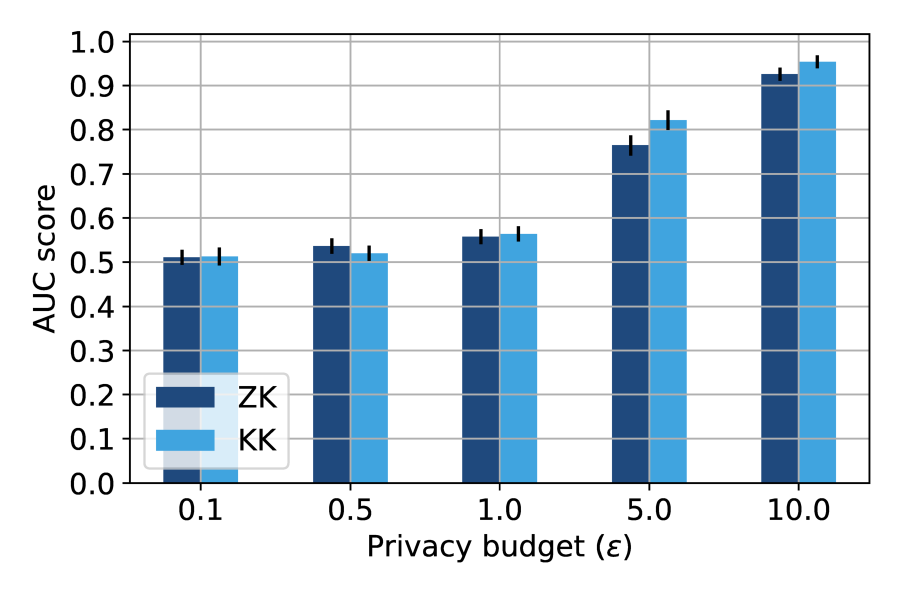}
        \caption{User-day level DP: CDR dataset}
        \label{fig:user_day_DP_CDR}
    \end{subfigure}
    \begin{subfigure}[b]{0.49\linewidth}
        \centering
        \includegraphics[width=0.7\linewidth]{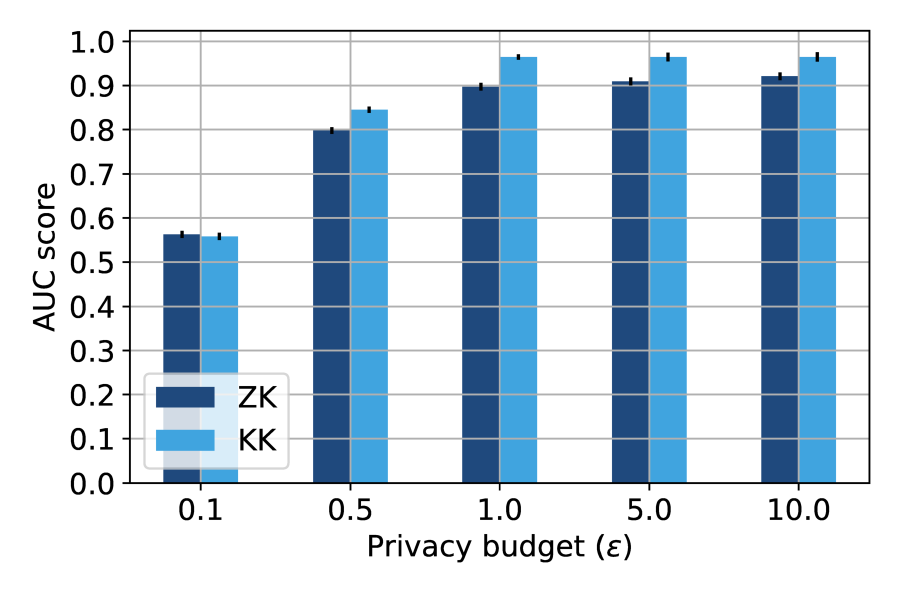}
        \caption{Event level DP: Milan}
        \label{fig:event_level_DP_Milan}
    \end{subfigure}%
    \begin{subfigure}[b]{0.49\linewidth}
        \centering
        \includegraphics[width=0.7\linewidth]{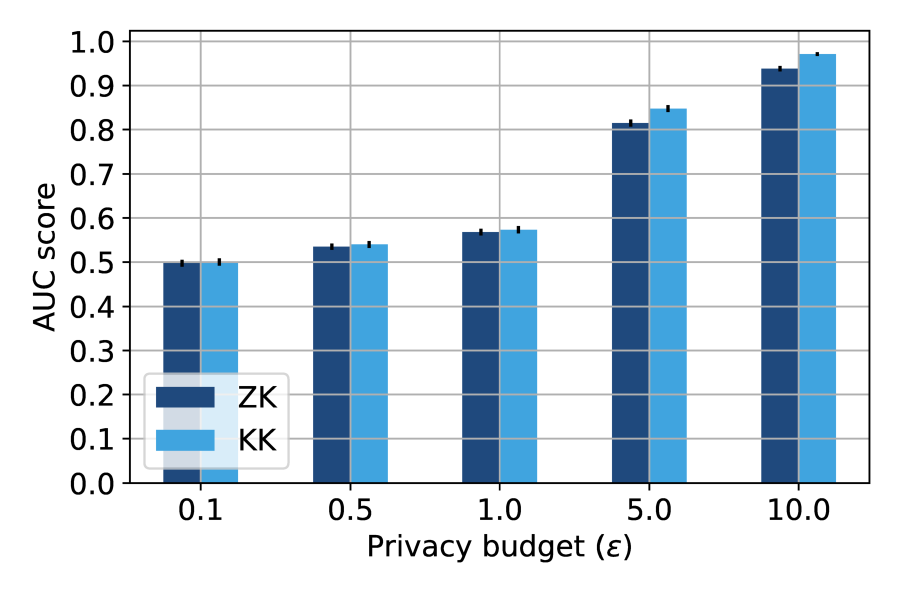}
        \caption{User-day level DP: Milan}
        \label{fig:user_day_DP_Milan}
    \end{subfigure}
    \caption{Mean AUC scores with standard error for ZK and KK on size $1000$ aggregates with various privacy units and budgets \(\varepsilon\).}
    \label{fig:dp_comparison}
\end{figure*}

\textbf{Zero Auxiliary Knowledge.}\indent
Our Zero Auxiliary Knowledge MIA is structurally identical to KK (PS), but ZK has a synthetic reference $Ref$, rather than a set of real traces. This reference is created according to the methodology detailed in Section \ref{sec:methodology}. Setting $|Ref|=5000$ allows for a direct comparison of the functionality of synthetic traces to real ones for the purpose of the MIA. However, we remark that capping the number of synthetic traces at $5000$ is an artificial restriction, since $Adv$ may generate arbitrarily more synthetic traces and achieve better performance, as shown in Figure \ref{fig:n_syn} of the Appendix. By default, we assume full access to the target trace $L^{u^*}$. This assumption is relaxed for experiment \ref{exp:partial_trace}.

\subsection{Evaluation}

\textbf{Default Experimental Parameters. } We randomly select $n_{targets} = 50$ targets from the dataset for evaluation. These targets are re-used for each experiment. Furthermore, in each experiment, $n_{test} = 100$ \textit{independently sampled} balanced test aggregates are created for each target, and shared across both MIAs to ensure that the test sets are identical. As done in ~\cite{pyrgelis2017knock}, the test aggregate user groups are sampled from a disjoint set of users to the Knock-Knock adversarial reference $Ref$, plus the target $u^*$. This corresponds to roughly $11, 000 - 5000 = 6000$ user traces for CDR and $5000-2500 = 2500$ user traces for Milan.

We perform all experiments on size $m=1000$ aggregates, which matches the largest aggregate size tested in ~\citet{pyrgelis2017knock} and exceeds the largest aggregate size tested in ~\citet{zhang2020locmia} (800). We do not vary $m$ since the relationship between aggregate group size and MIA effectiveness has already been documented extensively ~\cite{pyrgelis2017knock, pyrgelis2020measuring, zhang2020locmia}. We perform the Knock-Knock and Zero Auxiliary Knowledge MIAs 
in this setting under different privacy measures. We also perform experiments such that $Adv$ only knows a fraction $p_{u^*}$ of the target trace $L^{u^*}$, but by default, we assume that they know the full trace, i.e. $p_{u^*} = 1$. 

\textbf{Evaluation Metrics. \indent} In the past, MIAs on aggregate location data have been primarily evaluated using the area under the ROC curve (AUC score) as a metric ~\cite{pyrgelis2017knock, zhang2020locmia}. For this reason, and its suitability for assessing the strength of a binary classifier, we use the mean AUC over all targets as our primary metric. However, we also include the mean attack accuracy over all targets as a secondary metric, listing these scores in Section \ref{sec: accuracy} of the Appendix.

\section{Experimental Results}
\label{sec:results}
\subsection{Against Suppression of Small Counts}
\label{exp:ssc}

We first compare the performances of KK and ZK on aggregates whose counts have been suppressed according to threshold $k$. We apply SSC with thresholds $k \in \{0,1,2,3,4,5\}$ on test aggregates of size $m=1000$. We note that the case $k=0$ corresponds to releasing a raw aggregate.  We remark that there is a trivial rule that sufficiently determines non-membership in this special case.

\textbf{Rule ($k=0$):} If $u^*$ visits ROI $s$ during epoch $t$ and no users in the aggregation group $\Users$ visit $(s,t)$, then $u^*$ cannot be in $\Users$, i.e.,
$$
\exists s \in \Si, t \in \T: A^\Users_{s,t}=0\wedge L^{u^*}_{s,t} = 1 \implies u^* \not \in \Users
$$

\begin{figure*}[h!]
    \centering
    \begin{subfigure}[b]{0.49\linewidth}
        \centering
        \includegraphics[width =0.7\linewidth, page=1]{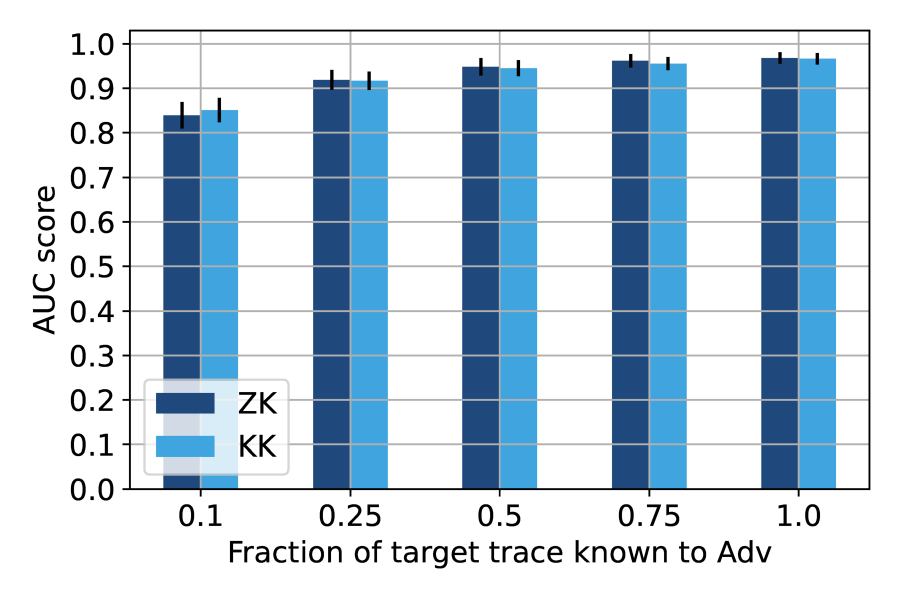}
        \caption{CDR dataset} 
        \label{fig:frac_targ_trace_CDR}
    \end{subfigure}%
    \hfill 
    \begin{subfigure}[b]{0.49\linewidth}
        \centering
        \includegraphics[width=0.7\textwidth, page=1]{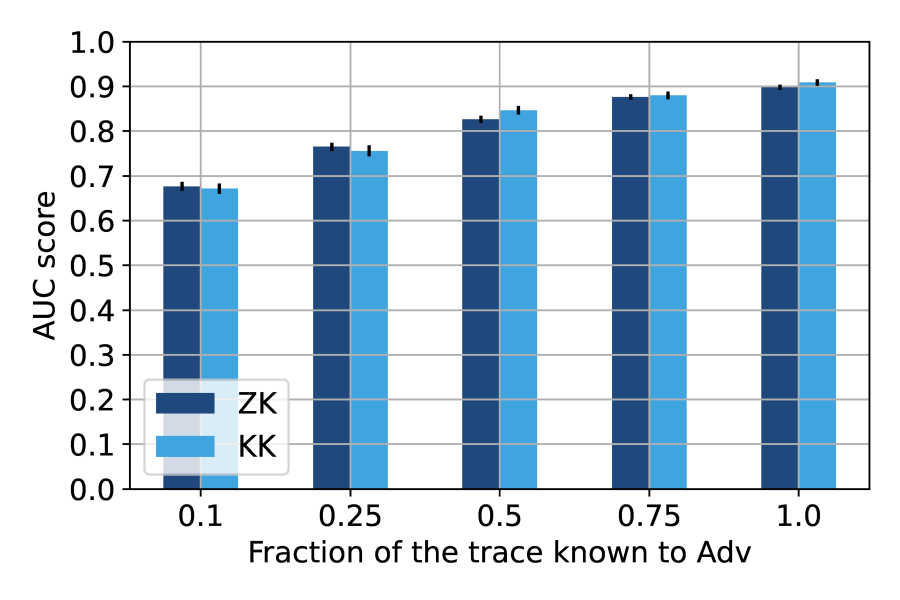}
        \caption{Milan }
        \label{fig:frac_targ_trace_Milan}
    \end{subfigure}
    \caption{Mean AUC scores with standard error for ZK and KK on size $1000$ aggregates with $\varepsilon=1$ event DP protection and $k=1$ suppression for varying fractions of the known target trace.}
    \label{fig:KK_ZK_frac_target_trace}
\end{figure*}

We therefore incorporate this rule when $k=0$, such that both MIAs first check if the released aggregate elicits the contradiction. If so, we immediately predict \textit{OUT}. Otherwise, we train, validate, and test the classifier as usual. The rule is invalid for $k>0$, since it would predict \textit{OUT} whenever $u^*$ has a visit to a suppressed entry.

\textbf{Results. \indent}
Figure \ref{fig:bucket_ZK_KK} shows that membership inference is a trivial task when applied to raw aggregates ($k=0$).
Both ZK and KK achieve near perfect AUC ($\ge 0.99$) on both datasets. This implies that aggregation is not an effective privacy mechanism in itself to protect high-dimensional location data from MIAs by weak or strong adversaries.

The results for $k>0$ reveal two general patterns. First, ZK compares closely with KK across different levels of SSC. On the CDR dataset, ZK's AUC stays within $0.02$ of KK for each $k$. We observe slightly worse results on the Milan dataset, but ZK still stays within $0.9$ AUC of KK for each $k$. Second, there is a monotonic decrease in performance when the threshold $k$ is increased. For $k=5$, the AUC is always less than $0.55$. This follows from the fact that suppression reduces the amount of available information, and $99\%$ of all nonzero entries are suppressed by $k=5$, as shown in Figure \ref{fig:bucket_info_loss} of the Appendix. Therefore, although SSC eventually mitigates the MIAs, it may come at the cost of destroying virtually all utility. 

Both MIAs perform worse on the Milan dataset compared to the CDR dataset. This is expected because we observe ~$6$ times less data per target in the Milan dataset, as shown in Figure \ref{fig:target_histogram}. Indeed, people generally tweet less often than they text and call. ZK MIA is more affected by dataset sparsity, given its dependence on marginal distribution estimates, which become less reliable in sparser datasets. 

\subsection{Against $\varepsilon$-DP Noise Addition}
\label{exp:DP}

Informed by practical applications of DP ~\cite{desfontaines2021list}, we consider event level and user-day level to be the privacy units, and we vary the privacy budget $\varepsilon \in \{0.1, 0.5, 1.0, 5.0, 10.0\}$ for each unit.

\textbf{Event Level DP.} \indent
An event is equivalent  to a visit by a user to $(s,t) \in \Si \times \T$. 
To offer privacy protection over an event, we set the global sensitivity $\Delta=1$. $\epsilon$-DP is then ensured by adding $Lap(\frac{1}{\epsilon})$ noise to each count in the aggregate matrix.

\textbf{User-day Level DP.} \indent
In order to protect each user's daily contributions without adding excessive noise, it is common to restrict user contributions prior to aggregation to achieve a smaller global sensitivity $\Delta$ ~\cite{aktay2020google, herdaugdelen2021protecting}. We analysed daily activity distributions, and had them preprocessed such that a user may only contribute up to $\Delta = 20$ visits in any given day for CDR, and $\Delta = 10$ visits in any given day for Milan.
$\epsilon$-DP at the user-day level is then ensured by adding $Lap(\frac{\Delta}{\epsilon})$ noise to each count in the aggregate matrix.

\textbf{Results. \indent}Figure \ref{fig:dp_comparison} shows that ZK MIA matches KK MIA across all tested DP settings. Indeed, ZK maintained a mean AUC within $0.06$ of KK (PS) across each of the $10$ privacy settings for both datasets. KK and ZK notably succeeded for many of the tested privacy budgets $\varepsilon$, particularly in the event level setting. Indeed, we observed $AUC \ge 0.9$ for both MIAs whenever the noise scale $\frac{\Delta}{\epsilon} \le 2$ for the CDR dataset, and $\frac{\Delta}{\epsilon} \le 1$ for the Milan dataset. These settings are in line with many real-life applications~\cite{kohli2023privacy, desfontaines2021list}. Conversely, user-day level DP with privacy budget $\varepsilon \le 1.0$ effectively reduced both MIAs to an AUC below $0.55$. We discuss the significance of these results with respect to practical mitigations in Section \ref{sec:dicussion}. 

\subsection{Partial Knowledge of the Target Trace}
\label{exp:partial_trace}

We now relax the assumption that $Adv$ knows the full target trace $L^{u^*}$. This is in line with our ZK threat model, and we expand KK MIA for this setting to be able to compare methods. To simulate a weaker adversary, we suppose that $Adv$ only knows a subset of the target $u^*$'s visits. We assume that $Adv$ only knows a random fraction $p_{u^*} \in \{0.1, 0.25, 0.5, 0.75, 1.0\}$ of the trace $L^{u^*}$. The number of retained visits is rounded up to the next integer to prevent cases where $Adv$ knows $0$ visits. For example, if a target has $4$ visits and $p_{u^*}=0.1$, then this would correspond to $Adv$ knowing $1$ visit. This partial trace is used instead of the full trace when creating \textit{IN} training and validation aggregates. The full trace $L^{u^*}$ is still used for \textit{IN} test aggregates.

We perform this experiment in the setting where the data collector applies event-level DP with $\epsilon=1$, followed by $k=1$ suppression. We choose this setting for a couple of reasons. First, to study the degradation of the MIAs with decreasing information about the target, we choose a setting where $Adv$ would succeed given the full target trace. Previous experiments revealed that $\epsilon=1$-DP at the event level and $k=1$ suppression were not effective in preventing the MIAs by themselves, as the MIAs achieved AUC $>0.97$ on the CDR dataset and AUC $>0.9$ for Milan. Second, we combine the two defense mechanisms to see if suppression has an observable mitigation effect when applied following DP noise addition. By the post-processing property of DP, this would not alter the theoretical performance bound. However, zeroing all counts $\le 1$ might add a layer of complexity that affect MIAs in practice.

\textbf{Results. \indent}
First, we note that applying $k=1$ SSC on top of $\varepsilon=1.0$ event level DP has an insignificant effect on the MIAs. For the full target trace, we continue to observe AUC $>0.97$ on the CDR dataset and AUC $>0.9$ on the Milan dataset. The only MIA with a noticeable decline was KK MIA on the Milan dataset, which dropped from $0.97$ AUC to $0.9$ AUC.

Although decreasing the fraction of the target trace known to the adversary from $1$ to $0.1$ decreases the performance of the MIAs, the corresponding degradation is relatively gradual. All AUCs are captured within a range of $0.13$ on the CDR dataset, and within a range of $0.22$ on the Milan dataset. Even the lowest observed AUC by ZK MIA on the CDR dataset ($0.84$ when $10\%$ of the target trace is known) achieves high discrimination. We note that the $50$ targets in both datasets have a wide variation in trace size, as shown in Figure \ref{fig:target_histogram} of the Appendix. For some targets, $Adv$ will only know one of the target's visits, whereas for others, they will still know dozens and be able to infer membership easily. Interestingly, we note that knowing a single visit from a target trace can still train a classifier that is better than random. For one CDR target with $9$ visits in their full trace, we observed ZK achieve an AUC of $0.660$ across $100$ aggregates when only $1$ random visit was known. Although far from perfect, we found this surprising, as it shows that even a single visit by the target can inform an MIA against a noisy aggregate over $1000$ users. 

\begin{figure*}[h!]
    \centering
    \begin{subfigure}[b]{0.49\linewidth}
        \centering
        \includegraphics[width = 0.7\linewidth, page=1]{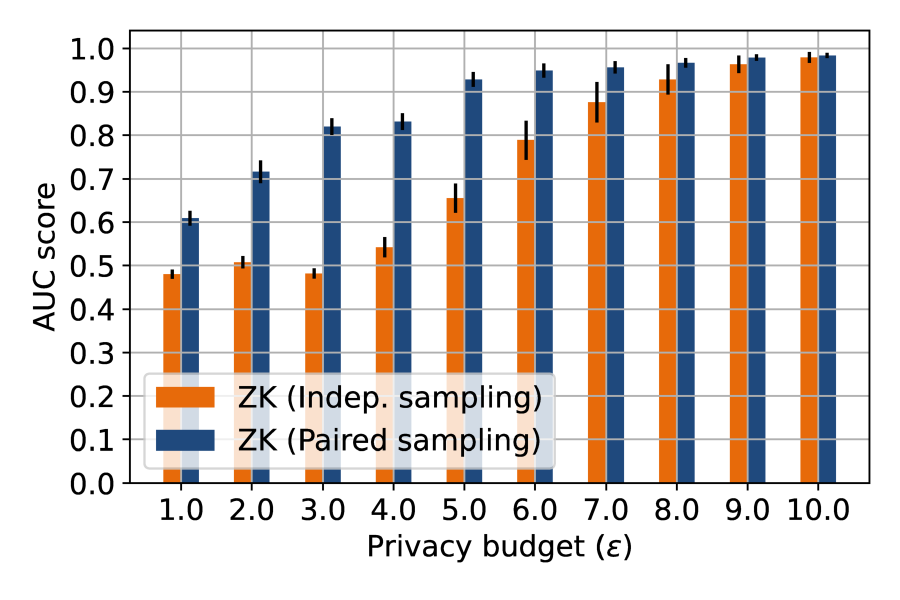}
        \caption{Zero Auxiliary Knowledge}
        \label{fig:ZK_paired_ind_DP}
    \end{subfigure}%
    \hfill 
    \begin{subfigure}[b]{0.49\linewidth}
        \centering
        \includegraphics[width =0.7\linewidth]{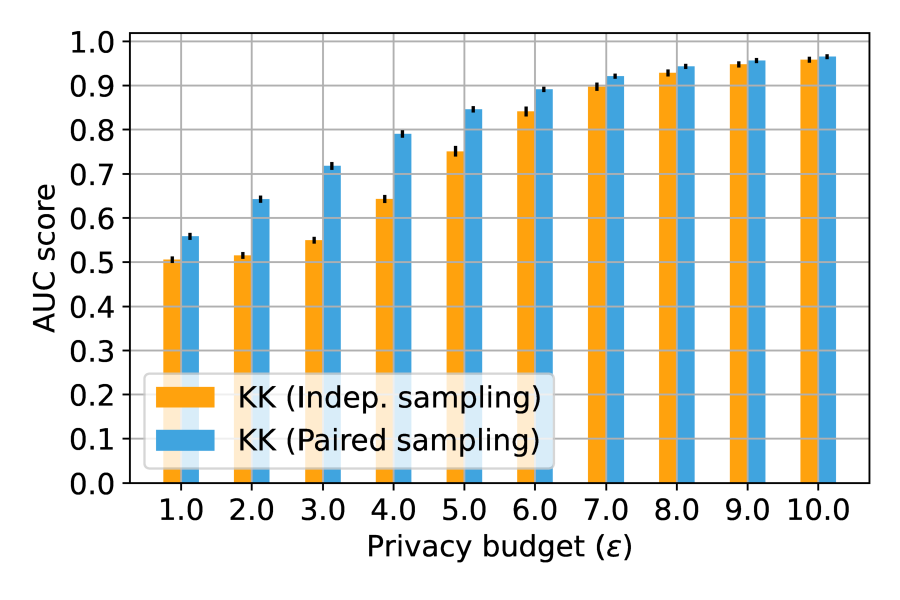}
        \caption{Knock-Knock}
    \label{fig:KK_paired_ind_DP}
    \end{subfigure}
    \caption{Mean AUC scores with standard error for ZK and KK on size $1000$ aggregates from the Milan dataset across different privacy budgets $\varepsilon$ under user-day level $\varepsilon$-DP, for MIAs using independent sampling vs. paired sampling.}
    \label{fig:KK_ZK_ps_is}
\end{figure*}

\subsection{Paired Sampling vs. Independent Sampling}
\label{exp: paired_sampling}

In this section, we study the performance of KK MIA and ZK MIA when we vary the sampling technique used for creating their training aggregates, i.e., paired sampling (PS) or independent sampling (IS). To test this, we consider the implementation of user-day $\varepsilon-DP$ on the Milan dataset across the privacy budgets $\varepsilon=1,2,3,..., 10$ for all four possible MIAs: KK (PS), KK (IS), ZK (PS), ZK (IS).

\textbf{Results. \indent}
From Figure \ref{fig:KK_ZK_ps_is}, we see that the paired sampling MIA always outperforms its independent sampling equivalent across all privacy budgets $\varepsilon=1, 2,3, ..., 10$. Paired sampling provides the largest boost when the inference task is challenging but not intractable. In particular, we notice a few striking examples of ZK (PS) drastically outperforming ZK (IS) in the middle of the graph. For $\epsilon=4$, ZK (IS) is basically random ($AUC=0.54$), yet simply switching to paired sampling enables the classifier to achieve an AUC of $0.83$. 

The improvement achieved by switching from independent sampling to paired sampling is less significant for KK in this experiment. This suggests that using training aggregates sampled from a reference dataset of real traces may introduce less randomness to the membership classifier's decision boundary, compared to when we use a synthetic reference. This is intuitive because of ZK's probabilistic generation method, which relies on sampling from three different estimated distributions. However, the MIAs have indistinguishable performance when both attacks use paired sampling, with the difference in AUC always staying within $0.02$ of one another across the $10$ privacy settings. This suggests that paired sampling effectively eliminates the noise contributed by coincidental patterns in random entries, and enables the membership classifier to form a suitable decision boundary.

\section{Discussion}
\label{sec:dicussion}
We first provide a critical analysis of the experimental results, followed by a discussion of mitigation strategies and their practicality. We then consider limitations in our methods and evaluations, before discussing how our methodology may be generalized to MIAs beyond the setting of aggregate location data.

\subsection{Analysis of Results: Practical Risk of MIA}

In Section \ref{sec:results}, we observed that our ZK MIA achieves approximately the same performance as the KK MIA across all experiments, and that both MIAs performed effectively across a range of common privacy settings. This has several important implications.

First, the ZK MIA significantly increases the attack surface of aggregate location data, since no auxiliary dataset is needed for the MIA. Although previous MIAs on aggregate location data have been successful, the strong assumption of a large auxiliary dataset prevents these attackers from attempting the MIA in most real-life cases. The auxiliary dataset comprises sensitive user-level information collected from the same dataset that is being aggregated. However, aggregation is applied to prevent the release of personal information. Moreover, $Adv$ would be restricted by the size of their auxiliary dataset, since they would not be able to perform MIAs on aggregates computed over more users than there are in their reference. In contrast, we demonstrated that our Zero Auxiliary Knowledge $Adv$ can create an arbitrary number of synthetic traces upon seeing the released aggregate, without an auxiliary dataset. This offers $Adv$ the flexibility to attack aggregates of any size. Section~\ref{appendix:varyingm} in the Appendix shows results for KK and ZK MIA across aggregates of size $m=100, 500, 1000, 2000, 3000$. Figure \ref{fig:n_syn} of the Appendix also shows that the ZK $Adv$ can boost their own performance up to diminishing marginal returns, simply by generating more traces.

We have also shown that MIAs on aggregate location data are more powerful than previously known. By incorporating paired sampling for training, we have demonstrated more effective MIA results on size $1000$ aggregates than previously reported \cite{pyrgelis2020measuring}, particularly when protected by differential privacy. Our results therefore demonstrate that MIAs on aggregate location data are easily performed without auxiliary data, more effective than previously believed, and that common privacy measures fail to protect against the risk.

\subsection{Proposed Mitigations}
Our results show that aggregated location data requires more stringent privacy safeguards to protect against MIAs. This is an inherently challenging task because our ZK MIA was able to succeed by using the aggregate to estimate where the population moves (space marginal), when the population moves (time marginal), and how frequently the population moves (activity marginal). However, location aggregates naturally leak this information. In fact, much of its utility is derived from these marginal statistics. Therefore, while one can mitigate the ZK MIA by perturbing the aggregate to the point where the basic mobility patterns of its population are unrecoverable, doing so may also destroy the aggregates' utility.

In light of these results, we advise that data practitioners be mindful of the parameters that they select for $\epsilon-$DP, because DP does not guarantee sufficient protection from MIAs if the parameters are chosen too loosely. Since data practitioners often prioritize utility, it is common to pick more relaxed parameters for the privacy unit (e.g. event or user-day instead of user) and budget $\varepsilon$ (e.g. $\epsilon > 1$). For example, \citet{kohli2023privacy} studied $\epsilon \in \{0.1, 0.5, 1.0\}$ at the event level in the context of aggregate O-D mobility matrices, Facebook used $\epsilon = 0.45$ at the event level when collecting data about URLs shared on the site, and Apple uses $\epsilon$ between $2$ and $16$ at the user-day level when collecting IOS data \cite{desfontaines2021list}. Recall that ZK and KK achieved $AUC \ge 0.9$ on the CDR dataset whenever the noise scale $\frac{\Delta}{\epsilon} $ was $2$ or less. This corresponds to $\epsilon \ge 0.5$ for event level DP and $\epsilon \ge 10$ for user-day level DP with up to $\Delta =20$ daily visits. ZK and KK therefore both achieved high discrimination on privacy settings that are in line with many real-life applications. 

However, we do observe DP mitigating the MIAs when we pick sufficiently strict parameters. For example, no MIA achieved better than random performance for $\epsilon = 0.5$ in the user-day setting. We also note that we did not evaluate using the user level setting, which would achieve the strongest privacy protection. We note that the suitability of privacy parameters depends on the desired utility and sensitivity of the dataset. A stricter parameter choice is particularly relevant if the aggregate is publicly released and/or pertains to sensitive data.  

Although $\varepsilon$-DP always offers privacy guarantees, our experimental results emphasize the importance of picking appropriate parameters. In particular, we observed that event-level DP was largely ineffective in preventing MIAs from both strong and weak adversaries. We instead encourage the use of user-day or user-level DP with carefully selected privacy budgets $\varepsilon$ to mitigate the practical threat of an MIA.

\subsection{Limitations} 
We have so far taken the Knock-Knock MIA to refer to the Subset of Locations setting~\cite{pyrgelis2017knock}. We now address why we do not consider the Knock-Knock Participation in Past Groups~\cite{pyrgelis2017knock} threat model in this paper. Under the Participation in Past Groups setting, the adversary has access to a set of past aggregates $\{\overline{A}^{\tilde{\Users_1}}, ..., \overline{A}^{\tilde{\Users_N}}\}$, collected over the same ROIs $\Si$ as the released aggregate $\overline{A^\Users}$. Moreover, $Adv$ is assumed to know the membership status of the target $u^*$ in each of these aggregates. That is, $Adv$ knows whether or not $u^* \in \Users_i$ for all $i=1, ..., N$. This last assumption is crucial because $Adv$ directly uses $\{\overline{A}^{\tilde{\Users_1}}, ..., \overline{A}^{\tilde{\Users_N}}\}$ as their training data for the membership classifier in this setting.
This is unrealistic for multiple reasons. First, to train an effective membership classifier, there would need to be hundreds of labeled aggregates to have sufficient training data. More importantly, there would be no reason for the membership status of an individual within an aggregate to be released in practice. We argue that the only plausible scenario in which $Adv$ would know the membership status of each aggregate is if they created the aggregates themselves. This reduces to the Subset of Locations setting that we have assumed in this paper.

In terms of limitations for our ZK MIA, recall that the Delaunay triangulation of the ROIs, $DT(\Si)$, is the only non-probabilistic parameter used to generate synthetic traces for the ZK MIA. The triangulation only depends on the locations of the ROIs, which we have so far assumed to be shared as part of the aggregate release (Section \ref{subsec: threat_model}). We believe this to be a realistic assumption, as omitting the locations of the ROIs would strongly diminish the utility of aggregate location data. Nonetheless, there might exist cases where released location aggregates do not relay the positions of the ROIs. For example, Google binned ROIs into categories, ex. restaurants, parks, and hospitals, when publicly releasing their mobility report during COVID \url{https://www.google.com/covid19/mobility/}. In this setting, the adversary would proceed without knowing where the ROIs are situated with respect to one another. The privacy risk under this setting is not known, and we identify it as an area of future research. Similarly, there might exist cases where the adversary knows the ROIs that were visited by the target (ex. home and work), but not the visitation times. We show in Appendix ~\ref{appendix:onlyrois} that only knowing the visited ROIs substantially reduces the effectiveness of both MIAs.

ZK MIA also requires that we estimate statistical parameters from the released aggregate. It may be difficult to estimate these precisely if the aggregate size is small or if the collected location data is not regular. However, we still observe strong performance by ZK MIA on both datasets for small aggregate sizes (see Appendix~\ref{appendix:varyingm}).

Furthermore, aggregate location data collected over large metropolitan populations are known to obey high regularity across different cities and time periods. These patterns include log-normal activity distributions~\cite{schneider2013unravelling, seshadri2008mobile, farzanehfar2021risk} and periodic "circadian rhythm" time marginals~\cite{csaji2013exploring, song2010limits, seshadri2008mobile, farzanehfar2021risk}. This suggests that our statistical parameter estimation should be highly transferable across sufficiently regular datasets. However, we acknowledge that there are scenarios where the observed population is not regular (e.g. taxi drivers). 

\subsection{Generalizations to MIAs in other Settings}
In this paper, we have proposed a new methodology to perform membership inference attacks on aggregate data, by training the attack on synthetic records, generated from the released aggregate. We believe that this approach can be adapted for MIAs in settings beyond aggregate location data. Our methodology can be broken down into two main steps: 1) extracting noise-less global statistics from the released aggregate, 2) use these statistics to create individual-level records to train the MIA. 

In the setting of location data, the relevant statistics pertain to the mobility trends of large-scale human populations \cite{schneider2013unravelling, seshadri2008mobile, farzanehfar2021risk,csaji2013exploring, song2010limits, seshadri2008mobile}, and individual location ~\cite{kulkarni2017generating,kulkarni2018generative,ouyang2018non, karagiannis2007power, lee2009slaw} which have both been well established in the literature. This facilitates both steps of our methodology, as we 
know in advance what location data should look like at both the global and individual level. 

Although the trends will be distinct from aggregate location data, aggregate releases for other types of data will generally reveal global statistics. For instance, categorical tabular data is modeled by discrete random variables, whereas location data is modeled by continuous random variables, and approximated by high-dimensional discrete data. Our methods for denoising and debiasing statistics from differentially private and suppressed aggregates are however not specific to location data, and should generalize to other data releases. Regarding the second step, using the statistics to create individual records for training, the probabilistic method used for our ZK MIA, drawing from the Delaunay triangulation and the relevant marginal distributions, is partially specific to location data. One would thus need to carefully consider the statistical properties of the type of data to create high quality individual records.

\section{Conclusion}
\label{sec:conclusion}
Aggregate location data is widely shared and used by governments~\cite{NYT2021DIASurveillance, hope2021millions, oli2021canada}, companies~\cite{apple2017learning, aktay2020google, herdaugdelen2021protecting}, and researchers~\cite{trasberg2023spatial, jeffrey2020anonymised, kohli2023privacy} because of its insights into human behaviour and its presumed security against reidentification.

In this paper, we demonstrated that aggregate location data is susceptible to MIAs by realistic adversaries, who only know some of their target's location history. With ZK MIA, we introduced the first MIA on aggregate location data that does not require an auxiliary dataset. We accomplished this by generating appropriate synthetic traces, using statistics that are estimated from the released aggregate. We also equipoed our parameter estimation with techniques that automatically correct for bias and noise from popular privacy mechanisms like suppression of small counts and $\varepsilon$-DP noise.

We then showed that MIAs on aggregate location data are significantly improved by incorporating a paired sampling technique, which helps isolate the effect of the target trace within a high dimensional aggregate. Hence, the vulnerability of aggregate location data is further heightened by these improved attacks.

Our evaluations over two large datasets demonstrate that, despite the absence of an auxiliary dataset, ZK MIA performs as well as the state-of-the-art KK MIA, with both MIAs achieving high discrimination when commonly used privacy settings are applied. ZK MIA remains effective in realistic privacy settings, even when only a small fraction ($10\%$) of the target trace is known. These results emphasize the need for strict differential privacy guarantees on released aggregate location data.

Taken together, our findings show that membership inference attacks are not merely a theoretical privacy threat posed by unrealistically strong adversaries, but also a realistic threat to contend with in practice.

\begin{acks}
Ana-Maria Cretu did most of her work while she was at Imperial College London and was partially funded by the Agence Française de Développement via the Flowminder Foundation. 

The authors would like to sincerely thank the Flowminder team for their support with this work, in particular Galina Veres and James Harrison for their help in designing a secure data-sharing pipeline to test the MIAs on the CDR dataset. The authors would like to further acknowledge Cyril Miras for his early work on MIAs on aggregate location data, including the implementation of the baseline rule MIA. Finally, the authors would like to thank the anonymous reviewers and shepherd for their feedback on the paper.
\end{acks}

\bibliographystyle{apalike} 
\bibliography{references.bib}

\begin{thebibliography}{}

\bibitem[wp2, 2014]{wp2014}
 (2014).
\newblock Article 29 data protection working party. opinion 05/2014 on
  anonymisation techniques.
\newblock
  \url{https://ec.europa.eu/justice/article-29/documentation/opinion-recommendation/files/2014/wp216_en.pdf}.

\bibitem[flo, 2024]{flowminder}
 (2024).
\newblock Flowminder website.
\newblock \url{https://www.flowminder.org/}.

\bibitem[Aktay et~al., 2020]{aktay2020google}
Aktay, A., Bavadekar, S., Cossoul, G., Davis, J., Desfontaines, D., Fabrikant,
  A., Gabrilovich, E., Gadepalli, K., Gipson, B., Guevara, M., et~al. (2020).
\newblock Google covid-19 community mobility reports: anonymization process
  description (version 1.1).
\newblock {\em arXiv preprint arXiv:2004.04145}.

\bibitem[Apple, 2017]{apple2017learning}
Apple, D. (2017).
\newblock Learning with privacy at scale.
\newblock {\em Apple Machine Learning Journal}, 1(8).

\bibitem[Barlacchi et~al., 2015]{barlacchi2015multi}
Barlacchi, G., De~Nadai, M., Larcher, R., Casella, A., Chitic, C., Torrisi, G.,
  Antonelli, F., Vespignani, A., Pentland, A., and Lepri, B. (2015).
\newblock A multi-source dataset of urban life in the city of milan and the
  province of trentino.
\newblock {\em Scientific data}, 2(1):1--15.

\bibitem[Bauer and Bindschaedler, 2020]{bauer2020towards}
Bauer, L.~A. and Bindschaedler, V. (2020).
\newblock Towards realistic membership inferences: The case of survey data.
\newblock In {\em Annual Computer Security Applications Conference}, pages
  116--128.

\bibitem[Bishop and Nasrabadi, 2006]{bishop2006pattern}
Bishop, C.~M. and Nasrabadi, N.~M. (2006).
\newblock {\em Pattern recognition and machine learning}, volume~4.
\newblock Springer.

\bibitem[Boorstein and Kelly, 2023]{boorstein2023colorado}
Boorstein, M. and Kelly, H. (2023).
\newblock Colorado catholic group bought app data that tracked gay priests.
\newblock {\em The Washington Post}.

\bibitem[Chen et~al., 2009]{chen2009privacy}
Chen, B.-C., Kifer, D., LeFevre, K., Machanavajjhala, A., et~al. (2009).
\newblock Privacy-preserving data publishing.
\newblock {\em Foundations and Trends{\textregistered} in Databases},
  2(1--2):1--167.

\bibitem[Cre{\c{t}}u et~al., 2021]{crectu2021correlation}
Cre{\c{t}}u, A.-M., Gu{\'e}pin, F., and de~Montjoye, Y.-A. (2021).
\newblock Correlation inference attacks against machine learning models.
\newblock {\em arXiv preprint arXiv:2112.08806}.

\bibitem[Cretu et~al., 2022]{cretu2022querysnout}
Cretu, A.-M., Houssiau, F., Cully, A., and de~Montjoye, Y.-A. (2022).
\newblock Querysnout: Automating the discovery of attribute inference attacks
  against query-based systems.
\newblock In {\em Proceedings of the 2022 ACM SIGSAC Conference on Computer and
  Communications Security}, pages 623--637.

\bibitem[Cs{\'a}ji et~al., 2013]{csaji2013exploring}
Cs{\'a}ji, B.~C., Browet, A., Traag, V.~A., Delvenne, J.-C., Huens, E.,
  Van~Dooren, P., Smoreda, Z., and Blondel, V.~D. (2013).
\newblock Exploring the mobility of mobile phone users.
\newblock {\em Physica A: statistical mechanics and its applications},
  392(6):1459--1473.

\bibitem[de~Montjoye et~al., 2018]{de2018privacy}
de~Montjoye, Y.-A., Gambs, S., Blondel, V., Canright, G., De~Cordes, N.,
  Deletaille, S., Eng{\o}-Monsen, K., Garcia-Herranz, M., Kendall, J., Kerry,
  C., et~al. (2018).
\newblock On the privacy-conscientious use of mobile phone data.
\newblock {\em Scientific data}, 5(1):1--6.

\bibitem[de~Montjoye et~al., 2013]{de2013unique}
de~Montjoye, Y.-A., Hidalgo, C.~A., Verleysen, M., and Blondel, V.~D. (2013).
\newblock Unique in the crowd: The privacy bounds of human mobility.
\newblock {\em Scientific reports}, 3(1):1--5.

\bibitem[Delaunay et~al., 1934]{delaunay1934sphere}
Delaunay, B. et~al. (1934).
\newblock Sur la sphere vide.
\newblock {\em Izv. Akad. Nauk SSSR, Otdelenie Matematicheskii i Estestvennyka
  Nauk}, 7(793-800):1--2.

\bibitem[Desfontaines, 2021]{desfontaines2021list}
Desfontaines, D. (2021).
\newblock A list of real-world uses of differential privacy.

\bibitem[Dwork et~al., 2019]{dwork2019differential}
Dwork, C., Kohli, N., and Mulligan, D. (2019).
\newblock Differential privacy in practice: Expose your epsilons!
\newblock {\em Journal of Privacy and Confidentiality}, 9(2).

\bibitem[Dwork et~al., 2006]{dwork2006calibrating}
Dwork, C., McSherry, F., Nissim, K., and Smith, A. (2006).
\newblock Calibrating noise to sensitivity in private data analysis.
\newblock In {\em Theory of Cryptography: Third Theory of Cryptography
  Conference, TCC 2006, New York, NY, USA, March 4-7, 2006. Proceedings 3},
  pages 265--284. Springer.

\bibitem[Farzanehfar et~al., 2021]{farzanehfar2021risk}
Farzanehfar, A., Houssiau, F., and de~Montjoye, Y.-A. (2021).
\newblock The risk of re-identification remains high even in country-scale
  location datasets.
\newblock {\em Patterns}, 2(3):100204.

\bibitem[Gadotti et~al., 2019]{gadotti2019signal}
Gadotti, A., Houssiau, F., Rocher, L., Livshits, B., and De~Montjoye, Y.-A.
  (2019).
\newblock When the signal is in the noise: Exploiting diffix's sticky noise.
\newblock In {\em 28th USENIX Security Symposium (USENIX Security 19)}, pages
  1081--1098.

\bibitem[Ge and Fukuda, 2016]{ge2016updating}
Ge, Q. and Fukuda, D. (2016).
\newblock Updating origin--destination matrices with aggregated data of gps
  traces.
\newblock {\em Transportation Research Part C: Emerging Technologies},
  69:291--312.

\bibitem[Georgiadou et~al., 2019]{georgiadou2019location}
Georgiadou, Y., de~By, R.~A., and Kounadi, O. (2019).
\newblock Location privacy in the wake of the gdpr.
\newblock {\em ISPRS international journal of geo-information}, 8(3):157.

\bibitem[Grantz et~al., 2020]{grantz2020use}
Grantz, K.~H., Meredith, H.~R., Cummings, D.~A., Metcalf, C. J.~E., Grenfell,
  B.~T., Giles, J.~R., Mehta, S., Solomon, S., Labrique, A., Kishore, N.,
  et~al. (2020).
\newblock The use of mobile phone data to inform analysis of covid-19 pandemic
  epidemiology.
\newblock {\em Nature communications}, 11(1):4961.

\bibitem[Gu{\'e}pin et~al., 2023]{guepin2023synthetic}
Gu{\'e}pin, F., Meeus, M., Cretu, A.-M., and de~Montjoye, Y.-A. (2023).
\newblock Synthetic is all you need: removing the auxiliary data assumption for
  membership inference attacks against synthetic data.
\newblock {\em arXiv preprint arXiv:2307.01701}.

\bibitem[Hara and Yamaguchi, 2021]{hara2021japanese}
Hara, Y. and Yamaguchi, H. (2021).
\newblock Japanese travel behavior trends and change under covid-19
  state-of-emergency declaration: Nationwide observation by mobile phone
  location data.
\newblock {\em Transportation Research Interdisciplinary Perspectives},
  9:100288.

\bibitem[Herda{\u{g}}delen and Dow, 2021]{herdaugdelen2021protecting}
Herda{\u{g}}delen, A. and Dow, A. (2021).
\newblock Protecting privacy in facebook mobility data during the covid-19
  response (2020).
\newblock {\em URL https://research. fb.
  com/blog/2020/06/protecting-privacy-in-facebook-mobility-data-during-the-covid-19-response}.

\bibitem[Holmes et~al., 2013]{holmes2013mobile}
Holmes, A., Byrne, A., and Rowley, J. (2013).
\newblock Mobile shopping behaviour: insights into attitudes, shopping process
  involvement and location.
\newblock {\em International Journal of Retail \& Distribution Management},
  42(1):25--39.

\bibitem[Homer et~al., 2008]{homer2008resolving}
Homer, N., Szelinger, S., Redman, M., Duggan, D., Tembe, W., Muehling, J.,
  Pearson, J.~V., Stephan, D.~A., Nelson, S.~F., and Craig, D.~W. (2008).
\newblock Resolving individuals contributing trace amounts of dna to highly
  complex mixtures using high-density snp genotyping microarrays.
\newblock {\em PLoS genetics}, 4(8):e1000167.

\bibitem[Hope, 2021]{hope2021millions}
Hope, C. (2021).
\newblock Millions 'unwittingly tracked' by phone after vaccination to see if
  movements changed.

\bibitem[Houssiau et~al., 2022]{houssiau2022tapas}
Houssiau, F., Jordon, J., Cohen, S.~N., Daniel, O., Elliott, A., Geddes, J.,
  Mole, C., Rangel-Smith, C., and Szpruch, L. (2022).
\newblock Tapas: A toolbox for adversarial privacy auditing of synthetic data.
\newblock {\em arXiv preprint arXiv:2211.06550}.

\bibitem[Humphries et~al., 2023]{humphries2023investigating}
Humphries, T., Oya, S., Tulloch, L., Rafuse, M., Goldberg, I., Hengartner, U.,
  and Kerschbaum, F. (2023).
\newblock Investigating membership inference attacks under data dependencies.
\newblock In {\em 2023 IEEE 36th Computer Security Foundations Symposium
  (CSF)}, pages 473--488. IEEE.

\bibitem[Jagielski et~al., 2020]{jagielski2020auditing}
Jagielski, M., Ullman, J., and Oprea, A. (2020).
\newblock Auditing differentially private machine learning: How private is
  private sgd?
\newblock {\em Advances in Neural Information Processing Systems},
  33:22205--22216.

\bibitem[Jahromi et~al., 2016]{jahromi2016simulating}
Jahromi, K.~K., Zignani, M., Gaito, S., and Rossi, G.~P. (2016).
\newblock Simulating human mobility patterns in urban areas.
\newblock {\em Simulation Modelling Practice and Theory}, 62:137--156.

\bibitem[Jayaraman and Evans, 2019]{jayaraman2019evaluating}
Jayaraman, B. and Evans, D. (2019).
\newblock Evaluating differentially private machine learning in practice.
\newblock In {\em 28th USENIX Security Symposium (USENIX Security 19)}, pages
  1895--1912.

\bibitem[Jeffrey et~al., 2020]{jeffrey2020anonymised}
Jeffrey, B., Walters, C.~E., Ainslie, K.~E., Eales, O., Ciavarella, C., Bhatia,
  S., Hayes, S., Baguelin, M., Boonyasiri, A., Brazeau, N.~F., et~al. (2020).
\newblock Anonymised and aggregated crowd level mobility data from mobile
  phones suggests that initial compliance with covid-19 social distancing
  interventions was high and geographically consistent across the uk.
\newblock {\em Wellcome Open Research}, 5.

\bibitem[Kakakhel, 2022]{Kakakhel2022Optimising}
Kakakhel, S. (2022).
\newblock Optimising urban planning with location intelligence.
\newblock Quadrant Blog.
\newblock Accessed: 2024-03-07.

\bibitem[Karagiannis et~al., 2007]{karagiannis2007power}
Karagiannis, T., Le~Boudec, J.-Y., and Vojnovi{\'c}, M. (2007).
\newblock Power law and exponential decay of inter contact times between mobile
  devices.
\newblock In {\em Proceedings of the 13th annual ACM international conference
  on Mobile computing and networking}, pages 183--194.

\bibitem[Kohli et~al., 2023]{kohli2023privacy}
Kohli, N., Aiken, E., and Blumenstock, J. (2023).
\newblock Privacy guarantees for personal mobility data in humanitarian
  response.
\newblock {\em arXiv preprint arXiv:2306.09471}.

\bibitem[Kulkarni and Garbinato, 2017]{kulkarni2017generating}
Kulkarni, V. and Garbinato, B. (2017).
\newblock Generating synthetic mobility traffic using rnns.
\newblock In {\em Proceedings of the 1st Workshop on Artificial Intelligence
  and Deep Learning for Geographic Knowledge Discovery}, pages 1--4.

\bibitem[Kulkarni et~al., 2018]{kulkarni2018generative}
Kulkarni, V., Tagasovska, N., Vatter, T., and Garbinato, B. (2018).
\newblock Generative models for simulating mobility trajectories.
\newblock {\em arXiv preprint arXiv:1811.12801}.

\bibitem[Lee et~al., 2009]{lee2009slaw}
Lee, K., Hong, S., Kim, S.~J., Rhee, I., and Chong, S. (2009).
\newblock Slaw: A new mobility model for human walks.
\newblock In {\em IEEE INFOCOM 2009}, pages 855--863. IEEE.

\bibitem[Li et~al., 2013]{li2013membership}
Li, N., Qardaji, W., Su, D., Wu, Y., and Yang, W. (2013).
\newblock Membership privacy: A unifying framework for privacy definitions.
\newblock In {\em Proceedings of the 2013 ACM SIGSAC conference on Computer \&
  communications security}, pages 889--900.

\bibitem[Mart{\'\i}nez-Durive et~al., 2023]{martinez2023netmob23}
Mart{\'\i}nez-Durive, O.~E., Mishra, S., Ziemlicki, C., Rubrichi, S., Smoreda,
  Z., and Fiore, M. (2023).
\newblock The netmob23 dataset: A high-resolution multi-region service-level
  mobile data traffic cartography.
\newblock {\em arXiv preprint arXiv:2305.06933}.

\bibitem[Meeus et~al., 2023]{meeus2023achilles}
Meeus, M., Guepin, F., Cre{\c{t}}u, A.-M., and de~Montjoye, Y.-A. (2023).
\newblock Achilles’ heels: vulnerable record identification in synthetic data
  publishing.
\newblock In {\em European Symposium on Research in Computer Security}, pages
  380--399. Springer.

\bibitem[Miguel~Alonso and Richard, 2004]{miguel2004tempo}
Miguel~Alonso, B.~D. and Richard, G. (2004).
\newblock Tempo and beat estimation of musical signals.
\newblock In {\em Proceedings of the International Conference on Music
  Information Retrieval (ISMIR), Barcelona, Spain}.

\bibitem[Morgan and Lovelace, 2021]{morgan2021travel}
Morgan, M. and Lovelace, R. (2021).
\newblock Travel flow aggregation: Nationally scalable methods for interactive
  and online visualisation of transport behaviour at the road network level.
\newblock {\em Environment and Planning B: Urban Analytics and City Science},
  48(6):1684--1696.

\bibitem[M\"uller, 2015]{FMP_C3S1_LogCompression}
M\"uller, M. (2015).
\newblock Logarithmic compression.
\newblock
  \url{https://www.audiolabs-erlangen.de/resources/MIR/FMP/C3/C3S1_LogCompression.html}.

\bibitem[Nasr et~al., 2019]{nasr2019comprehensive}
Nasr, M., Shokri, R., and Houmansadr, A. (2019).
\newblock Comprehensive privacy analysis of deep learning: Passive and active
  white-box inference attacks against centralized and federated learning.
\newblock In {\em 2019 IEEE symposium on security and privacy (SP)}, pages
  739--753. IEEE.

\bibitem[Nasr et~al., 2021]{nasr2021adversary}
Nasr, M., Songi, S., Thakurta, A., Papernot, N., and Carlin, N. (2021).
\newblock Adversary instantiation: Lower bounds for differentially private
  machine learning.
\newblock In {\em 2021 IEEE Symposium on security and privacy (SP)}, pages
  866--882. IEEE.

\bibitem[O2, 2019]{o2_transport_smart_steps_2019}
O2 (2019).
\newblock O2 transport smart steps product sheet.
\newblock
  \url{https://static-www.o2.co.uk/sites/default/files/2019-04/o2-transport-smart-steps-product-sheet.pdf}.
\newblock [Online].

\bibitem[Oehmichen et~al., 2019]{oehmichen2019opal}
Oehmichen, A., Jain, S., Gadotti, A., and de~Montjoye, Y.-A. (2019).
\newblock Opal: High performance platform for large-scale privacy-preserving
  location data analytics.
\newblock In {\em 2019 IEEE International Conference on Big Data (Big Data)},
  pages 1332--1342. IEEE.

\bibitem[{Office of the Privacy Commissioner of Canada},
  2023]{OPC2023Investigation}
{Office of the Privacy Commissioner of Canada} (2023).
\newblock Investigation into the collection and use of de-identified mobility
  data in the course of the covid-19 pandemic.
\newblock Accessed: 2023-09-14.

\bibitem[Oli, 2021]{oli2021canada}
Oli, S. (2021).
\newblock Canada's public health agency admits it tracked 33 million mobile
  devices during lockdown.
\newblock {\em National Post}, 24.

\bibitem[Ouyang et~al., 2018]{ouyang2018non}
Ouyang, K., Shokri, R., Rosenblum, D.~S., and Yang, W. (2018).
\newblock A non-parametric generative model for human trajectories.
\newblock In {\em IJCAI}, volume~18, pages 3812--3817.

\bibitem[Popa et~al., 2011]{popa2011privacy}
Popa, R.~A., Blumberg, A.~J., Balakrishnan, H., and Li, F.~H. (2011).
\newblock Privacy and accountability for location-based aggregate statistics.
\newblock In {\em Proceedings of the 18th ACM conference on Computer and
  communications security}, pages 653--666.

\bibitem[Precisely, 2024]{precisely_placeiq_movement}
Precisely (2024).
\newblock Placeiq movement.
\newblock
  \url{https://www.precisely.com/product/precisely-placeiq/placeiq-movement}.
\newblock Accessed: 2024-02-13.

\bibitem[Pyrgelis et~al., 2017]{pyrgelis2017knock}
Pyrgelis, A., Troncoso, C., and De~Cristofaro, E. (2017).
\newblock Knock knock, who's there? membership inference on aggregate location
  data.
\newblock {\em arXiv preprint arXiv:1708.06145}.

\bibitem[Pyrgelis et~al., 2020]{pyrgelis2020measuring}
Pyrgelis, A., Troncoso, C., and De~Cristofaro, E. (2020).
\newblock Measuring membership privacy on aggregate location time-series.
\newblock {\em Proceedings of the ACM on Measurement and Analysis of Computing
  Systems}, 4(2):1--28.

\bibitem[{SafeGraph}, 2024]{safegraph_spend}
{SafeGraph} (2024).
\newblock Enrich pois with aggregated transaction data.
\newblock \url{https://www.safegraph.com/products/spend}.
\newblock Accessed: 2024-02-13.

\bibitem[Salem et~al., 2018]{salem2018ml}
Salem, A., Zhang, Y., Humbert, M., Berrang, P., Fritz, M., and Backes, M.
  (2018).
\newblock Ml-leaks: Model and data independent membership inference attacks and
  defenses on machine learning models.
\newblock {\em arXiv preprint arXiv:1806.01246}.

\bibitem[Sankararaman et~al., 2009]{sankararaman2009genomic}
Sankararaman, S., Obozinski, G., Jordan, M.~I., and Halperin, E. (2009).
\newblock Genomic privacy and limits of individual detection in a pool.
\newblock {\em Nature genetics}, 41(9):965--967.

\bibitem[Savage, 2021]{NYT2021DIASurveillance}
Savage, C. (2021).
\newblock Intelligence analysts use u.s. smartphone location data without
  warrants, memo says.
\newblock {\em The New York Times}.
\newblock Available at
  \url{{https://www.nytimes.com/2021/01/22/us/politics/dia-surveillance-data.html}}.

\bibitem[Schneider et~al., 2013]{schneider2013unravelling}
Schneider, C.~M., Belik, V., Couronn{\'e}, T., Smoreda, Z., and Gonz{\'a}lez,
  M.~C. (2013).
\newblock Unravelling daily human mobility motifs.
\newblock {\em Journal of The Royal Society Interface}, 10(84):20130246.

\bibitem[Seshadri et~al., 2008]{seshadri2008mobile}
Seshadri, M., Machiraju, S., Sridharan, A., Bolot, J., Faloutsos, C., and
  Leskove, J. (2008).
\newblock Mobile call graphs: beyond power-law and lognormal distributions.
\newblock In {\em Proceedings of the 14th ACM SIGKDD international conference
  on Knowledge discovery and data mining}, pages 596--604.

\bibitem[Shokri et~al., 2017]{shokri2017membership}
Shokri, R., Stronati, M., Song, C., and Shmatikov, V. (2017).
\newblock Membership inference attacks against machine learning models.
\newblock In {\em 2017 IEEE symposium on security and privacy (SP)}, pages
  3--18. IEEE.

\bibitem[Song et~al., 2010]{song2010limits}
Song, C., Qu, Z., Blumm, N., and Barab{\'a}si, A.-L. (2010).
\newblock Limits of predictability in human mobility.
\newblock {\em Science}, 327(5968):1018--1021.

\bibitem[SpazioDati and di~Milano, 2015]{DVN/9IZALB_2015}
SpazioDati and di~Milano, D.~P. (2015).
\newblock {Social Pulse - Milano}.

\bibitem[Stadler et~al., 2022]{stadler2022synthetic}
Stadler, T., Oprisanu, B., and Troncoso, C. (2022).
\newblock Synthetic data--anonymisation groundhog day.
\newblock In {\em 31st USENIX Security Symposium (USENIX Security 22)}, pages
  1451--1468.

\bibitem[{Telus}, 2024]{telus_insights_2024}
{Telus} (2024).
\newblock {Telus Insights Location API}.
\newblock \url{https://docs.insights.telus.com/}.
\newblock [Online].

\bibitem[Tournier and de~Montjoye, 2022]{tournier2022expanding}
Tournier, A.~J. and de~Montjoye, Y.-A. (2022).
\newblock Expanding the attack surface: Robust profiling attacks threaten the
  privacy of sparse behavioral data.
\newblock {\em Science Advances}, 8(33):eabl6464.

\bibitem[Trasberg and Cheshire, 2023]{trasberg2023spatial}
Trasberg, T. and Cheshire, J. (2023).
\newblock Spatial and social disparities in the decline of activities during
  the covid-19 lockdown in greater london.
\newblock {\em Urban Studies}, 60(8):1427--1447.

\bibitem[Truex et~al., 2019]{truex2019demystifying}
Truex, S., Liu, L., Gursoy, M.~E., Yu, L., and Wei, W. (2019).
\newblock Demystifying membership inference attacks in machine learning as a
  service.
\newblock {\em IEEE Transactions on Services Computing}, 14(6):2073--2089.

\bibitem[Van~Zoonen, 2016]{van2016privacy}
Van~Zoonen, L. (2016).
\newblock Privacy concerns in smart cities.
\newblock {\em Government Information Quarterly}, 33(3):472--480.

\bibitem[Xu et~al., 2015]{xu2015understanding}
Xu, Y., Shaw, S.-L., Zhao, Z., Yin, L., Fang, Z., and Li, Q. (2015).
\newblock Understanding aggregate human mobility patterns using passive mobile
  phone location data: a home-based approach.
\newblock {\em Transportation}, 42:625--646.

\bibitem[Yabe et~al., 2022]{yabe2022mobile}
Yabe, T., Jones, N.~K., Rao, P. S.~C., Gonzalez, M.~C., and Ukkusuri, S.~V.
  (2022).
\newblock Mobile phone location data for disasters: A review from natural
  hazards and epidemics.
\newblock {\em Computers, Environment and Urban Systems}, 94:101777.

\bibitem[Yeom et~al., 2018]{yeom2018privacy}
Yeom, S., Giacomelli, I., Fredrikson, M., and Jha, S. (2018).
\newblock Privacy risk in machine learning: Analyzing the connection to
  overfitting.
\newblock In {\em 2018 IEEE 31st computer security foundations symposium
  (CSF)}, pages 268--282. IEEE.

\bibitem[Zang and Bolot, 2011]{zang2011anonymization}
Zang, H. and Bolot, J. (2011).
\newblock Anonymization of location data does not work: A large-scale
  measurement study.
\newblock In {\em Proceedings of the 17th annual international conference on
  Mobile computing and networking}, pages 145--156.

\bibitem[Zhang et~al., 2020]{zhang2020locmia}
Zhang, G., Zhang, A., and Zhao, P. (2020).
\newblock Locmia: Membership inference attacks against aggregated location
  data.
\newblock {\em IEEE Internet of Things Journal}, 7(12):11778--11788.

\bibitem[Zhou, 2017]{zhou2017understanding}
Zhou, T. (2017).
\newblock Understanding location-based services users’ privacy concern: An
  elaboration likelihood model perspective.
\newblock {\em Internet Research}, 27(3):506--519.

\bibitem[Zhu et~al., 2022]{zhu2022post}
Zhu, K., Fioretto, F., and Van~Hentenryck, P. (2022).
\newblock Post-processing of differentially private data: A fairness
  perspective.
\newblock {\em arXiv preprint arXiv:2201.09425}.

\end{thebibliography}

\appendix
\addcontentsline{toc}{section}{Appendices}

\section*{Appendix}
\label{sec:appendix}
\begin{table}[!h]
 \caption{Glossary of notations.}
  \centering
    \resizebox{0.5\textwidth}{!}{%
  \begin{tabular}{ll} \\
    \toprule
    \textbf{Notation}  & \textbf{Definition}\\
    \midrule
    $\Si$ & Set of regions of interests (ROIs) \\
    $\T$ & Set of epochs in observation period  \\
    $\Pop$ & Set of all users in the dataset \\
    $L^u$ & Location trace of user $u \in \Pop$ over $\Si \times \T$ \\
    $\Users$ & Aggregation group of users sampled from $\Pop$\\
    $A^{\Users}$ & Raw aggregate count matrix in $\Si \times T$ over users in $\Users$ \\
    ${A}_{DP}^{\Users}(\varepsilon)$ & An $\varepsilon$-DP aggregate\\[1em]
    $A_{SSC}^{\Users}(k)$ & An aggregate with counts $\le k$ suppressed \\[1em]
    ${A}_{DP, SSC}^{\Users}(\varepsilon, k)$ & An $\varepsilon$-DP aggregate with counts $\le k$ suppressed \\[1em]
    $\overline{A}^{\Users}$ & The released aggregate count matrix \\
    $m$ & Number of users in the aggregation group $\Users$\\
    $u^*$ & Target drawn from full population, $u^* \in \Pop $\\
    $Adv$ &  Adversary performing MIA on $u^*$\\
    \bottomrule \\
  \end{tabular}
  }\label{tab:glossary_notations}
\end{table}

\begin{table}
 \caption{Default experiment parameters.}
  \centering
  \begin{tabular}{ll} \\
    \textbf{Default value}  & \textbf{Definition}\\
    \hline
     $n_{train} = 400$ & Number of training aggregates \\
     $n_{val} = 100$ & Number of validation aggregates\\
     $n_{test} = 100$ & Number of test aggregates \\
     $n_{target} = 50$ & Number of targets \\
     $m=1000$ & Aggregate size  \\
     $|Ref|=5000$ (CDR),  & Traces in $Adv$'s real (KK) \\
     $2500$ (Milan)  &or synthetic (ZK) reference\\
     $p_{u^*}= 1$ & Fraction of $L^{u^*}$ known by $Adv$\\
    \hline \\
  \end{tabular}
  \label{tab:parameters}
\end{table}

\section{Supplementary Proofs}

\label{sec: proofs}
\begin{definition} (Oracle average count) Given a raw aggregate $A^\Users$, we define the oracle average count function $\mu_{S}: \Si \to \R^+$ as 
\begin{align}
    \mu_{S}(s) &=\lim_{|\T|\to \infty}\frac{\sum_{t=1}^{|\T|} A_{s,t}^\Users}{|\T|}.  
\end{align}
Letting $|\T| \to \infty$ corresponds to extending the observation period indefinitely. Thus, $\mu_{S}(s)$ represents the expected number of users $\Users$ who visit ROI $s$ at a randomly selected epoch, given infinite location data over the ROIs $\Si$.
\end{definition}

\begin{definition} (Strong sparsity)
We say that $A^\Users$ is \textit{strongly sparse} if
\begin{align}
    \mu_{S}(s) = 0 \text{ } \forall s \in \Si
\end{align}
Equivalently, $\sum_{t=1}^{|\T|} A_{s,t}^\Users \in o(|\T|)$ $\forall s \in \Si$. This is a strong assumption, as it implies that the visitation rate to each ROI decreases at a sublinear rate.

\begin{lemma}
Given a fixed geographic region in which location data is collected,
\begin{align}
    \lim_{|\Si| \to \infty}\frac{\sum_{s=1}^{|\Si|} A_{s,t}^\Users}{|\Si|}=0
\end{align}
\label{lemma:visitation_rate_space}
\end{lemma}
\begin{proof}
${\sum_{s=1}^{|\Si|} A_{s,t}^\Users}$ corresponds to the number of users who registered a visit during epoch $t$. Letting $|\Si| \to \infty$ corresponds to increasing creating finer regional partitions within the fixed geographic region. ${\sum_{s=1}^{|\Si|} A_{s,t}^\Users}$ is invariant to increasing $|\Si| \to \infty$, since the same users are observed over the same time. It follows that  $\lim_{|\Si| \to \infty}\frac{\sum_{s=1}^{|\Si|} A_{s,t}^\Users}{|\Si|}=0$.
\end{proof}

\end{definition}

\begin{theorem}(Convergence of empirical marginals to uniform distribution under $\varepsilon$-DP)
    Let $\Delta > 0$ be the global sensitivity and suppose that $\varepsilon$-DP is applied on an aggregate release $\overline{A}^\Users=A_{DP}^{\Users}(\varepsilon)$ with post-processed non-negative counts. If the original raw counts $A^\Users$ are strongly sparse, then the empirical space and time marginals, $\mathcal{P}_{S}^0$ and $\mathcal{P}_{T}^0$, each converge to discrete uniform distributions:
    \begin{itemize}
        \item $
\widehat{\mathcal{P}}_{S}^0 \to Unif(\Si)$ in distribution as $|\T| \to \infty$
\item  $\widehat{\mathcal{P}}_{T}^0 \to Unif(\T)$ in distribution as $|\Si| \to \infty$
    \end{itemize}
    \label{theorem:marg_conv_unif}
\end{theorem}
\begin{proof}
We first consider $\widehat{\mathcal{P}}_{S}^0$. It suffices to show that as $\epsilon \to 0$, $
\widehat{\mathcal{P}}_{S}^0(s_0) \xrightarrow{\text{a.s.}} \frac{1}{|\Si|}$ for each $s_0 \in \Si$.

Let $\epsilon>0$ and let $b = \frac{\Delta}{\epsilon}$. Recall that $\varepsilon$-DP with post-processed non-negative counts is obtained by $\overline{A}^\Users_{s,t}=(A_{s,t}^{\Users} + L_{s,t}(b)) \vee 0$, where $A_{s,t}^{\Users}$ is the true number of visits by users in $\Users$ to $(s,t)$ and $\{L_{s,t} \sim Lap(b): s \in \Si, t \in \T\}$ are i.i.d Laplacian noise samples (Section \ref{subsec:DP}). By definition, 
\begin{align*}
\widehat{\mathcal{P}}_{S}^0(s_0) &= \frac{\sum_{t=1}^{|\T|}\overline{A}^\Users_{s_0,t}}{\sum_{s=1}^{|\Si|}\sum_{t=1}^{|\T|}\overline{A}^\Users_{s,t}}\\
&=\frac{\sum_{t=1}^{|\T|}\left((A_{s_0,t}^\Users + L_{s_0,t}(b)) \vee 0\right)}{\sum_{s=1}^{|\Si|}\sum_{t=1}^{|\T|}\left((A_{s,t}^\Users + L_{s,t}(b)) \vee 0\right)}\\
&= \frac{\frac{\sum_{t=1}^{|\T|}\left((A_{s_0,t}^\Users + L_{s_0,t}) \vee 0\right)}{|\T|}}{\sum_{s=1}^{|\Si|}\frac{\sum_{t=1}^{|\T|}\left((A_{s,t}^\Users + L_{s,t}) \vee 0\right)}{|\T|}}.
\end{align*}

We now express $(A_{s,t}^{\Users} + L_{s,t}) \vee 0 = L_{s,t} \vee 0 + X_{s,t}$, for some $X_{s,t}$, in order to apply Lemma \ref{lemma:Laplace_mean} later. Since $A_{s,t}^{\Users} \ge 0$, there are three cases:
\[
X_{s,t}=
\begin{cases}
 A_{s,t}^{\Users}, & \text{if $L_{s,t}  \ge 0$ }  \\
 A_{s,t}^{\Users} + {L}_{s,t}, & \text{if $L_{s,t}  < 0$ and $(A_{s,t}^{\Users} + L_{s,t} ) \vee 0 > 0$} \\
0, & \text{if $L_{s,t}  < 0$ and $(A_{s,t}^{\Users} + L_{s,t} ) \vee 0 = 0$ }
\end{cases}
\]
We therefore have
\begin{align}
    \widehat{\mathcal{P}}_{S}^0(s) &= \frac{\frac{\sum_{t=1}^{|\T|}X_{s_0,t}}{|\T|} + \frac{\sum_{t=1}^{|\T|}L_{s_0,t}  \vee 0}{|\T|}}{\sum_{s=1}^{|\Si|}\left(\frac{\sum_{t=1}^{|\T|}X_{s,t}}{|\T|} + \frac{\sum_{t=1}^{|\T|}L_{s,t}  \vee 0}{|\T|}\right)}.
\end{align}
By sparsity, for each $s \in \Si$
\begin{align}
     \frac{\sum_{t=1}^{|\T|}A_{s,t}^\Users}{|\T|} \to 0 \text{ as $|\T| \to \infty$}
\end{align}
Also, by the Strong Law of Large Numbers, since $\{L_{s,t}\}$ are  i.i.d., and $\E[Lap(b)] = 0$, we have 
\begin{align*}
     \frac{\sum_{t=1}^{|\T|}L_{s,t}}{|\T|} \xrightarrow{\text{a.s.}} 0 \text{ as $|\T| \to \infty$}
\end{align*}
By linearity,
\begin{align*}
     \frac{\sum_{t=1}^{|\T|} A_{s,t}^\Users + {L}_{s,t}}{|\T|} \xrightarrow{\text{a.s.}} 0 \text{ as $|\T| \to \infty$}
\end{align*}
Hence, in all three possible cases,
\begin{align*}
    \frac{\sum_{t=1}^{|\T|}X_{s,t}}{|\T|} \xrightarrow{\text{a.s.}} 0 \text{ as $|\T| \to \infty$}
\end{align*}
This allows us to simplify
\begin{align*}
    \widehat{\mathcal{P}}_{S}^0(s) &\xrightarrow{\text{a.s.}} \frac{\frac{\sum_{t=1}^{|\T|}L_{s_0,t} \vee 0}{|\T|}}{\sum_{s=1}^{|\Si|} \frac{\sum_{t=1}^{|\T|}L_{s,t} \vee 0}{|\T|}}.
\end{align*}
Since for all $s,t$, $L_{s,t} \vee 0 \sim Lap(b) \vee 0$, Lemma \ref{lemma:Laplace_mean} implies $\E[L_{s,t} \vee 0] = \frac{b}{2}$. Hence, by the Strong Law of Large Numbers,
\begin{align*}
    \frac{\sum_{t=1}^{|\T|} L_{s,t} \vee 0}{|\T|} \xrightarrow{\text{a.s.}} \frac{b}{2} \text{ as $|\T| \to \infty$}
\end{align*}
Finally, for any set of ROIs $\Si$, and any $s \in \Si$,
\begin{align*}
    \widehat{\mathcal{P}}_{S}^0(s) &\xrightarrow{\text{a.s.}} \frac{\frac{b}{2}}{\sum_{s=1}^{|\Si|}\frac{b}{2}} = \frac{b}{|\Si|b} = \frac{1}{|\Si|},
\end{align*}
A symmetric argument proves $\widehat{\mathcal{P}}_{T}^0 \to Unif(\T)$ in distribution as $|\Si| \to \infty$, using Lemma \ref{lemma:visitation_rate_space} instead of strong sparsity.
\end{proof}

\textbf{Remark. } We note that strong sparsity is assumed in Eq. (14) to prove that $\frac{\sum_{t=1}^{|\T|}X_{s,t}}{|\T|} \xrightarrow{\text{a.s.}} 0 \text{ as $|\T| \to \infty$}$. Although we expect the oracle average count $\mu_S(s)$ to be very small for most $s \in \Si$, due to the sparsity of aggregate location data, it is unlikely to observe $\mu_S = 0$ for real data. Substituting $\mu_S(s)$ in place of $0$ in Eq. (14) will not yield the uniform probability $\widehat{\mathcal{P}}_{S}^0(s) \xrightarrow{\text{a.s.}} \frac{1}{|\Si|}$, but it will be a close approximation, provided that $\frac{\Delta}{\epsilon} >> \mu_S$ and that the number of epochs is large. 

In practice, fixed dimensions for $S$ and $T$ will prevent the empirical marginals from completely converging to the uniform distribution. This is demonstrated for different noise scales on the Milan dataset (which has $|S|=100$ and $|T|=168$) in Figure \ref{fig:Milan_laplace}.

\begin{lemma}
Suppose that $Y \sim L \vee 0$, with $L \sim Lap(b)$. Then, $Y$ has mean
\begin{align*}
    \E[Y] = \frac{b}{2}
\end{align*}
\label{lemma:Laplace_mean}
\end{lemma}
\begin{proof}
Let $L\sim Lap(b)$. Then, its probability density function (pdf) $f_L$ is given by 
\begin{align*}
    f_L(x) =
\frac{1}{2b}e^{\frac{|x|}{b}} & \text{ for } x \in \R
\end{align*}
which is symmetric about $x=0$. Hence, $P(X \le 0) = \frac{1}{2}$. It follows that $Y = X \vee 0$ has the pdf $f_Y$
\[
f_Y(x) =
\begin{cases}
0, & \text{for } x < 0 \\
\frac{\delta(x)}{2}, & \text{for } x = 0 \\
\frac{1}{2b}e^{-\frac{x}{b}}, & \text{for } x > 0
\end{cases}
\]
where \( \delta(x) \) is the Dirac delta function representing the accumulated probability mass at zero. We then evaluate
\begin{align*}
    \E[Y] &= \frac{1}{2b} \int_{0}^{\infty} xe^{-\frac{x}{b}}dx\\
    &= \frac{1}{2b}\left(-bxe^{-\frac{x}{b}} \bigg|_{0}^{\infty} +b\int_{0}^\infty e^{-\frac{x}{b}}dx\right)\\
    &=\frac{1}{2b}\left(-b^2e^{-\frac{x}{b}}\bigg|_{0}^{\infty} \right)\\
    &=\frac{1}{2b}\left(b^2\right) = \frac{b}{2}
\end{align*}

\end{proof}

\section{Algorithms}
In this section, we present the main algorithms required to generate synthetic traces from the released aggregate for our ZK MIA.

 Algorithm \ref{alg:synthetic_generation} describes how we adapted the unicity model from ~\citet{farzanehfar2021risk} to generate synthetic traces for ZK MIA. We note that the procedure for generating a synthetic trace can also be interpreted as running a Markov chain $\{X_i: i=1, ..., n_{visits}\}$ over the state space of spatiotemporal pairs $(s,t) \in C(s_0) \times \T$ with transition probabilities to $(s', t') \in C(s_0) \times \T$ proportional to the product of the pmfs $\mathcal{P}_S(s')\mathcal{P}_T(t')$.

Algorithm \ref{alg:marginals} estimates the three marginal probability distributions required to run  Algorithm \ref{alg:synthetic_generation}: the space marginal $\mathcal{P}_S$, the time marginal $\mathcal{P}_T$, and the activity marginal $\mathcal{P}_A$ from an aggregate release $\overline{A}^U$. We estimate the marginals via our denoising and debiasing techniques (from Section \ref{subsec: est_space_time_marginals}), depending on the application of privacy measures on $\overline{A}^U$.

Algorithm \ref{alg:estimate_mean}  describes our procedure for achieving an estimate $\hat{\mu}$ for the mean number of visits per user in the dataset given a privacy-aware aggregate release. Recall that $\mathcal{P}_A$ is set to $Exp(\hat{\mu})$. Algorithm \ref{alg:p_selection} describes our procedure for computing which degree $p$ will work best in the power transformation, to correct the empirical marginal $\widehat{\mathcal{P}}_{\underscore}^0$ obtained directly from a $\varepsilon$-DP aggregate release.

\begin{algorithm}[h!]
    \caption{\textsc{GenerateSyntheticTrace}}
    \label{alg:synthetic_generation}
    \begin{algorithmic}[1]
        \Inputs{$\mathcal{P}_S$: Approximated space marginal over ROIs\\ $\mathcal{P}_T$: Approximated time marginal over epochs \\ 
        $\mathcal{P}_A$: Approximated activity marginal over trace sizes\\ $DT(\Si)$: Delaunay triangulation of ROIs
        }
        \Output{$L^s$: A synthetic trace}
        \comm{We sample an origin ROI.}
        \State{$s_0$ $\leftarrow$ sample\_from\_distribution($\mathcal{P}_S, 1$)}\\
        \comm{Use $DT(\Si)$ to create a connected subgraph of ROIs including the origin ROI}
        \State{$C(s_0)$ $\leftarrow$ generate\_connected\_subgraph($s_0$, $DT(\Si)$, n\_rois\_subgraph $=10$ (default value from \cite{farzanehfar2021risk})}\\
        \comm{Normalize $\mathcal{P}_s$ restricted to $C(s_0)$.}
        \State{$\mathcal{P}_{C(s_0)}\leftarrow$ normalize(restrict($\mathcal{P}_s$, $C(s_0)$)) }\\
        \comm{Sample the trace size (\# visits).}
        \State{n\_visits $\leftarrow$ round(sample\_from\_distribution($\mathcal{P}_A, 1$))}\\
        \comm{Randomly sample n\_visits ROIs and epochs with replacement.}
        \State{ROIs $\leftarrow$ sample\_from\_distribution($\mathcal{P}_{C(s_0)}$, n\_visits)}
        \State{epochs $\leftarrow$ sample\_from\_distribution($\mathcal{P}_T$, n\_visits)}\\
        \Return{$L^s$ $\leftarrow$ [ (ROIs[i], epochs[i]) for i = 1 ... n\_visits ]}
    \end{algorithmic}
\end{algorithm}

\begin{algorithm}[h!]
    \caption{\textsc{Approximate Marginals From Aggregate}}
    \label{alg:marginals}
    \begin{algorithmic}[1]
        \Inputs{$\overline{A}^{\Users}:$ Released aggregate\\
        $m$: Aggregate group size\\
        $p$: Specified probability distribution family}
        \Output{$\mathcal{P}_S$: Approximated space marginal over ROIs\\ 
        $\mathcal{P}_T$: Approximated time marginal over epochs \\ 
        $\mathcal{P}_A$: Approximated activity marginal over trace sizes}
        \comm{Compute direct estimates.}
        \State{$\mathcal{P}_S, \mathcal{P}_T$ $\leftarrow$ compute\_empirical\_marginals($\overline{A}^\Users$)}
        \State{$\mu_{visits}^0$ $\leftarrow$ sum\_entries$(\overline{A}^\Users) / m$}
        \If{$\overline{A}^\Users=A^{\Users}$}\\
            \comm{Return direct estimates if no privacy.}
            \State{$\mathcal{P}_A$ $\leftarrow$ fit\_dist($p$, $\mu_{visits}^0$)}
            \State \Return{$\mathcal{P}_S, \mathcal{P}_T,\mathcal{P}_A$}
        \EndIf
        \If{$\overline{A}^\Users=A_{SSC}^{\Users}(k)$  and $k>0$}\\
            \comm{Apply log compression if SSC.}
            \State{$\mathcal{P}_S, \mathcal{P}_T$ $\leftarrow$ log\_compression$(\mathcal{P}_S, \mathcal{P}_T)$}
        \EndIf
        \If{$\overline{A}^\Users=A_{DP}^{\Users}(\varepsilon)$ or $\overline{A}^\Users=A_{DP, SSC}^{\Users}(\varepsilon,k)$}\\
            \comm{Apply power transformation if DP.}
            \State{$\mathcal{P}_S, \mathcal{P}_T$ $\leftarrow$ power\_transform$(\mathcal{P}_S, \mathcal{P}_T)$}
            \EndIf\\
            \comm{Apply Algorithm \ref{alg:estimate_mean} from Appendix.}
            \State{$\mu_{visits}$ $\leftarrow$ estimate\_mean($\mu_{visits}^0$, $A$, $m$, $\mathcal{P}_S, \mathcal{P}_T$, $k$, $\varepsilon$)}
            \State{$\mathcal{P}_A$ $\leftarrow$ fit\_dist($p$, $\mu_{visits}$)}
            \State \Return{$\mathcal{P}_S, \mathcal{P}_T,\mathcal{P}_A$}
    \end{algorithmic}
\end{algorithm}

\begin{algorithm}[h!]
    \caption{\textsc{EstimateMean}}
    \label{alg:estimate_mean}
    \begin{algorithmic}[1]
        \Inputs{$\mu_0$: Initial guess for mean visits\\
        $A:$ Released  aggregate\\
        $DT(\mathcal{S})$: Delaunay triangulation of ROIs\\
        $\mathcal{P}_s$: Estimated space marginal\\
        $\mathcal{P}_t$: Estimated time marginal\\
        $k$: Suppression threshold\\
        $\varepsilon, \Delta$: DP parameters\\
        $m$: Aggregate group size}\\
        \textbf{Additional parameters}
        \State{tol: Tolerance for stopping}
        \State{max\_iter: Max iterations}
        \Output{$\mu$: Approximated mean visits per user}
        \State{$\mu$ $\leftarrow$ $\mu_0$}
        \For{$i=1$ to max\_iter}
        \State{$A_{1}$ $\leftarrow$ initialize\_matrix()}\\
        \comm{Create a synthetic aggregate of size $m$.}
        \For{$j=1$ to $m$}\\
        \comm{Generate synthetic trace via with $\mu$ vistis.}
        \State{$A_1$ $\leftarrow$ $A_1 + $ generate\_synthetic\_trace($\mathcal{P}_s$, $\mathcal{P}_t$, $\mu$, $DT(\Si)$)}
        \EndFor
        \State{$A_1$ $\leftarrow$ apply\_privacy\_measures($A_1$, $k$, $\epsilon$, $\Delta$)}\\
        \comm{Increase or decrease the estimate $\mu$ accordingly}
        \State{$\mu \leftarrow \mu_0 +$ (sum($A$)- sum($A_1$))/$m$}
        \If{$|\mu-\mu_0| < tol$}
        \Return{$\mu$}
        \EndIf
        \EndFor\\
        \Return{$\mu$}
    \end{algorithmic}
\end{algorithm}

\begin{algorithm}[h!]
    \caption{\textsc{pSelection}}
    \label{alg:p_selection}
    \begin{algorithmic}[1]
        \Inputs{$\sigma_0$: Reference variance\\
        $\mathcal{P}$: Space or time marginal to be modified\\
        $\epsilon_{tol}$: Error tolerance}
        \Output{$p$: Degree for transformation $x^p$ that sets the variance of $\mathcal{P}$ to approximately match $\sigma_0$}
        \State{$\sigma$ $\leftarrow$ $compute\_variance(\mathcal{P})$}
        \comm{We compute the variance from the original marginal value.}
        \State{$p$ $\leftarrow$ 1}
        \While{$|\sigma_0 - \sigma| > \epsilon_{tol}$}
        \State{$\sigma$ $\leftarrow$ $compute\_variance(pow(\mathcal{P},p))$}\\
        \comm{Increment the power $p$ until estimate is in range.}
        \State{$p$ $\leftarrow$ $p+0.01$}
        \EndWhile\\
        \Return{$\sigma$}
    \end{algorithmic}
\end{algorithm}

\newpage

\section{Accuracy Results}
\label{sec: accuracy}

In this section, we present the accuracy scores of ZK MIA and KK MIA for the experiments on suppression of small counts and $\varepsilon$-DP noise addition.

Table \ref{tab:bucket_MIAs} presents the accuracy scores obtained by ZK and KK from the experiments on suppression of small counts from Section \ref{exp:ssc}. Table \ref{tab:epsilon_MIAs} presents the accuracy scores obtained by ZK and KK from the experiments on event level $\varepsilon$-DP from Section \ref{exp:DP}. Table \ref{tab:epsilon_MIAs_user_day} presents the accuracy scores obtained by ZK and KK from the experiments on user-day level $\varepsilon$-DP.

We observe that the accuracy scores of KK and ZK are close in each experiment, as observed already with the AUC metric in the main text.

\begin{table}[htbp]
\centering
\caption{Mean accuracy scores with standard error for KK and ZK on size $1000$ aggregates from the CDR and Milan datasets across various suppression thresholds $k$. } 
\label{tab:bucket_MIAs}
\begin{tabular}{lcccc}
\toprule
\multirow{2}{0pt}{\textbf{$k$}} & \multicolumn{2}{c}{\textbf{CDR dataset}}  & \multicolumn{2}{c}{\textbf{Milan dataset}}\\
& \textbf{KK} & \textbf{ZK} & \textbf{KK} & \textbf{ZK}\\

\midrule
0 & $0.980 \pm 0.012$ & $\it{0.991 \pm 0.008}$ & $0.990 \pm 0.005$ & $\it{0.990 \pm 0.003}$\\
1 & $\it{0.907 \pm 0.028}$& $0.879 \pm 0.025$& $\it{0.767 \pm 0.018}$ & $\it{0.700 \pm 0.017}$\\
2 & $0.807 \pm 0.031$& $\it{0.827 \pm 0.026}$& $\it{0.683 \pm 0.019}$ & $\it{0.631 \pm 0.013}$\\
3 & $\it{0.685 \pm 0.032}$& $0.687 \pm 0.031$& $\it{0.600 \pm 0.016}$ & $\it{0.550 \pm 0.009}$\\
4 & $\it{0.597 \pm 0.024}$& $0.603 \pm 0.027$& $\it{0.566 \pm 0.011}$ & $\it{0.512 \pm 0.003}$\\
5 & $0.543 \pm 0.018$ & $\it{0.528 \pm 0.019}$& $\it{0.536 \pm 0.010}$ & $\it{0.500 \pm 0.000}$\\
\bottomrule
\end{tabular}
\end{table}

\begin{table}[htbp]
\centering
\caption{Mean accuracy scores with standard error for KK and ZK on size $1000$ aggregates from the CDR and Milan datasets across various privacy budgets $\varepsilon$ for event level DP.} 
\label{tab:epsilon_MIAs}
\begin{tabular}{lccccc}
\toprule
\multirow{2}{0pt}{\textbf{$\varepsilon$}} & \multicolumn{2}{c}{\textbf{CDR dataset}}  & \multicolumn{2}{c}{\textbf{Milan dataset}}\\
& \textbf{KK} & \textbf{ZK} & \textbf{KK} & \textbf{ZK}\\

\midrule
 0.1 & $\it{0.588 \pm 0.019}$ & $0.555 \pm 0.014$& $\it{0.539 \pm 0.006}$ & $\it{0.549 \pm 0.007}$\\
 0.5 & $\it{0.848 \pm 0.028}$ & $0.791 \pm 0.019$& $\it{0.744 \pm 0.010}$ & $\it{0.634 \pm 0.007}$\\
 1.0 & $0.920 \pm 0.026$ & $0.907 \pm 0.018$& $\it{0.850 \pm 0.016}$ & $\it{0.594 \pm 0.008}$\\
 5.0 & $0.906 \pm 0.035$ & $0.934 \pm 0.019$& $\it{0.881 \pm 0.022}$ & $\it{0.660 \pm 0.018}$\\
 10.0 & $0.923 \pm 0.029$ & $0.934 \pm 0.018$& $\it{0.920 \pm 0.021}$ & $\it{0.671 \pm 0.019}$\\
\bottomrule
\end{tabular}
\end{table}

\begin{table}[htbp]
\centering
\caption{Mean accuracy scores with standard error for KK and ZK on size $1000$ aggregates from the CDR and Milan datasets across various privacy budgets $\varepsilon$ for user-day level DP.} 
\label{tab:epsilon_MIAs_user_day}
\begin{tabular}{lcccc}
\toprule

\multirow{2}{0pt}{\textbf{$\varepsilon$}} & \multicolumn{2}{c}{\textbf{CDR dataset}}  & \multicolumn{2}{c}{\textbf{Milan dataset}}\\
& \textbf{KK} & \textbf{ZK} & \textbf{KK} & \textbf{ZK}\\

\midrule
  0.1 & $0.502 \pm 0.014$ & $\it{0.502 \pm 0.010}$ & $0.497 \pm 0.006$ & $0.496 \pm 0.006$\\
 0.5 & $0.508 \pm 0.014$ & $\it{0.526 \pm 0.016}$ & $0.517 \pm 0.006$ & $0.519 \pm 0.006$\\
 1.0 & $0.533 \pm 0.014$ & $\it{0.539 \pm 0.014}$ & $0.534 \pm 0.006$ & $0.544 \pm 0.007$\\
 5.0 & $\it{0.723 \pm 0.020}$ & $0.676 \pm 0.018$ & $0.746 \pm 0.009$ & $0.680 \pm 0.008$\\
 10.0 & $\it{0.874 \pm 0.025}$ & $0.825 \pm 0.019$ & $0.870 \pm 0.014$ & $0.777 \pm 0.018$\\
\bottomrule
\end{tabular}
\end{table}

\newpage

\section{Additional Experiments}

\subsection{Varying the size of the aggregate}
\label{appendix:varyingm}

Since ZK MIA requires the estimation of statistics from the aggregate, there may be concerns about its performance when the aggregate size is small. However, like previous MIAs, ZK MIA performs more effectively on smaller-scale aggregates compared to larger aggregates. This is shown in Figure \ref{fig:variation_aggregate_size} for aggregate sizes $m=100, 250, 500, 1000$ and different privacy budgets $\varepsilon$.

To further understand how MIA performance scales with aggregate size $m$, we also consider $m > 1000$ in this experiment. To this end, we vary $m= 100, 500, 1000, 2000, 3000$ and compare the performance of KK MIA and ZK MIA on raw ($k=0$) and suppressed ($k=1$) aggregates.  Results on the CDR dataset are reported in Tables~\ref{tab:k0_CDR} and ~\ref{tab:k1_CDR} and results on the Milan dataset are reported in Tables~\ref{tab:k0_Milan} and ~\ref{tab:k1_Milan}. $m=3000$ was not run on the Milan dataset due size limitations.

In these settings with mild privacy protection, the attacks always succeed regardless of the value of $m$. We also observe a few intuitive trends. First, when raw aggregates ($k=0$) are attacked, increasing the size of the aggregates slowly decreases the performance of the attack. On the CDR dataset, KK and ZK attain AUCs $0.999$ and $1.0$ for $m=100$, which decreases to $0.919$ and $0.977$ for $m=3000$. Second, when we apply suppression $k=1$, the attacks initially perform poorly when the aggregate size is small. We hypothesize this to be due to a larger percentage of entries being suppressed when fewer traces are aggregated, leaving less information in the release. This effect gradually decrease as aggregate size increases. It is then counterbalanced by the first effect, that increasing the size of the aggregates slowly decreases the performance of the attack, when aggregate sizes increase. This is visible for $m \leq 1000$ in the CDR dataset. For the Milan dataset, AUC however still monotonically increases even beyond $m \leq 1000$ as the dataset is more sensitive to suppression with the average user has approximately $6$ times less visits, as shown in Table \ref{fig:trace_size_Milan}.

\begin{table}[htbp]
\centering
\caption{Mean AUCs of KK and ZK MIAs for \( k = 0 \) on the CDR dataset with varying \( m \).}
\label{tab:k_0_cdr}
\begin{tabular}{lcc}
\toprule
\textbf{\( m \)} & \textbf{KK} & \textbf{ZK} \\
\midrule
100  & \(1.000 \pm 0.000\) & \(1.000 \pm 0.000\) \\
500  & \(1.000 \pm 0.000\) & \(1.000 \pm 0.000\) \\
1000 & \(1.000 \pm 0.000\) & \(0.999 \pm 0.001\) \\
2000 & \(0.997 \pm 0.003\) & \(0.994 \pm 0.006\) \\
3000 & \(0.988 \pm 0.011\) & \(0.977 \pm 0.021\) \\
\bottomrule
\end{tabular}
\label{tab:k0_CDR}
\end{table}

\begin{table}[htbp]
\centering
\caption{Mean AUCs of KK and ZK MIAs for \( k = 0 \) on the Milan dataset with varying \( m \).}
\label{tab:k_0_milan}
\begin{tabular}{lcc}
\toprule
\textbf{\( m \)} & \textbf{KK} & \textbf{ZK} \\
\midrule
100  & \(1.000 \pm 0.000\) & \(1.000 \pm 0.000\) \\
500  & \(1.000 \pm 0.000\) & \(1.000 \pm 0.000\) \\
1000 & \(0.995 \pm 0.002 \) & \(1.000 \pm 0.000\) \\
2000 & \(0.977 \pm 0.011\) & \(1.000 \pm 0.000\) \\
\bottomrule
\end{tabular}
\label{tab:k0_Milan}
\end{table}

\begin{table}[htbp]
\centering
\caption{Mean AUCs of KK and ZK MIAs for \( k = 1 \) on the CDR dataset with varying \( m \).}
\label{tab:k_1_cdr}
\begin{tabular}{lcc}
\toprule
\textbf{\( m \)} & \textbf{KK} & \textbf{ZK} \\
\midrule
100  & \(0.856 \pm 0.019\) & \(0.779 \pm 0.041\) \\
500  & \(0.961 \pm 0.008\) & \(0.987 \pm 0.003\) \\
1000 & \(0.981 \pm 0.007\) & \(0.976 \pm 0.010\) \\
2000 & \(0.979 \pm 0.010\) & \(0.965 \pm 0.013\) \\
3000 & \(0.973 \pm 0.011\) & \(0.938 \pm 0.016\) \\
\bottomrule
\end{tabular}
\label{tab:k1_CDR}
\end{table}

\begin{table}[htbp]
\centering
\caption{Mean AUCs of KK and ZK MIAs for \( k = 1 \) on the Milan dataset with varying \( m \).}
\label{tab:k_1_milan}
\begin{tabular}{lcc}
\toprule
\textbf{\( m \)} & \textbf{KK} & \textbf{ZK} \\
\midrule
100  & \(0.756 \pm 0.020\) & \(0.701 \pm 0.046\) \\
500  & \(0.889 \pm 0.018\) & \(0.885 \pm 0.025\) \\
1000 & \(0.916 \pm 0.015\) & \(0.919 \pm 0.016\) \\
2000 & \(0.981 \pm 0.009\) & \(0.972 \pm 0.007\) \\
\bottomrule
\end{tabular}
\label{tab:k1_Milan}
\end{table}

 \begin{figure}[h!]
    \centering
    \includegraphics[width = 0.8\linewidth, page=1]{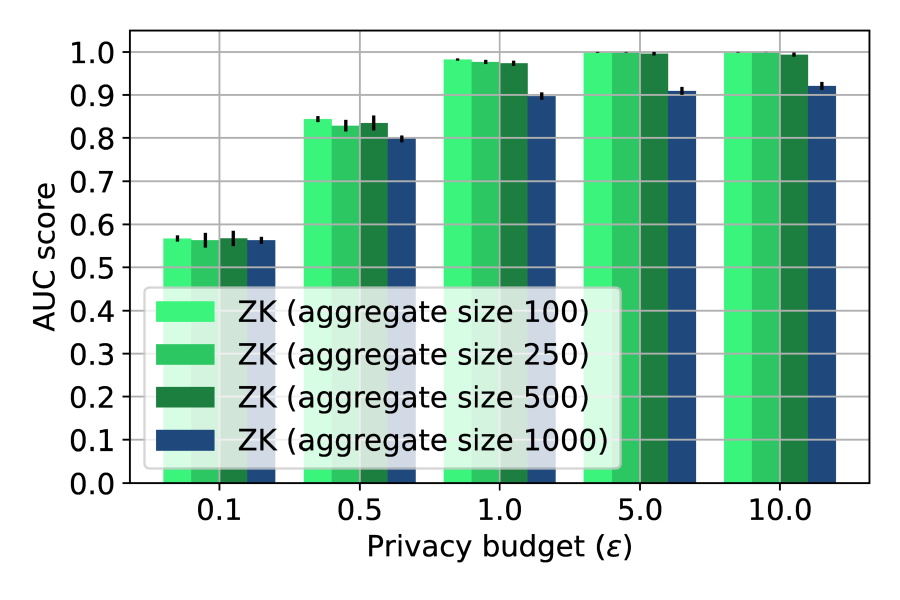}
    \caption{Mean AUC scores with standard error for ZK on event level $\varepsilon$-DP for varying privacy budgets $\varepsilon$ on aggregates of varying sizes from the Milan dataset.}
    \label{fig:variation_aggregate_size}
\end{figure}

\subsection{Increasing the size of ZK synthetic reference}

Figure \ref{fig:n_syn} illustrates how increasing the number of synthetic traces available to the attacker improves the MIA's performance up to marginal returns.

\begin{figure}[h!]
    \centering
    \includegraphics[width = 0.8\linewidth, page=1]{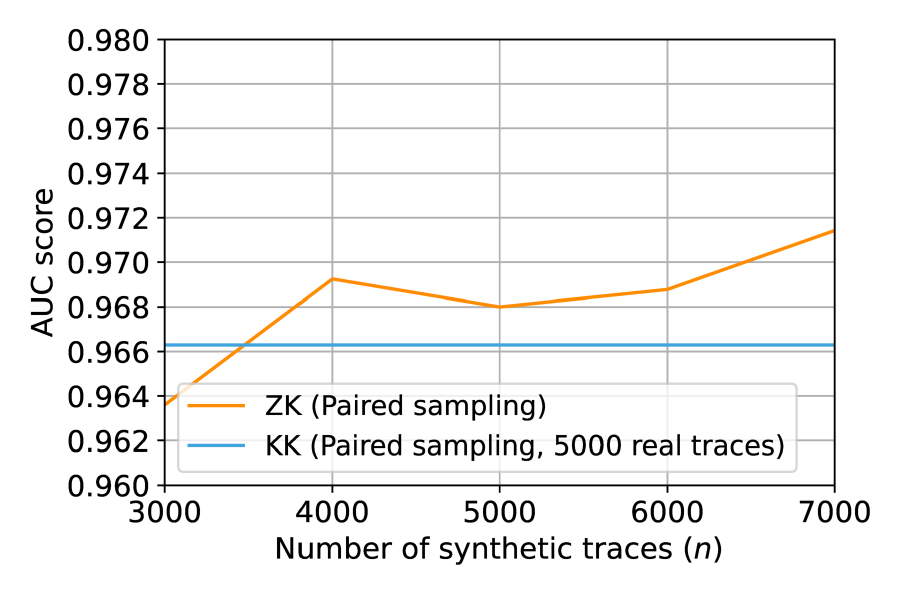}
    \caption{AUC of the Zero Auxiliary Knowledge MIA for different number of synthetic traces generated in the setting $m=1000$, $k=1$, and $\varepsilon=1$-DP at the event level.}
    \label{fig:n_syn}
\end{figure}

\newpage

\subsection{No time information}
\label{appendix:onlyrois}
In this experiment, we now assume that the adversary only has access to some of the locations that the target has visited, without knowing the epochs during which the visits were done. For example, the adversary may know the target's home and work. To model this attack setting, we suppose that the adversary either knows the target's top-$K$ most visited ROIs, for $K=1,2,3$, or the full set of the target's visited ROIs during the observation period. In one implementation, which we call "greedy", the adversary assumes that the target visits each known ROI during every epoch in the observation period. This ensures that the visits to these ROIs are reflected in the target trace, but it also sets many incorrect visits. Results are presented in Table~\ref{tab:greedy_sampling_onlyrois}. In our second implementation, which we call "random sampling", the adversary distributes the target's visits uniformly across the known ROIs. For example, if the adversary knows the top-$3$ ROIs, and their estimate for the mean number of visits per user is $\mu$, then they would sample $\frac{\mu}{3}$ visits for each of the top-$3$ ROIs. The corresponding epochs for each visit are sampled from the estimated time marginal. For simplicity, we assume that $\mu$ is the true mean number of visits and that the estimated time marginal is the true one. Results on raw aggregates of size $m=1000$ are presented in Table~\ref{tab:random_sampling_onlyrois}.

Table~\ref{tab:random_sampling_onlyrois} shows that both MIAs perform poorly ($AUC<0.63$) when the adversary uses random sampling to approximate the target trace. This suggests that random sampling fails to estimate the target trace, due to the omission of true target visits, and the inclusion of incorrect visits. 

In contrast, Table~\ref{tab:greedy_sampling_onlyrois} shows that KK was able to perform significantly better than random when the adversary knew more than $2$ of the target's most visited ROIs and used the greedy implementation (ex. $AUC = 0.86$ on Milan when knowing all visited ROIs). This suggests that, although the greedy implementation includes many incorrect visits, the guaranteed inclusion of some of the target's actual visits enables membership inference to an extent. ZK, on the other hand, fails to attain $0.6$ AUC. Since ZK already replaces real individual traces with synthetic traces, we hypothesize that membership inference becomes too difficult if the estimated target trace contains significantly incorrect information. 

We however note that our current implementation for sampling the visits under this prior knowledge might be suboptimal and that better implementations might exist. For example, \cite{zhang2020locmia} uses a synthetic target trace,  using social network information and the traces of the target's friends. We leave this exploration for future work.

\begin{table*}[htbp]
    \centering
    \begin{tabular}{ccccccccc} \toprule
        \multirow{2.5}{*}{Dataset}&\multicolumn{4}{c}{Knock Knock} & \multicolumn{4}{c}{Zero Auxiliary Knowledge} \\ \cmidrule(r){2-5} \cmidrule(r){6-9}
        &Top 1 & Top 2 & Top 3 & All & Top 1 & Top 2 & Top 3 & All \\ \midrule
        CDR & $0.542 \pm 0.04$ & $0.538 \pm 0.03$ & $0.531 \pm 0.02$ & $0.525 \pm 0.02$ & $0.524 \pm 0.03$ & $0.528  \pm 0.02$ & $0.515 \pm 0.02$ & $0.510 \pm 0.02$ \\
        Milan &  $0.576 \pm 0.02$ & $0.628 \pm 0.03$ & $0.614  \pm 0.03$ & $0.560 \pm 0.03$ & $0.518  \pm 0.02$ & $0.553 \pm 0.02$ & $0.568 \pm 0.02$ & $0.556 \pm 0.04$ \\ \bottomrule
    \end{tabular}
    \caption{Mean AUC scores with standard error for KK and ZK on raw aggregates of size $m=1000$ when the adversary only knows some of the target's visited ROIs and employs the random sampling approach of distributing random visits uniformly across each known ROI.}
    \label{tab:random_sampling_onlyrois}
\end{table*}

\begin{table*}[htbp]
    \centering
    \begin{tabular}{ccccccccc} \toprule
        \multirow{2.5}{*}{Dataset}&\multicolumn{4}{c}{Knock Knock} & \multicolumn{4}{c}{Zero Auxiliary Knowledge} \\ \cmidrule(r){2-5} \cmidrule(r){6-9}
        &Top 1 & Top 2 & Top 3 & All & Top 1 & Top 2 & Top 3 & All \\ \midrule
        CDR & $0.571 \pm 0.05$ & $0.611 \pm 0.05$ & $0.647 \pm 0.06$ & $0.825 \pm 0.05$ & $0.512 \pm 0.02$ & $0.527 \pm 0.01$ & $0.514 \pm 0.02$  & $0.516 \pm 0.01$ \\
        Milan & $0.682 \pm 0.03 $& $0.764 \pm 0.03 $& $0.822 \pm 0.03$ & $0.860 \pm 0.02$ & $0.542 \pm 0.02$ & $0.524 \pm 0.02$ & $0.542 \pm 0.02$ & $0.545 \pm 0.02$ \\ \bottomrule
    \end{tabular}
    \caption{Mean AUC scores with standard error for KK and ZK on raw aggregates of size $m=1000$ when the adversary only knows some of the target's visited ROIs and employs the greedy approach of assuming that the target visits each known ROI during every epoch.}
    \label{tab:greedy_sampling_onlyrois}
\end{table*}

\newpage
\section{Additional Plots}
We present additional figures demonstrating statistics related to the location datasets.

\begin{figure}[h!]
    \centering
    \includegraphics[width = 0.8\linewidth, page=1]{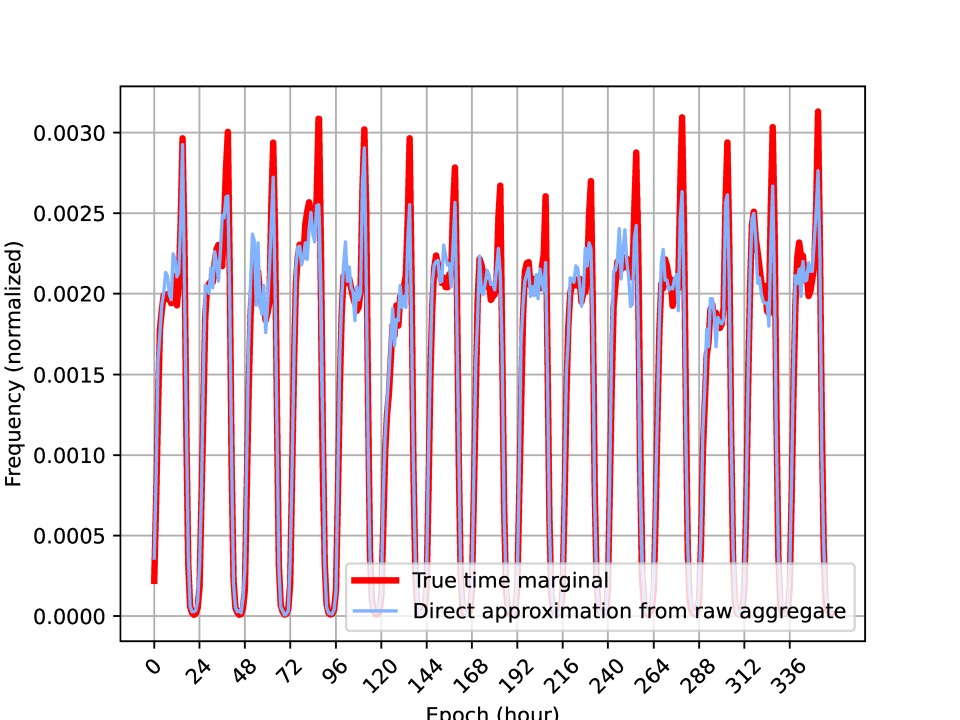}
    \caption{Time marginal from a raw aggregate over $m=1000$ users from the CDR dataset.}
    \label{fig:raw_time_marg}
\end{figure}

\begin{figure}[h!]
    \centering
    \begin{subfigure}[b]{0.49\textwidth}
        \centering
        \includegraphics[width = 0.8\linewidth, page=1]{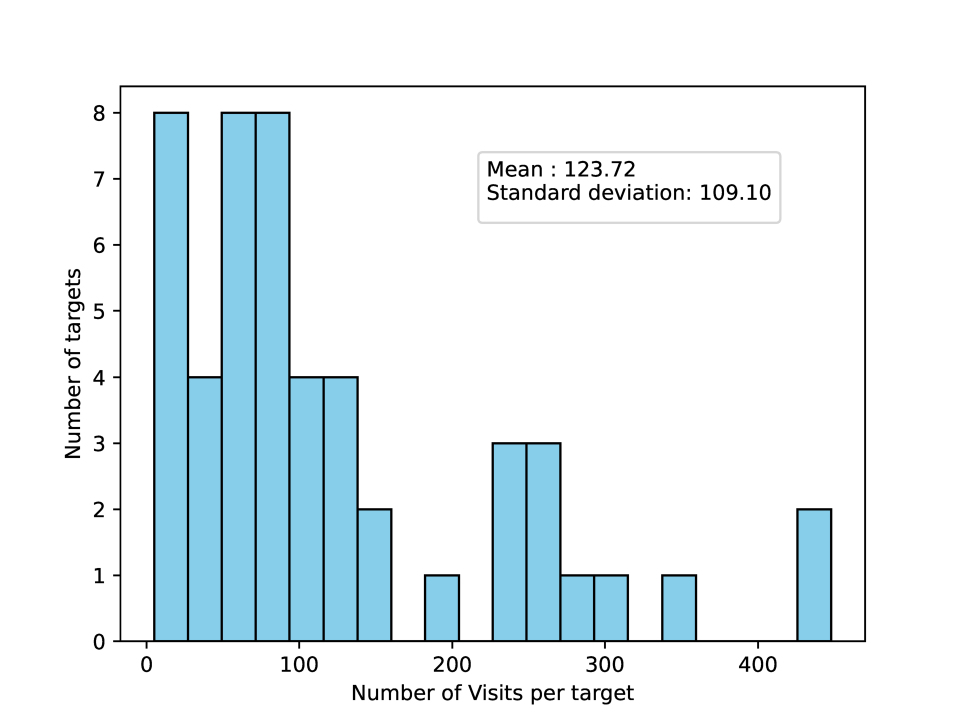}
        \caption{CDR dataset}
        \label{fig:trace_size_CDR}
    \end{subfigure}%
    \hfill 
    \begin{subfigure}[b]{0.49\textwidth}
        \centering
        \includegraphics[width=0.8\textwidth, page=1]{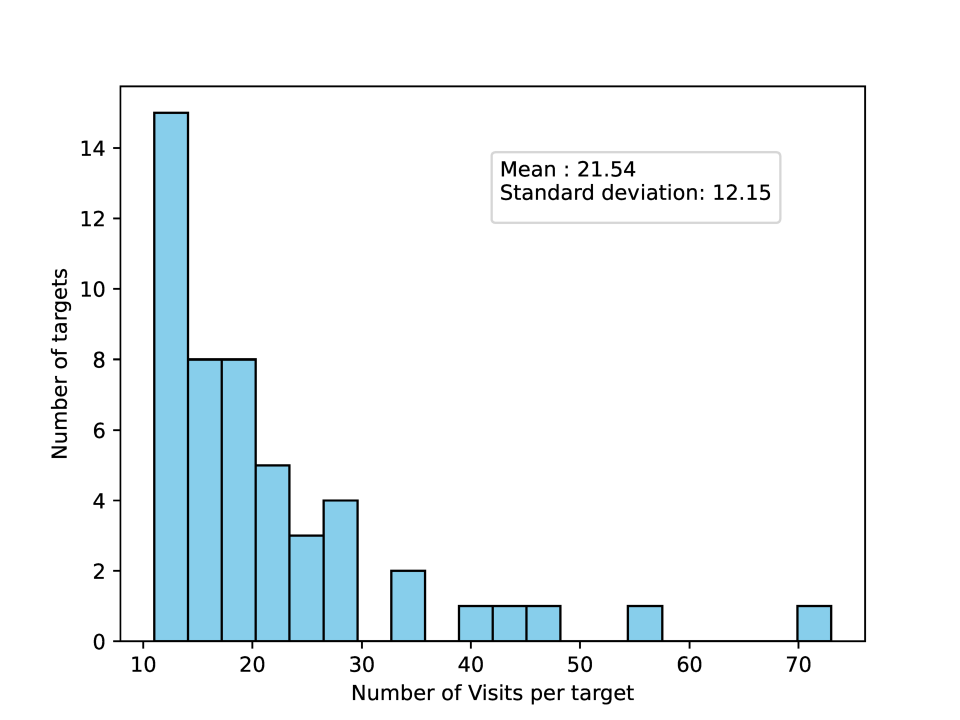}
        \caption{Milan}
        \label{fig:trace_size_Milan}
    \end{subfigure}
    \caption{Number of visits per target over the $50$ targets.}
    \label{fig:target_histogram}
\end{figure}

\begin{figure*}[h!]
    \centering
    \begin{subfigure}[b]{0.49\linewidth}
        \centering
        \includegraphics[width = 0.8\linewidth, page=1]{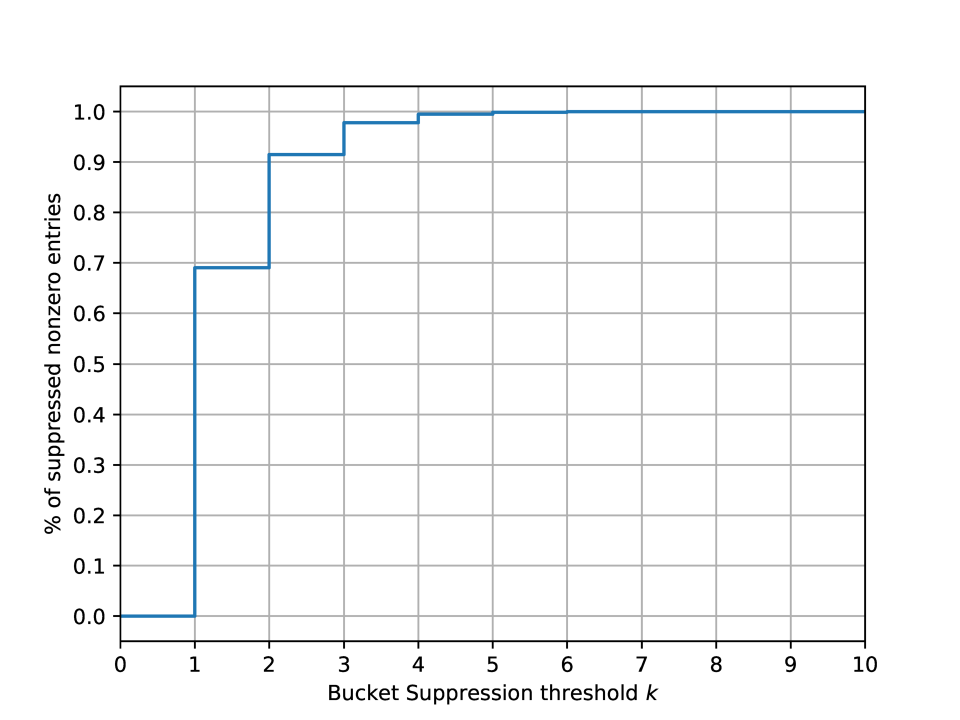}
        \caption{CDR dataset}
        \label{fig:info_loss_CDR}
    \end{subfigure}%
    \hfill 
    \begin{subfigure}[b]{0.49\linewidth}
        \centering
        \includegraphics[width = 0.8\linewidth, page=1]{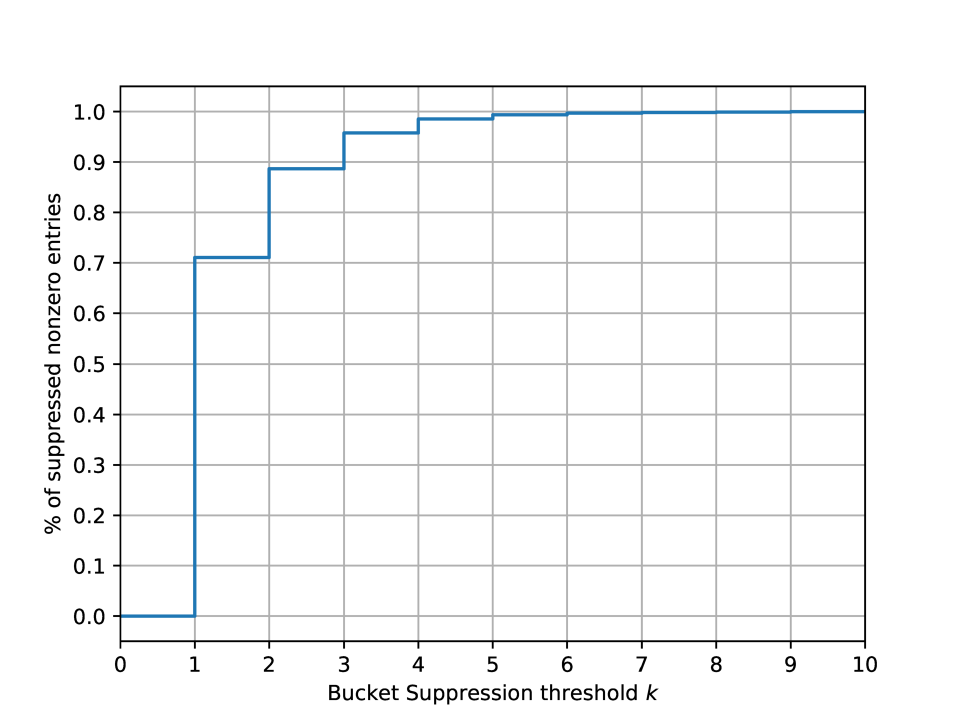}
        \caption{Milan}
        \label{fig:info_loss_Milan}
    \end{subfigure}
    \caption{The percentage of nonzero entries that are suppressed in a size $m=1000$ aggregate after undergoing SSC with threshold $k$ is plotted.}
    \label{fig:bucket_info_loss}
\end{figure*}

\begin{figure*}[ht!]
  \centering
  \begin{subfigure}{0.24\textwidth}
    \includegraphics[width=\linewidth]{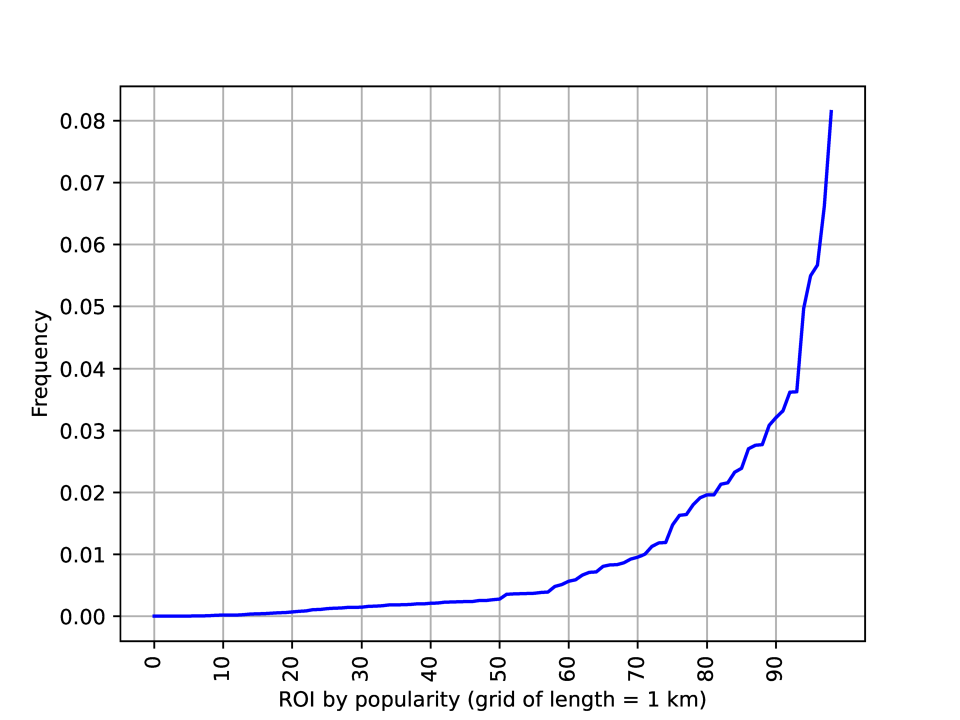}
    \caption{Space Marginal}
    \label{fig:sub-space}
  \end{subfigure}
  \hfill
  \begin{subfigure}{0.24\textwidth}
    \includegraphics[width=\linewidth]{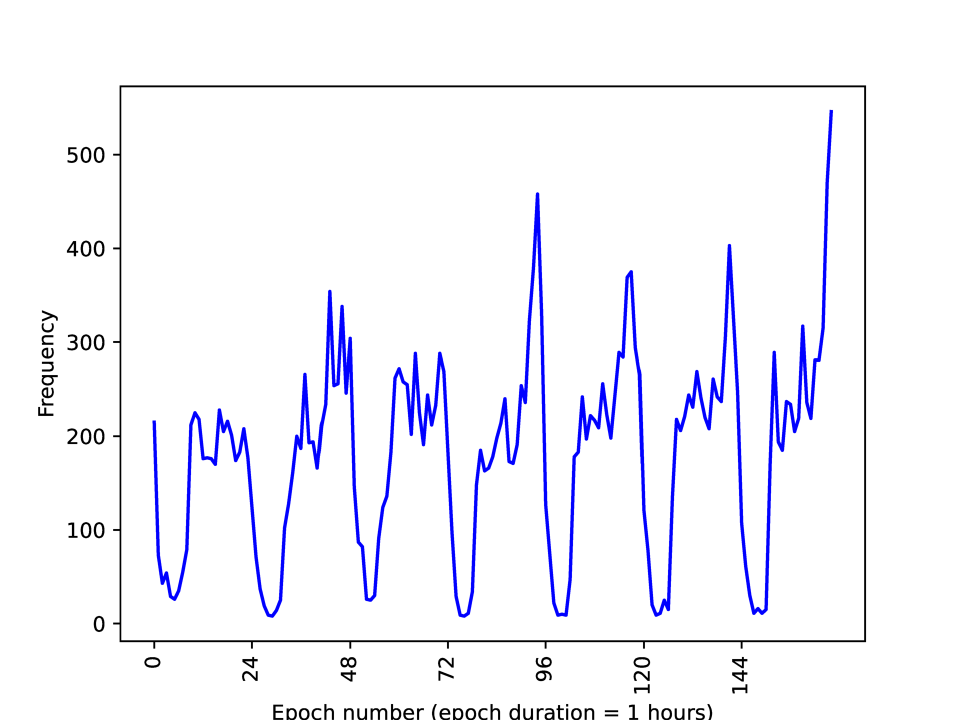}
    \caption{Time Marginal}
    \label{fig:sub-time}
  \end{subfigure}
  \hfill
  \begin{subfigure}{0.24\textwidth}
    \includegraphics[width=\linewidth]{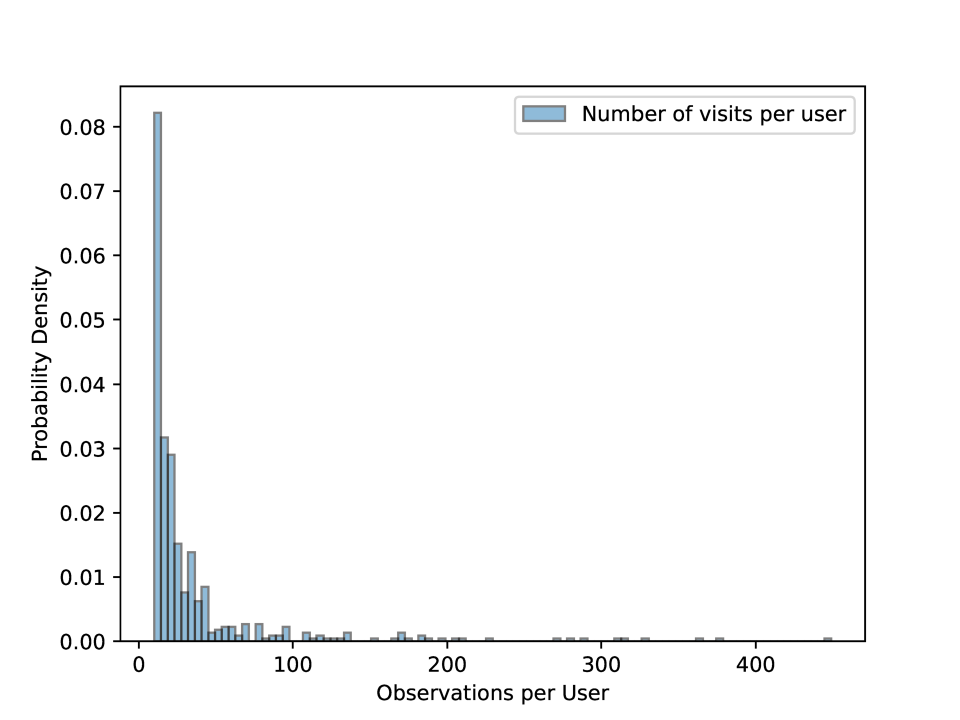}
    \caption{Activity Marginal}
    \label{fig:sub-activity}
  \end{subfigure}
  \hfill
  \begin{subfigure}{0.24\textwidth}
    \centering\includegraphics[width=0.8\linewidth]{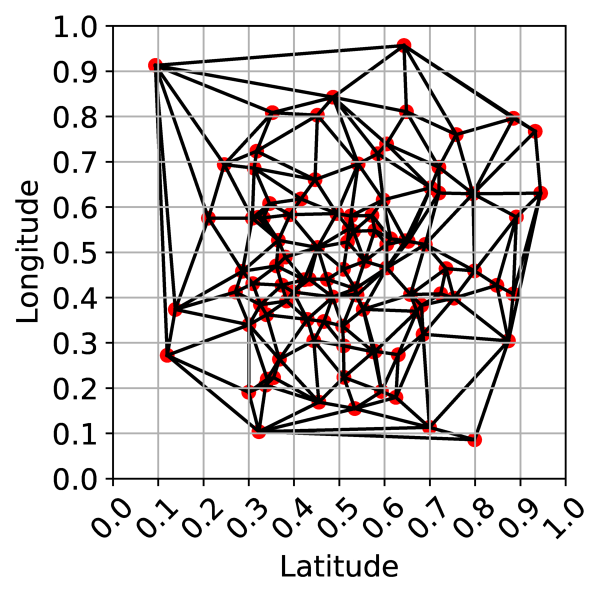}
    \caption{Delaunay Triangulation}
    \label{fig:sub-delaunay}
  \end{subfigure}

  \caption{The four statistical parameters for the unicity model by~\citet{farzanehfar2021risk} include marginal distributions in Figures \ref{fig:sub-space}-\ref{fig:sub-activity} (the dataset's true marginal distributions are shown in red) and the Delaunay triangulation of ROIs in Figure \ref{fig:sub-delaunay}.}
  \label{fig:four_inputs}
\end{figure*}

\begin{figure*}[h!]
    \centering
    \begin{subfigure}[b]{0.75\linewidth}
        \centering
        \includegraphics[width=0.8\textwidth]{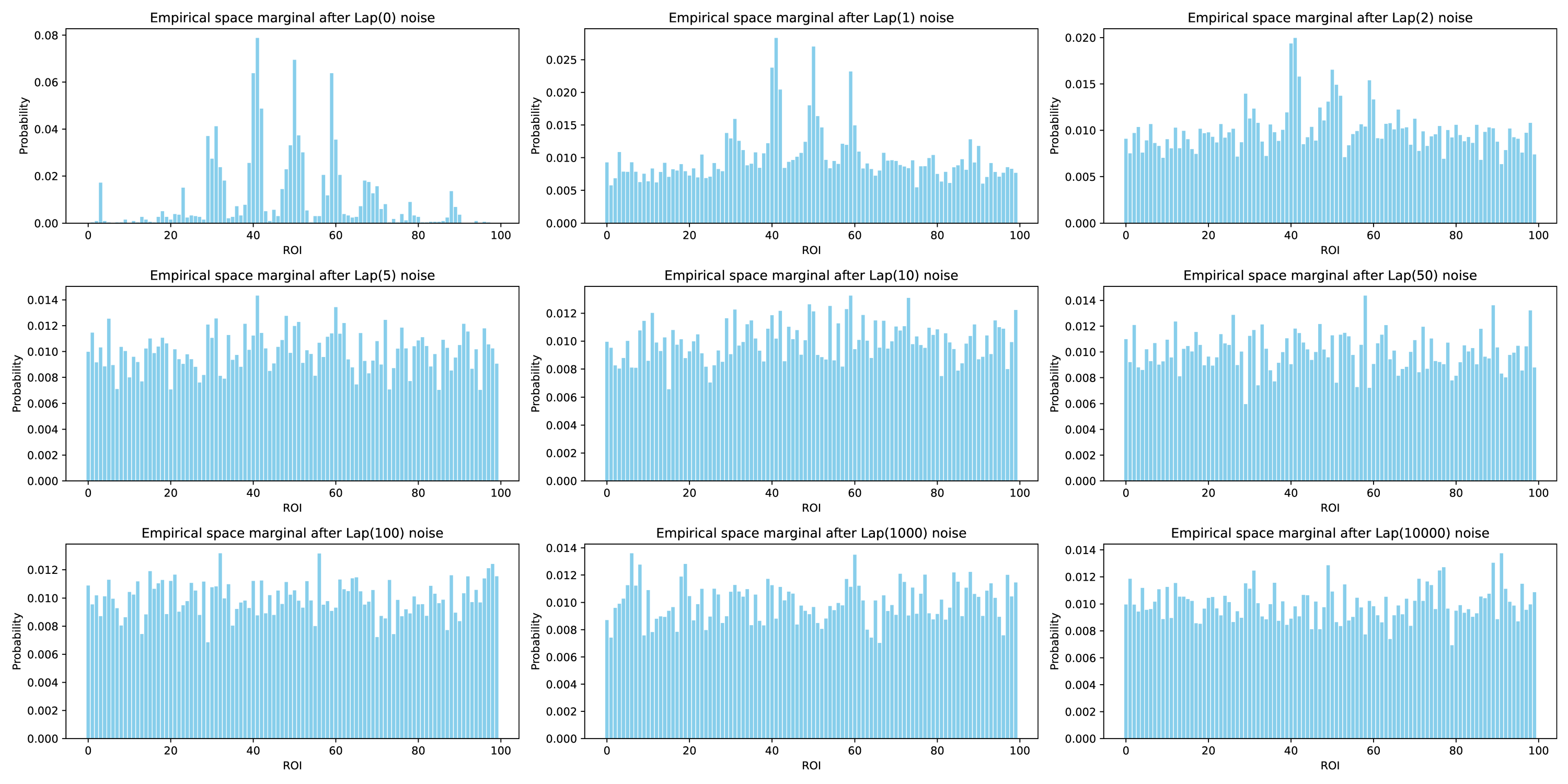}
        \caption{Milan space marginals estimated from $\varepsilon$-DP aggregates }
        \label{fig:Milan_space_laplace}
    \end{subfigure}
    \hfill
    \begin{subfigure}[b]{0.75\linewidth}
        \centering
        \includegraphics[width=0.8\textwidth]{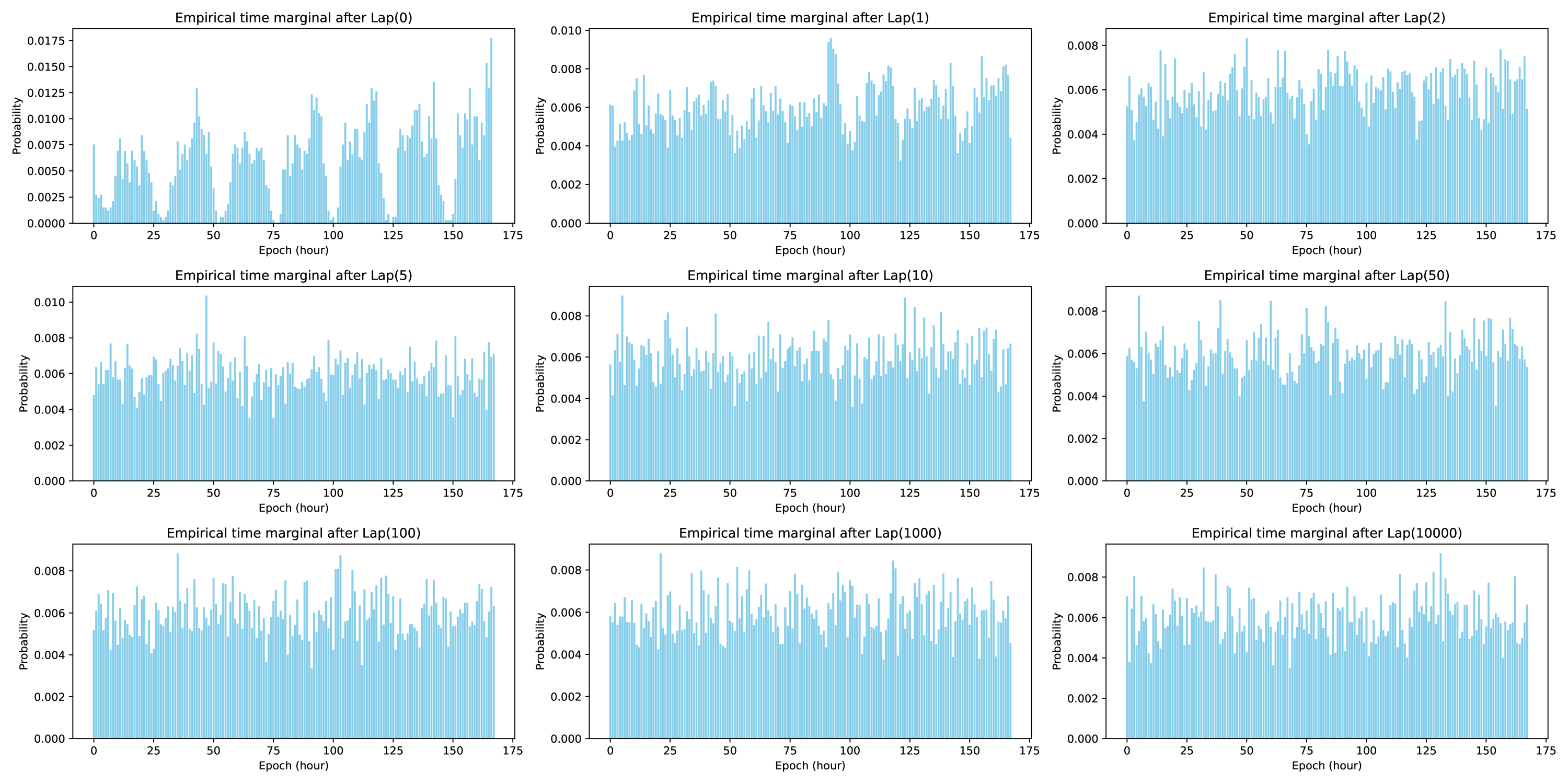}
        \caption{Milan time marginals estimated from $\varepsilon$-DP aggregates}
    \label{fig:Milan_time_laplace}
    \end{subfigure}
    \caption{The space and time marginals directly obtained from $\varepsilon$-DP aggregates over $m=1000$ users from the Milan dataset are plotted for different noise scales $\frac{\Delta}{\varepsilon}$. Interestingly, the distribution does not converge to a uniform distribution as the noise scale increases, due to the increasing variance of $Lap(\frac{\Delta}{\varepsilon})$}
    \label{fig:Milan_laplace}
\end{figure*}

\end{document}